\documentclass[11pt]{article}%
\usepackage{amsmath}
\usepackage{subfigure}
\usepackage{graphicx}
\usepackage{amsfonts}
\usepackage{amssymb}
\usepackage{amsthm}
\usepackage{algorithm}
\usepackage{algorithmic}%
\providecommand{\U}[1]{\protect\rule{.1in}{.1in}}
\newtheorem{theorem}{Theorem}[section]

\newtheorem{conjecture}[theorem]{Conjecture}
\newtheorem{corollary}[theorem]{Corollary}

\newtheorem{definition}[theorem]{Definition}

\newtheorem{lemma}[theorem]{Lemma}
{\theoremstyle{definition}

}

{\theoremstyle{definition}
\newtheorem{remark}[theorem]{Remark}
}

\begin{document}

\title{Searching the Nodes of a Graph: Theory and Algorithms}
\author{Ath. Kehagias, G. Hollinger and A. Gelastopoulos\thanks{emails:
kehagiat@auth.gr, geoff.hollinger@gmail.com, alexgelas@hotmail.com.}}
\date{19 May, 2009}
\maketitle

\begin{abstract}
One or more \emph{searchers} must \emph{capture} an invisible \emph{evader}
hiding in the \emph{nodes} of a graph. We study this version of the
\emph{graph search} problem under additional restrictions, such as
monotonicity and connectedness. We emphasize that we study \emph{node search},
i.e., the capture of a \emph{node-located }evader; this problem has so far
received much less attention than \emph{edge search}, i.e., the capture of an
\emph{edge-located} evader.

We show that in general graphs the problem of node search is easier than that
of edge search, Namely, every edge clearing search is also node clearing, but
the converse does not hold in general (however node search is NP-complete,
just like edge search). Then we concentrate on the \emph{internal monotone
connected }(IMC)\ node search of \emph{trees} and show that it is essentially
equivalent to IMC\ edge search; hence Barriere's tree search algorithm
\cite{Barriere1}, originally designed for edge search, can also be used for
node search.

We return to IMC\ node search on general graphs and present (several variants
of) a new algorithm: \emph{GSST }(Guaranteed Search by Spanning Tree). GSST
clears a graph $\mathbf{G}$ by performing all its clearing moves along a
spanning tree $\mathbf{T}$ of $\mathbf{G}$. Because spanning trees can be
generated and cleared very quickly, GSST\ can test a large number of spanning
trees and find one which clears $\mathbf{G}$ with a small (though not
necessarily minimal)\ number of searchers. We prove the existence of
probabilistically complete variants of GSST (i.e., these variants are
guaranteed to find a minimal IMC node clearing schedule if run for
sufficiently long time). Our experiments also indicate that GSST can
efficiently node-clear large graphs given only a small running time. \emph{An
implementation of GSST\ }(running on Windows and Linux computers) \emph{is
also provided and made publicly available}.

\end{abstract}

\section{Introduction}

\label{sec01}

In this paper we study a variant of \emph{graph search}. Our motivation comes
from applied robotics problems. As an introduction, consider the following two
problems, which can both be modeled as graph searches.

\begin{enumerate}
\item An \emph{evader} is hiding in a \emph{building} and a team of one or
more \emph{searchers} are trying to \emph{capture} him. The building is
represented by a graph $\mathbf{G}$: nodes are \emph{rooms} and edges are
\emph{doors} between rooms. Both searchers and evader occupy the nodes of the
graph and move from node to node by sliding along the edges. The evader is
captured when he is located in the same node as a searcher. It is assumed that
the evader has the following characteristics:\ (a)\ he wants to avoid capture,
(b)\ he is invisible to the searchers (unless he is located in the same node
as one of them), (c)\ he is always aware of the searchers' locations and
(d)\ he is arbitrarily fast.

\item An evader and a team of searchers are located inside a \emph{cave}. The
cave is again represented by a graph $\mathbf{G}$: edges are \emph{tunnels}
and nodes are the \emph{intersections} of tunnels. The searchers move from
node to node, sliding along the edges; capture takes place if a searcher
slides along an edge occupied by the evader, or if the evader moves through a
node occupied by a searcher. The previously mentioned properties of the evader
are assumed to hold in this case as well.
\end{enumerate}

The two problems are very similar but differ in one important respect. In the
first problem the evader is assumed to reside in the \emph{nodes} of the graph
(the rooms); in the second problem he is assumed to reside in the \emph{edges}
(the tunnels). Accordingly, we will call the first problem \emph{node search}
and the second \emph{edge search}. Edge- or node-location of the evader is an
important aspect of the graph search problem and can be used to categorize the
related literature. Historically, research has concentrated on the problem of
\textquotedblleft\emph{edge-located} evader\textquotedblright. To the best of
our knowledge, very little work has been done on the problem of
\textquotedblleft\emph{node-located} evader\textquotedblright. Interestingly,
the term \textquotedblleft node search\textquotedblright\ has been used in the
past to denote a version of the graph search for edge-located evader
\cite{Kiroussis}. Hence, to avoid confusion, let us repeat that \emph{we use
\textquotedblleft node search\textquotedblright\ to denote the search for a
node-located evader}.\emph{ }

Another important aspect of the graph search problem is the \textquotedblleft
visibility\textquotedblright\ of the evader. Namely, we say that the evader is
\emph{visible} if the searchers always know his location; we say that the
evader is \emph{invisible} if the evaders may become aware of his location
only when he is in the same node (or edge) as they.

Additional aspects of the graph search include whether it is \emph{internal}
(the searchers can only move along the edges of the graph) or not (a searcher
can, in one step, move to any node of the graph -- we call this
\textquotedblleft\emph{teleporting}\textquotedblright), \emph{monotone }(the
evader can never return to a searched part of the graph) and \emph{connected}
(the searched parts form a connected subgraph). These terms will be defined
rigorously in Section \ref{sec02}.

We now present a brief overview of the graph search literature. Early papers
include \cite{Breisch} where the problem of cave searching was posed and
\cite{Parsons1} where the first mathematical study of the problem was
presented. The problem was independently rediscovered a little later in
\cite{Petrov}. The edge search problem was placed in the context of graph
theory in \cite{Megiddo}. The already mentioned \textquotedblleft node
search\textquotedblright\ variant \ (actually dealing with edge-located
evader)\ appeared a little earlier in \cite{Kiroussis}. Another variant,
called \textquotedblleft mixed search\textquotedblright\ (and again dealing
with edge-located evader)\ appeared in \cite{Bienstock} and was further
studied in \cite{Takahashi,Yang1}. The study of \emph{connected }edge search
started relatively recently; see \cite{Barriere1, Barriere2, Fomin,
Fraigniaud1, Fraigniaud2, Yang2}. There is very little published on
\textquotedblleft true\textquotedblright\ node search, i.e., the search for a
node-located evader. There is a considerable literature on pursuit of
\emph{visible} node-located evaders, which we will not discuss, since visible
evader search is outside the scope of the current paper; let us simply mention
that an important early paper is \cite{Nowakowski}. Graph search is also
related to several \emph{graph parameters}, most notably \ \emph{pathwidth}
and \emph{vertex separation}. These connections are discussed in many papers,
e.g., in
\cite{BienstockReview,Dendris,Ellis,Kinnersley,Kiroussis,ThilikosParams}.
Finally there are several reviews of the graph search literature:\ an old and
deep one is \cite{BienstockReview}, a more recent one is \cite{Alspach} and a
very recent and very extensive one is \cite{FominThilikos}.

The above mentioned papers adopt a theoretical point of view. While there is
much discussion about graph search algorithms, we have found few actually
implemented algorithms which can tackle relatively large graphs (with the
exception of \emph{tree} search algorithms). A much more applied line of
research appears in the \emph{robotics} literature, for example in
\cite{Gerkey1,Gerkey2,Guibas,Lavalle} and the book \cite{LavalleBook}. These
papers present actual implementations of algorithms, as well as numerical and
even physical experiments, but they provide little (if any)\ theoretical
justification of their methods. In particular, the distinction between edge
and node search often appears to be misunderstood. We attempt to clarify this
distinction and also combine the robotics and graph theoretic points in two
technical reports we have previously published \cite{Hollinger1,Kehagias1}.

In the current paper we study graph search for an invisible, node-located
evader. We also examine connections to the problem of invisible, edge-located
evader. We do \emph{not} deal with the case of visible evader. We are mainly
interested in \emph{internal, monotone, connected }(IMC)\ node search. This
version of the problem is the one most relevant to robotics applications
which, as already stated, provide our main motivation. In general, we try to
provide a balanced combination of theory and implementation.

The main contribution of the paper is the introduction of (several variants
of)\ the node-search algorithm GSST (\emph{Guaranteed Search by Spanning
Tree}). This algorithm is presented in detail, theoretically motivated and
analyzed and also practically evaluated by a number of numerical
experiments\footnote{An implementation of the algorithm (for Windows and Linux
computers)\ is available in the public domain; for details see Appendix
\ref{secD}.}. The \emph{main idea }behind the algorithm is that every
IMC\ node search generates a spanning tree of the graph; conversely, every
spanning tree can be used to perform a node search. Trees can be searched much
more efficiently than general graphs; GSST\ exploits this fact to quickly
compute a large number of searches and then output the one which requires the
least number of searchers\footnote{In fact GSST\ is an \textquotedblleft
anytime\textquotedblright\ algorithm \cite{Zilberstein}. Namely, the longer it
runs the better solutions are provided, but a reasonable solution becomes
available even in the initial stages of the algorithm.}. \ Our experiments
also indicate that GSST can efficiently node-clear large graphs in a small
running time. From the theoretical point of view, we establish several results
about node search and its relation to edge search. These results are used to
motivate the GSST\ algorithm and also to prove that some variants of GSST are
probabilistically complete (i.e., they are guaranteed to find a minimal IMC
node clearing schedule if given sufficient time).

The paper is organized as follows.\ Preliminary concepts and notation are
introduced in Section \ref{sec02}; edge search and node search are compared in
Section \ref{sec03}; IMC\ node and edge search on \emph{trees} is studied in
Section \ref{sec04}; IMC\ node search on general graphs is studied in Section
\ref{sec05} where the GSST\ search algorithm is also introduced; the algorithm
is evaluated by numerical experiments in Section \ref{sec06}; conclusions and
future research directions appear in Section \ref{sec07}. Some edge search
results by Barriere et al. \cite{Barriere1,Barriere2} are presented in
Appendix \ref{secA}; the connection of node search to \emph{mixed} edge search
is discussed in Appendix \ref{secB}; connections to \emph{pathwidth} are
discussed in Appendix \ref{secC}; an implementation of GSST as an actual
executable program (which runs on Windows and Linux computers)\ is discussed
in Appendix \ref{secD}.

\section{Preliminaries}

\label{sec02}

\subsection{Basic concepts and notation}

We denote graphs by boldface letters, e.g. $\mathbf{G}=\left(  V,E\right)  $
where $V$ is the \emph{node set} and $E$ is the \emph{edge set}. We will
always label the nodes of $\mathbf{G}$ as $V=\left\{  1,2,...,N\right\}  $
(hence the graph contains $N$ nodes, i.e. $\left\vert V\right\vert =N$). Edges
are denoted as $\left\{  u,v\right\}  $ where $u,v\in V$; often we will write
$uv$ instead, for simplicity, but note that $vu$ and $uv$ are the same edge
(i.e. we study \emph{undirected} graphs). We only consider connected graphs
without loops or multiple edges. \emph{Also, to avoid some trivial cases, we
will always consider graphs with at least one edge (and at least two nodes).}

Nodes $u,v$ are said to be \emph{neighbors} iff $uv\in E$.

Given a node sequence $u_{1}u_{2}...u_{L}$, where $u_{i}u_{i+1}\in E$ for $i=$
$1,2,...,L-1,$ we say that $u_{1}u_{2}...u_{L}$ is

\begin{enumerate}
\item a \emph{path} iff $u_{i}\neq u_{j}$ for $i,j\in\left\{
1,2,...,L\right\}  $ and $i\neq j$;

\item a \emph{cycle} iff $u_{i}\neq u_{j}$ for $i,j\in\left\{
1,2,...,L\right\}  $ (and $i\neq j$) except that $u_{1}=u_{L}$.
\end{enumerate}

A \emph{tree} is a connected graph without any cycles. Equivalent definitions
are:\ a tree is a graph for which there is a unique path between every pair of
nodes; a tree is a connected graph with $N$ nodes and $N-1$ edges. The
\emph{leaves }of the tree are the nodes which have exactly one neighbor.

A \emph{rooted tree} is a tree with a distinguished node $u_{0}$ called the
\emph{root }of the tree. Given a tree $\mathbf{T=}\left(  V,E\right)  $, we
will denote the same tree rooted at $u_{0}$ by $\mathbf{T}_{u_{0}}=\left(
V,E,u_{0}\right)  $. Take any node $u_{L}\neq u_{0}$ and let $u_{0}%
u_{1}...u_{L-1}u_{L}$ be the unique path from $u_{0}$ to $u_{L}$; then
$u_{L-1}$ is the \emph{parent }of $u_{L}$; $u_{L}$ is the \emph{child }of
$u_{L-1}$. Given a node $x$, its children, their children and so on are the
\emph{descendants} of $x$. Given a rooted tree $\mathbf{T}_{x}=\left(
V,E,x\right)  $ and a node $y$, consider the node set
\[
V\left[  y\right]  =\left\{  z:\text{ }z\text{ is }y\text{ or a descendant of
}y\right\}
\]
and the edge set%
\[
E\left[  y\right]  =\left\{  zu:zu\in E\text{ and }z,u\in V\left[  y\right]
\right\}  ;
\]
note that both $V\left[  y\right]  $ and $E\left[  y\right]  $ are determined
by \emph{both} $y$ and $x$. The graph $\left(  V\left[  y\right]  ,E\left[
y\right]  \right)  $ is a tree (a \emph{sub-tree }of $\mathbf{T}_{x}$). The
\emph{rooted }tree $\left(  V\left[  y\right]  ,E\left[  y\right]  ,y\right)
$ will be denoted by $\mathbf{T}_{x}\left[  y\right]  $.

A \emph{search schedule} on the graph $\mathbf{G}=\left(  V,E\right)  $ is a
sequence of \emph{ordered }pairs of nodes:
\[
\mathbf{S}=(\left(  u\left(  1\right)  ,v\left(  1\right)  \right)  ,\left(
u\left(  2\right)  ,v\left(  2\right)  \right)  ,...,\left(  u\left(
t_{fin}\right)  ,v\left(  t_{fin}\right)  \right)  ),
\]
subject to $\left\{  u\left(  t\right)  ,v\left(  t\right)  \right\}  \in E$
for $t=1,2,...,t_{fin}$. We call $\mathbf{S}\left(  t\right)  =\left(
u\left(  t\right)  ,v\left(  t\right)  \right)  $ the $t$-th move (or
step)\ of the search schedule. The $t$-th move can be any one of the following

\begin{enumerate}
\item \emph{placing }a searcher at node $v$ (in which case $u\left(  t\right)
=0$ and $v\left(  t\right)  \in V$) or

\item \emph{sliding} a searcher from node $u\left(  t\right)  \in V$ to node
$v\left(  t\right)  \in V$ or

\item \emph{removing }a searcher from node $u$ (in which case $u\left(
t\right)  \in V$ and $v\left(  t\right)  =0$).
\end{enumerate}

We will also use the more evocative notation $\mathbf{S}\left(  t\right)
=\left(  u\rightarrow v\right)  $, i.e. a searcher is moved from node $u$ to
node $v$. Moves of the form: $0\rightarrow v$ (a new searcher is introduced in
the graph)\ and $u\rightarrow0$ (a searcher is removed from the graph) involve
the fictitious \textquotedblleft source\ node\textquotedblright\ 0 (it is not
an element of $V$), in which searchers are kept whenever they are not actively
involved in the graph search; the evader does not have access to the source node.

Given a search $\mathbf{S}$, the number of searchers inside the graph at time
$t$ will be denoted by $sn\left(  \mathbf{S},t\right)  $. The maximum number
of searchers used by $\mathbf{S}$ will be denoted by $\overline{sn}\left(
\mathbf{S}\right)  $, i.e.
\[
\overline{sn}\left(  \mathbf{S}\right)  =\max_{t}sn\left(  \mathbf{S,}%
t\right)  .
\]
These numbers must not be confused with the search number of a \emph{graph},
which will be defined in Section \ref{sec0202}.

Finally, we say that a node $u$ is \emph{guarded} at time $t$ iff a searcher
is located at $u$ (at time $t$); otherwise we say $u$ is \emph{unguarded}. \ A
path $u_{1}u_{2}...u_{L}$ is called \emph{n-unguarded }(node unguarded) iff
nodes $u_{i}$ ($i=1,2,...,L$) are unguarded (at time $t$); otherwise it is
called \emph{n-guarded.} The path is called \emph{e-unguarded }(edge
unguarded) iff nodes $u_{i}$ ($i=2,3,...,L-1$) are unguarded (at time $t$) and
e-guarded otherwise. It is easy to see that%
\begin{align*}
u_{1}u_{2}...u_{L}\text{ is e-unguarded}  &  \Rightarrow u_{2}u_{3}%
...u_{L-1}\text{ is n-unguarded,}\\
u_{1}u_{2}...u_{L}\text{ is n-unguarded}  &  \Rightarrow u_{1}u_{2}%
...u_{L}\text{ is e-unguarded.}%
\end{align*}
The reason for the two different definitions is that the first pertains to
\emph{node} recontamination and the second to \emph{edge }recontamination, as
will be seen in Section \ref{sec0202}.

\subsection{Node and edge search}

\label{sec0202}

We repeat the assumptions introduced in Section \ref{sec01} regarding the
evader. Namely, the evader wants to avoid capture, he is invisible to the
searchers (unless located in the same node), he is always aware of the
searchers' locations and arbitrarily fast. The net result of all these
assumptions is that (in both node and edge search)\ the evader can (and will)
always avoid capture if an escape route is available. Hence, from the
searchers' point of view, we can think of graph search as a process of
eliminating escape routes. This is expressed as follows:\ 

\begin{enumerate}
\item a node / edge is \emph{dirty} if it can \emph{possibly} contain the
evader and \emph{clear} otherwise (e.g., a node occupied by a searcher is clear);

\item a previously clear node / edge can become dirty (e.g., when a previously
guarded path between a clear and a dirty node becomes unguarded) -- this is
called \emph{recontamination};

\item graph search is the process of gradually decreasing the \emph{dirty set}
(of nodes or edges) until it becomes the empty set (i.e., the evader has no
escape route left).
\end{enumerate}

This is a \emph{worst case approach }which essentially eliminates the evader
from the graph search, introducing in his place the dirty set. We can think of
node (edge) search as a \emph{one-player} \textquotedblleft node
game\textquotedblright\ (\textquotedblleft edge game\textquotedblright). In
both these games the player controls all searchers, much like pieces in a game
of chess. We now present the rules of the two games. For reasons which will be
explained presently, we substitute the terms \textquotedblleft
clear\textquotedblright\ and \textquotedblleft dirty\textquotedblright\ with
the terms \textquotedblleft n-clear\textquotedblright\ and \textquotedblleft
n-dirty\textquotedblright\ (in the node game) and \textquotedblleft
e-clear\textquotedblright\ and \textquotedblleft e-dirty\textquotedblright%
\ (in the edge game).

\begin{center}
--------------------------------------------------------------------------------------------------

\textbf{Rules of the Node Game}
\end{center}

\begin{enumerate}
\item[\textbf{{N0}}] At time $t=0$ all nodes are \emph{n-dirty} and no
searcher is located in the graph.

\item[\textbf{N1}] At every time $t=1,2,...$ the player performs \emph{one} of
the following moves:

\begin{enumerate}
\item[\textbf{{N1a}}] place a searcher on a node,

\item[\textbf{{N1b}}] remove a searcher from a node,

\item[\textbf{{N1c}}] slide a searcher along one edge.
\end{enumerate}

\item[\textbf{{N2}}] An n-dirty node becomes \emph{n-clear }when occupied by a searcher.

\item[\textbf{{N3}}] An n-clear node $u$ becomes n-dirty when it is connected
to to an n-dirty node $v$ by an \emph{n-unguarded\ path}.

\item[\textbf{{N4}}] An edge is n-dirty if it is adjacent to an n-dirty node;
otherwise it is clear.

\item[\textbf{{N5}}] The game is concluded when all nodes (and consequently
also all edges) are n-clear.
\end{enumerate}

\begin{center}
--------------------------------------------------------------------------------------------------

\textbf{Rules of the Edge Game}
\end{center}

\begin{enumerate}
\item[\textbf{{E0}}] At time $t=0$ all edges are \emph{e-dirty} and no
searcher is located in the graph.

\item[\textbf{E1}] At every time $t=1,2,...$ the player performs \emph{one} of
the following moves:

\begin{enumerate}
\item[\textbf{{E1a}}] place a searcher on a node,

\item[\textbf{{E1b}}] remove a searcher from a node,

\item[\textbf{{E1c}}] slide a searcher along one edge.
\end{enumerate}

\item[\textbf{{E2}}] An e-dirty edge becomes e-clear when a searcher slides
along the edge.

\item[\textbf{{E3}}] An e-clear edge $uv$ becomes e-dirty when it is connected
to to a e-dirty edge $xy$ by an \emph{e-unguarded\ path}.

\item[\textbf{{E4}}] A node is e-dirty if it is unguarded and adjacent to an
e-dirty edge; otherwise it is e-clear.

\item[\textbf{{E5}}] The game is concluded when all edges (and consequently
also all nodes)\ are e-clear.
\end{enumerate}

\begin{center}
--------------------------------------------------------------------------------------------------
\end{center}

\begin{remark}
\label{prp0201}The reason for using \textquotedblleft
n-clear\textquotedblright\ and \textquotedblleft e-clear\textquotedblright%
\ (instead of simply \textquotedblleft clear\textquotedblright)\ is that an
edge can be clear in the node game and dirty in the edge game. An example will
illustrate this point. Consider the graph of Fig.1 and the search schedule
$0\rightarrow1,$ $0\rightarrow1,$ $1\rightarrow2,$ $2\rightarrow4,$
$4\rightarrow3.$ If this schedule is executed in a node game, after the final
move all nodes are n-clear and so all edges are also n-clear. But in an edge
game, after the final move, edge $\left\{  1,3\right\}  $ has still not been
traversed and hence it is still e-dirty. More generally, while
\textquotedblleft n-clear\textquotedblright\ and \textquotedblleft
e-clear\textquotedblright\ refer to similar physical situations, there is no a
priori reason that their mathematical definitions are equivalent (in Section
\ref{sec04} we will show that they \emph{are }equivalent \emph{provided
certain conditions are satisfied}).
\end{remark}

\begin{figure}[h]
\centering
\includegraphics[width=2.25in]{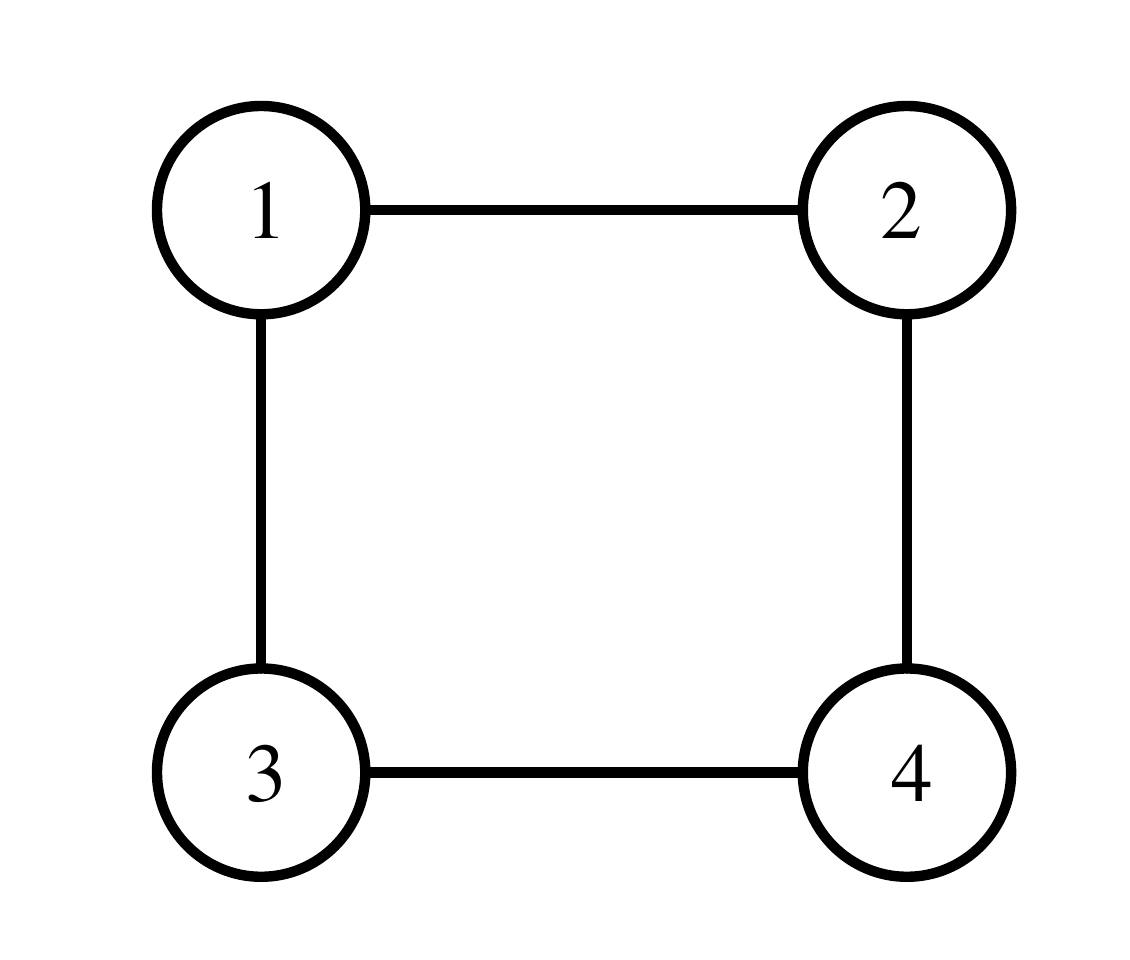}\caption{A graph in which node
clearing and edge clearing are not equivalent. For example, placing a searcher
into node 0 ($0\rightarrow1$), then placing a second searcher at 1 and sending
him through 2, 4, 3 ($0\rightarrow1$, $1\rightarrow2$, $2\rightarrow4$,
$4\rightarrow3$) is node clearing but not edge clearing.}%
\label{fig:graph1}%
\end{figure}

\begin{remark}
\label{prp0202}Note that in the edge game clear edges are \textquotedblleft
recorded\textquotedblright\ at the completion of each move. This is important
to note in the case that an edge is e-cleared and e-dirtied \emph{during the
same move }(i.e. for a single value of $t$). To clarify this, suppose that in
the graph of Fig.1 the following moves are executed:\ $0\rightarrow1$ at
$t=1$; $1\rightarrow2$ at $t=2$. In this case edge $\left\{  1,2\right\}  $ is
first e-cleared (because traversed)\ and then e-dirtied, \emph{both during
}$t=2$. After completion of the second move, \emph{edge }$\left\{
1,2\right\}  $ \emph{is e-dirty}.
\end{remark}

We now introduce more detailed notation regarding clear/dirty nodes/edges. For
the node game we will use the following notation.%
\begin{align*}
V_{N}^{C}\left(  t\right)   &  =\left\{  u\in V:\text{at time }t\text{,
}u\text{ is n-clear }\right\}  ;\\
V_{N}^{D}\left(  t\right)   &  =\left\{  u\in V:\text{at time }t\text{,
}u\text{ is n-dirty}\right\}  ;\\
E_{N}^{C}\left(  t\right)   &  =\left\{  uv\in E:\text{ at time }t\text{,
}uv\text{ is n-clear}\right\}  ;\\
E_{N}^{D}\left(  t\right)   &  =\left\{  uv\in E:\text{ at time }t\text{,
}uv\text{ is n-dirty}\right\}  ;\\
\mathbf{G}_{N}^{C}\left(  t\right)   &  =\left(  V_{N}^{C}\left(  t\right)
,E_{N}^{C}\left(  t\right)  \right)  \text{ is the }\emph{n-clear\ graph}%
\text{ at time }t.\
\end{align*}

We use exactly analogous notation for the edge game, but with an $E$ (rather
than $N$) subscript:$\ V_{E}^{C}\left(  t\right)  $, $V_{E}^{D}\left(
t\right)  $, $E_{E}^{C}\left(  t\right)  $, $E_{E}^{D}\left(  t\right)  $,
$\mathbf{G}_{E}^{C}\left(  t\right)  =\left(  V_{E}^{C}\left(  t\right)
,E_{E}^{C}\left(  t\right)  \right)  $. Obviously we have
\begin{align*}
V_{N}^{C}\left(  t\right)  \cup V_{N}^{D}\left(  t\right)   &  =V,\quad
V_{N}^{C}\left(  t\right)  \cap V_{N}^{D}\left(  t\right)  =\emptyset,\quad
E_{N}^{C}\left(  t\right)  \cup E_{N}^{D}\left(  t\right)  =E,\quad E_{N}%
^{C}\left(  t\right)  \cap E_{N}^{D}\left(  t\right)  =\emptyset,\\
V_{E}^{C}\left(  t\right)  \cup V_{E}^{D}\left(  t\right)   &  =V,\quad
V_{E}^{C}\left(  t\right)  \cap V_{E}^{D}\left(  t\right)  =\emptyset,\quad
E_{E}^{C}\left(  t\right)  \cup E_{E}^{D}\left(  t\right)  =E,\quad E_{E}%
^{C}\left(  t\right)  \cap E_{E}^{D}\left(  t\right)  =\emptyset.
\end{align*}

\begin{remark}
\label{prp0203}The n-clear / e-clear graph definitions are consistent, i.e.
$\left(  V_{N}^{C}\left(  t\right)  ,E_{N}^{C}\left(  t\right)  \right)  $ /
$\left(  V_{E}^{C}\left(  t\right)  ,E_{E}^{C}\left(  t\right)  \right)  $
\emph{are} graphs. More specifically: \ 
\end{remark}

\begin{enumerate}
\item in the node game $\mathbf{G}_{N}^{C}\left(  t\right)  $ consists of the
n-clear nodes and all edges between such nodes (these are exactly the set
$E_{N}^{C}\left(  t\right)  $);

\item in the edge game $\mathbf{G}_{E}^{C}\left(  t\right)  $ consists of the
e-clear edges and none of these can be left \textquotedblleft
dangling\textquotedblright, i.e. if $uv\in E_{E}^{C}\left(  t\right)  $ then
$u$,$v\in V_{E}^{C}\left(  t\right)  $ (we will prove this as Lemma
\ref{prp0301} in Section \ref{sec03}); in addition $V_{E}^{C}\left(  t\right)
$ may contain guarded nodes adjacent to e-dirty edges; such nodes may appear
as disconnected components of $\mathbf{G}_{E}^{C}\left(  t\right)  $.
\end{enumerate}

In the current paper we will concentrate on some \emph{restricted} versions of
the edge and node games. More specifically, the restrictions concern the
allowable search schedules and they are the following.

\begin{enumerate}
\item We call a search schedule \emph{rooted }iff searchers can be placed only
into a single, prespecified node $u_{0}$, called the \emph{root} of the search.

\item We call a search schedule \emph{internal }iff searchers once placed into
the graph (a) can only move along the edges and (b)\ are never removed from
the graph. In other words, \textquotedblleft teleporting\textquotedblright%
\ (the arbitrary movement of a searcher from one node to another, irrespective
of the graph connectivity) is not allowed\footnote{Since no searchers are ever
removed during an internal search $\mathbf{S}$, we have $\overline{sn}\left(
\mathbf{S}\right)  =sn\left(  \mathbf{S,}t_{fin}\right)  $ where $t_{fin}$ is
the length of the search.}.

\item In the node game we call a search schedule $\mathbf{S}$ \emph{monotone
}iff the clear node set is never decreasing: $V_{N}^{C}\left(  t\right)
\subseteq V_{N}^{C}\left(  t+1\right)  $ for all $t.$ In other words, once a
node is n-cleared it never becomes n-dirty again. The definition is similar
for the edge game; i.e., $\mathbf{S}$ is monotone in the edge game iff
$E_{E}^{C}\left(  t\right)  \subseteq E_{E}^{C}\left(  t+1\right)  $ for all
$t.$\footnote{An important detail is worth mentioning at this point. There is
a property, call it \emph{strong monotonicity}, which characterizes schedules
in which a \emph{traversed} edge never becomes e-dirty again. The schedule of
Remark \ref{prp0202}. is monotone (since $E_{E}^{C}\left(  1\right)
=E_{E}^{C}\left(  2\right)  =\emptyset$) but \emph{not }strongly monotone. We
stress this detail mostly for the sake of completness; in what follows we will
\emph{not} use strong monotonicity.} \ 

\item In the node game (resp. edge game)\ we call a search schedule
\emph{connected }iff the clear graph $\mathbf{G}_{N}^{C}\left(  t\right)  $
(resp. $\mathbf{G}_{E}^{C}\left(  t\right)  $)\ is connected for $t=1,2,...$.
\end{enumerate}

In Section \ref{sec03} we will study general, unrestricted graph search; in
the remaining sections we will concentrate on internal monotone (\emph{IM}),
internal connected (\emph{IC})\ and, especially, internal monotone connected
(\emph{IMC})\ node / edge search. Our focus originates in our interest in
practical \emph{pursuit / evasion }problems which arise in robotics. The
nature of the problem dictates the restrictions on the search.

\begin{enumerate}
\item The internality constraint arises from the fact that a robot cannot
\textquotedblleft teleport\textquotedblright; it can only move through rooms
and doors connecting these rooms; in the graph model of the physical situation
this is reflected by the use of edge-sliding moves only.

\item Similarly, rootedness reflects the fact that usualy an environment has a
single entrance through which robots must enter.

\item Connectedness and monotonicity are desirable (but not mandatory)
characteristics of robotic search. For example, in a hostile environment it is
preferrable that the cleared area consists of a single connected region, which
is easier to guard and control than multiple regions. Similarly, if a high
cost is associated with clearing an area, it is desirable that this area is
cleared only once, i.e., that the search is monotone.
\end{enumerate}

A given graph $\mathbf{G}$ with $N$ nodes can always be node-cleared using $N$
searchers:\ just place one searcher in every node. Usually $\mathbf{G}$ can be
cleared with much $^{{}}$fewer searchers; the \emph{node search number} of
$\mathbf{G}$, denoted by $s_{N}\left(  \mathbf{G}\right)  $, is the minimum
number of searchers required to node-clear $\mathbf{G}$. Introducing
additional constraints on the search schedule we get additional search numbers:

\begin{center}%
\begin{tabular}
[c]{lll}%
$s_{N}^{i}\left(  \mathbf{G}\right)  $ : & min. nr. of searchers required to
clear $\mathbf{G}$ with & internal\ node search schedule;\\
$s_{N}^{im}\left(  \mathbf{G}\right)  :$ & \qquad\qquad\qquad-//- & IM\ node
search schedule;\\
$s_{N}^{ic}\left(  \mathbf{G}\right)  $ : & \qquad\qquad\qquad-//- & IC\ node
search schedule;\\
$s_{N}^{imc}\left(  \mathbf{G}\right)  :$ & \qquad\qquad\qquad-//- & IMC\ node
search schedule.
\end{tabular}

\end{center}

For the edge game we define the corresponding edge search numbers:\ $s_{E}%
\left(  \mathbf{G}\right)  $, $s_{E}^{i}\left(  \mathbf{G}\right)  $,
$s_{E}^{im}\left(  \mathbf{G}\right)  $, $s_{E}^{ic}\left(  \mathbf{G}\right)
$, $s_{E}^{imc}\left(  \mathbf{G}\right)  $.

In case of a search rooted at node $x$, we will modify the notation as
follows: $s_{N}\left(  \mathbf{G};x\right)  $ is the minimum number of
searchers required to node-clear $\mathbf{G}$ by a search rooted at $x$;
similarly for $s_{E}\left(  \mathbf{G;}x\right)  $, $s_{E}^{i}\left(
\mathbf{G;}x\right)  $ and so on. If the graph $\mathbf{G}$ is a rooted tree
$\mathbf{T}_{x}$, then we will assume it is always searched by rooted
searches; hence $s_{E}\left(  \mathbf{T}_{x}\right)  =s_{E}\left(
\mathbf{T;}x\right)  $, $s_{E}^{i}\left(  \mathbf{T}_{x}\right)  =s_{E}%
^{i}\left(  \mathbf{T;}x\right)  $ and so on.

A \emph{minimal} node search schedule is one which clears $\mathbf{G}$ using
$s_{N}\left(  \mathbf{G}\right)  $ searchers; a \emph{minimal} IM\ node search
schedule is an IM\ schedule which clears $\mathbf{G}$ using $s_{N}^{im}\left(
\mathbf{G}\right)  $ searchers and similarly for all other types of either
node or edge search.

\section{Comparison of node and edge search}

\label{sec03}

In this section we compare \textquotedblleft unrestricted\textquotedblright%
\ edge and node search. In other words, we compare the outcome of a search
schedule $\mathbf{S}$ when it is used first in a node game and then in edge
game, both played on the same graph $\mathbf{G}$ and without requiring
rootedness, internality, monotonicity or connectedness.

Our main result is Theorem \ref{prp0303} which, informally, says the
following: every search schedule $\mathbf{S}$ gives at least as good results
in the node game as in the edge game. More precisely, at every step $t$ of the
search, the set of e-clear nodes (and edges) is a subset of the set of n-clear
nodes (and edges). An immediate corollary is that an edge clearing search
schedule is also node clearing; The converse does not hold, as we have already
seen by the example of Remark \ref{prp0201}.

To prove Theorem \ref{prp0303} we will need the following lemmas.

\begin{lemma}
\label{prp0301}Given a graph $\mathbf{G}$ and a search schedule $\mathbf{S}$,
for $t=0,1,2,...$ we have%
\begin{align}
u  &  \in V_{E}^{D}\left(  t\right)  \Rightarrow\left(  \forall ux\in E:ux\in
E_{E}^{D}\left(  t\right)  \right) \label{eq0301a}\\
uv  &  \in E_{E}^{C}\left(  t\right)  \Rightarrow\left(  u,v\in V_{E}%
^{C}\left(  t\right)  \right)  \label{eq0301b}%
\end{align}
In other words: if a node $u$ is e-dirty, then all edges $ux$, adjacent to
$u$, are also e-dirty; and if an edge $uv$ is e-clear, then nodes $u,v$ are e-clear.
\end{lemma}

\begin{proof}
For (\ref{eq0301a}): since $u$ is e-dirty, it must be unguarded. We
distinguish two cases.

\noindent(1.i)\ Suppose $u$ is adjacent to a single edge, call it $uv$. If
$uv$ is e-clear then $u$ is not adjacent to any e-dirty edges and hence must
be e-clear, which is contrary to the original assumption. Hence $uv$ is e-dirty.

\noindent(1.ii)\ On the other hand, if $u$ is adjacent to more than one edges,
they cannot all be e-clear (then $u$ would also be e-clear); so at least one
edge, call it $uw$, is e-dirty. Suppose there is also an e-clear edge, call it
$uv$. But then the e-unguarded path $wuv$ connects e-dirty $wu$ to e-clear
$uv$, which leads to contradiction. Hence all edges adjacent to $u$ are e-dirty.

For (\ref{eq0301b}): suppose $u\in V_{E}^{D}\left(  t\right)  $. Then $u$ is
unguarded and there exists some $ux\in E_{E}^{D}\left(  t\right)  $ and $x\neq
v$. But the path $xuv$ is e-unguarded and so $uv\in E_{E}^{D}\left(  t\right)
$ which is a contradiction.
\end{proof}

\begin{lemma}
\label{prp0302}If a node $u$ is e-cleared exactly at time $t$, then $u$ is
entered by a searcher at $t$.
\end{lemma}

\begin{proof}
Suppose $u\in$ $V_{E}^{D}\left(  t-1\right)  \cap V_{E}^{C}\left(  t\right)
$. This can happen in two ways. (i)\ $u$ was unguarded and e-dirty at $t-1$
and then entered by a searcher at $t$. This is exactly the case described by
the theorem. (ii)\ The other way for $u$ to be cleared exactly at $t$ is if
$u$ was unguarded and e-dirty at $t-1$ and in addition a single edge, call it
$ux_{0}$, was e-dirty at $t-1$ and e-cleared at $t$ (while all other edges
$ux$ were e-clear at $t-1$ and remained e-clear at $t$). For this to happen,
the move at time $t$ is either $u\rightarrow x_{0}$ or $x_{0}\rightarrow u$.
But $u\rightarrow x_{0}$ is not possible, because then $u$ would be guarded
and e-clear at $t-1$. Hence the move at $t$ is $x_{0}\rightarrow u$ and $u$
was entered at $t$ (which shows that (ii)\ is a sub-case of (i)).
\end{proof}

Now we are ready to prove the following.

\begin{theorem}
\label{prp0303}Given a graph $\mathbf{G}$ and a search schedule $\mathbf{S}$,
we have
\begin{equation}
\text{for }t=0,1,2,...:\qquad V_{E}^{C}\left(  t\right)  \subseteq V_{N}%
^{C}\left(  t\right)  \quad\text{and}\quad E_{E}^{C}\left(  t\right)
\subseteq E_{N}^{C}\left(  t\right)  . \label{eq0312}%
\end{equation}

\end{theorem}

\begin{proof}
We will actually show the relationships (equivalent \ to (\ref{eq0312}) )
\[
V_{N}^{D}\left(  t\right)  \subseteq V_{E}^{D}\left(  t\right)  ,E_{N}%
^{D}\left(  t\right)  \subseteq E_{E}^{D}\left(  t\right)  .
\]
The proof is by induction. At $s=0$ all nodes are dirty in both the edge game
and the node game, hence $V_{N}^{D}\left(  0\right)  =V_{E}^{D}\left(
0\right)  =V$. Suppose $V_{N}^{D}\left(  s\right)  \subseteq V_{E}^{D}\left(
s\right)  $ for $s=1,2,...,t$. Now we take $s=t+1$ and consider first the
cleared and then the recontaminated nodes.

\begin{enumerate}
\item[\textbf{I.}] \emph{Cleared Nodes.} By Lemma \ref{prp0302}, a node
becomes e-cleared \emph{exactly }at $t+1$ only if entered at $t+1$. Hence at
most one such node (call it $x_{0}$)\ exists and, since entered at $t+1$, it
is also n-clear at $t+1$. In other words: $x_{0}\in V_{E}^{D}\left(  t\right)
\cap V_{E}^{C}\left(  t+1\right)  \Rightarrow x_{0}\in V_{N}^{C}\left(
t+1\right)  $.

\item[\textbf{II.}] \emph{Recontaminated Nodes.} Denote by $A$ the set of
nodes recontaminated (in the node game) \emph{exactly} at $t+1$; i.e.,
$A=V_{N}^{C}\left(  t\right)  \cap V_{N}^{D}\left(  t+1\right)  $. By the
definition of n-dirty node, every $u\in A$ must have a path (n-unguarded at
$t+1$) to a node $v\in V_{N}^{D}\left(  t+1\right)  $. Furthermore, at least
one $u\in A$ must have a path (n-unguarded at $t+1$) to a node $v\in V_{N}%
^{D}\left(  t\right)  $ (if this did not hold for any $u\in A$, then none of
the nodes of $A$ would get recontaminated to begin with) and in fact this
shows \emph{every} $u\in A$ must have a path (n-unguarded at $t+1$) to a node
$v\in V_{N}^{D}\left(  t\right)  $. So take any $u\in A$, then there exists a
path $v_{1}v_{2}...v_{K}u$ satisfying:%
\begin{align}
v_{1}  &  \in V_{N}^{D}\left(  t\right)  \cap V_{N}^{D}\left(  t+1\right)
,\label{eq0431}\\
&  v_{1}v_{2}...v_{K}u\text{ is n-unguarded at }t+1. \label{eq0432}%
\end{align}
Now we will show that $v_{1}$ must belong to $V_{E}^{D}\left(  t\right)  \cap
V_{E}^{D}\left(  t+1\right)  $ as well. First, $v_{1}\in V_{N}^{D}\left(
t\right)  \subseteq V_{E}^{D}\left(  t\right)  $. Second, if $v_{1}\notin
V_{E}^{D}\left(  t+1\right)  $ then $v_{1}$ was cleared exactly at $t+1$,
which (by Lemma \ref{prp0302}) means a searcher entered $v_{1}$ at $t+1$; but
then $v_{1}\notin V_{N}^{D}\left(  t+1\right)  $ which is contrary to
(\ref{eq0431}). Hence $v_{1}\in V_{E}^{D}\left(  t\right)  \cap V_{E}%
^{D}\left(  t+1\right)  $. \ Since $v_{1}$ is \emph{e-dirty} at time $t+1$,
$v_{1}v_{2}$ is e-dirty at $t+1$ by Lemma \ref{prp0301}. Also $v_{1}%
v_{2}...v_{K}u$ is e-unguarded at $t+1$. Hence $u\in V_{E}^{D}\left(
t+1\right)  $. Since this is true for every $u\in A$, we conclude%
\begin{equation}
A\subseteq V_{E}^{D}\left(  t+1\right)  . \label{eq0407}%
\end{equation}

\end{enumerate}

If the move at $t+1$ was placing or sliding a searcher into node $x_{0}$, let
$B=\left\{  x_{0}\right\}  $; else let $B=\emptyset$. Since $V_{N}^{D}\left(
t\right)  \subseteq V_{E}^{D}\left(  t\right)  $ we have
\begin{equation}
V_{N}^{D}\left(  t\right)  -B\subseteq V_{E}^{D}\left(  t\right)  -B.
\end{equation}
Since the e-dirty nodes at $t+1$ include all the e-dirty nodes at $t$ except
the node $x_{0}$ (which was e-cleared at $t+1$), we have%
\begin{equation}
V_{E}^{D}\left(  t\right)  -B\subseteq V_{E}^{D}\left(  t+1\right)  .
\end{equation}
And so we have
\begin{equation}
V_{N}^{D}\left(  t\right)  -B\subseteq V_{E}^{D}\left(  t+1\right)  .
\label{eq0308}%
\end{equation}
The n-dirty nodes at $t+1$ are exactly the n-dirty nodes at $t$, excluding
nodes n-cleared at $t+1$, plus the nodes recontaminated at $t+1$. In other
words%
\begin{equation}
V_{N}^{D}\left(  t+1\right)  =\left(  V_{N}^{D}\left(  t\right)  -B\right)
\cup A \label{eq0309}%
\end{equation}
Combining eqs. (\ref{eq0407}), (\ref{eq0308}) and (\ref{eq0309}) we get%
\begin{equation}
V_{N}^{D}\left(  t+1\right)  =\left(  V_{N}^{D}\left(  t\right)  -B\right)
\cup A\subseteq V_{E}^{D}\left(  t+1\right)  .
\end{equation}
Hence by induction we get $V_{N}^{D}\left(  t\right)  \subseteq V_{E}%
^{D}\left(  t\right)  $ \ (and $V_{E}^{C}\left(  t\right)  \subseteq V_{N}%
^{C}\left(  t\right)  $)\ for $t=0,1,2,...$ .

Let us now consider the edge sets $E_{N}^{D}\left(  t\right)  $ and $E_{E}%
^{D}\left(  t\right)  $. Take any $t$ and any edge $uv\in E_{N}^{D}\left(
t\right)  $. Without loss of generality we can assume $u\in V_{N}^{D}\left(
t\right)  \subseteq V_{E}^{D}\left(  t\right)  $ and hence, by Lemma
\ref{prp0301}, $uv\in E_{E}^{D}\left(  t\right)  $. Hence, for $t=0,1,2,...$,
we have
\[
E_{N}^{D}\left(  t\right)  \subseteq E_{E}^{D}\left(  t\right)  \text{\quad
and\quad}E_{E}^{C}\left(  t\right)  \subseteq E_{N}^{C}\left(  t\right)  ,
\]
which completes the proof.
\end{proof}

\begin{corollary}
\label{prp0304}Given a graph $\mathbf{G}$ and an internal search schedule
$\mathbf{S}$, if $\mathbf{S}$ is edge clearing it is also node clearing.
\end{corollary}

\begin{proof}
Suppose the final move of $\mathbf{S}$ took place at time $t_{fin}$, then
$E_{E}^{D}\left(  t_{fin}\right)  =\emptyset$. Since $E_{N}^{D}\left(
t_{fin}\right)  \subseteq E_{E}^{D}\left(  t_{fin}\right)  =\emptyset$ we have
$E_{N}^{D}\left(  t_{fin}\right)  =\emptyset$ which also means $V_{N}%
^{D}\left(  t_{fin}\right)  =\emptyset$ (if a node $u$ was n-dirty, an
adjacent edge $uv$ would also be n-dirty).
\end{proof}

\begin{remark}
\label{prp0305}From Theorem \ref{prp0303} we see that edge search is
\textquotedblleft weaker\textquotedblright\ than node search, i.e., for every
search schedule\textbf{,} and at every time step, the clear set in the edge
game is smaller or equal than the one in the node game. There is a variant of
edge search, the so-called \emph{mixed edge search }which, as we will see in
Appendix \ref{secB}, is \emph{equivalent} to node search, i.e. every search
schedule produces the same clear and dirty sets at every step of both the edge
and node game (this is stated and proved as Theorem \ref{prpA01}). From this
fact follow two consequences. First, \emph{node search is NP-complete} (since
mixed edge search is NP-complete \cite{Bienstock}). Second, Theorem
\ref{prp0303} can be obtained as a corollary of Theorem \ref{prpA01}. A more
detailed discussion of these issues appears in Appendix \ref{secB}.
\end{remark}

\section{Search on Trees}

\label{sec04}

In this section we will focus on \emph{rooted IMC\ (internal, monotone,
connected)\ search on trees} and we will compare node and edge search in this context.

Several authors have studied rooted IMC\ \emph{edge} search \cite{Barriere1,
Barriere2, Fomin, Fraigniaud1, Fraigniaud2, Yang2}. Our main results in this
section are Theorems \ref{prp0401} and \ref{prp0402} which establish the
equivalence of rooted IMC\ node and edge search on trees. Hence several
results already established for edge search also hold for node search (on
trees) as will be seen in the sequel.

\begin{theorem}
\label{prp0401}Given a tree $\mathbf{T}$ and an internal \emph{rooted }search
schedule $\mathbf{S}$, if $\mathbf{S}$ is monotone connected \emph{in the node
game} then
\[
\text{for }t=0,1,2,...:\qquad V_{E}^{C}\left(  t\right)  =V_{N}^{C}\left(
t\right)  \ \text{ and }\ E_{E}^{C}\left(  t\right)  =E_{N}^{C}\left(
t\right)
\]
and $\mathbf{S}$ is monotone connected in the edge game.
\end{theorem}

\begin{proof}
The proof is inductive. We have $V_{E}^{C}\left(  0\right)  =V_{N}^{C}\left(
0\right)  =\emptyset$ and $E_{E}^{C}\left(  0\right)  =E_{N}^{C}\left(
0\right)  =\emptyset$. Also, since the search is rooted and internal, the move
at $t=1$ will be the placement of a searcher into the root $u_{0}$; hence
$V_{E}^{C}\left(  1\right)  =V_{N}^{C}\left(  1\right)  =\left\{
u_{0}\right\}  $ and $E_{E}^{C}\left(  1\right)  =E_{N}^{C}\left(  1\right)
=\emptyset$. Now suppose that at time $t\geq1$ we have
\begin{equation}
V_{E}^{C}\left(  t\right)  =V_{N}^{C}\left(  t\right)  \text{,\qquad}E_{E}%
^{C}\left(  t\right)  =E_{N}^{C}\left(  t\right)  \label{eq0401}%
\end{equation}
and let us consider the move $u\rightarrow v$, performed at $t+1$, $t>1$.

We first dispose of the case where the move involves the node 0. Because of
internality, $v\neq0$ (no searcher is removed from the graph). If the move is
$0\rightarrow v$, by rootedness we must have $v=u_{0}$. Also by rootedness,
$u_{0}$ is the first node cleared (at some $t=1$). Now, in the node game:
(a)\ by monotonicity, no node is recontaminated at $t+1$ and (b)\ no node is
n-cleared either (since $u_{0}$ \emph{remains }n-clear, the move $0\rightarrow
u_{0}$ implies $V_{N}^{C}\left(  t+1\right)  -V_{N}^{C}\left(  t\right)
=\emptyset$); hence $V_{N}^{C}\left(  t\right)  =V_{N}^{C}\left(  t+1\right)
$ which also means $E_{N}^{C}\left(  t\right)  =E_{N}^{C}\left(  t+1\right)
$. In the edge game, $0\rightarrow u_{0}$ (a)\ does not clear any new edges
(so it does not clear any new nodes either) and (b)\ does not remove any
searchers from any node, so neither edge \ nor node recontamination is
possible; hence $V_{E}^{C}\left(  t\right)  =V_{E}^{C}\left(  t+1\right)  $
and $E_{E}^{C}\left(  t\right)  =E_{E}^{C}\left(  t+1\right)  $. In short, if
the $\left(  t+1\right)  $-th move involves node 0, from (\ref{eq0401}%
)\ follows
\[
V_{E}^{C}\left(  t+1\right)  =V_{N}^{C}\left(  t+1\right)  \text{,\qquad}%
E_{E}^{C}\left(  t+1\right)  =E_{N}^{C}\left(  t+1\right)  .
\]

Next we examine the case where the $\left(  t+1\right)  $-th move is
$u\rightarrow v$, with $u,v\neq0$. We will examine the effects of this move
separately for the node and edge game. \noindent

\noindent\textbf{I.} \emph{Node Game}. $u$ was guarded at time $t$, so $u\in
V_{N}^{C}\left(  t\right)  $. Suppose, first, that at $t$ only one searcher
was located on $u,$ so $u$ is unguarded at $t+1$. However, since $\mathbf{S}$
was supposed monotone in the node game, there is no node $x\in$ $V_{N}%
^{C}\left(  t\right)  \cap$ $V_{N}^{D}\left(  t+1\right)  $ and, especially,
$u$ remains n-clear. Since no node has become n-dirty at $t+1$, no edge has
become n-dirty either (an edge becomes n-dirty iff incident on an n-dirty
node). As for clearings, $v$ is n-clear at $t+1$, since entered by a searcher;
and $uv$ is also n-clear since $\left\{  u,v\right\}  \subseteq V_{N}%
^{C}\left(  t+1\right)  $. No node other than $v$ was entered at $t+1$, so no
node other than $v$ was n-cleared. Is there an edge (other than $uv$)\ which
was n-cleared exactly at $t+1$? This is only possible if (a)\ the edge has the
form $vx_{1}$ and (b)$\ v\in V_{N}^{D}\left(  t\right)  \cap V_{N}^{C}\left(
t+1\right)  $ ($v$ became n-clear \emph{exactly} at $t+1$) and (c)\ $x_{1}\in
V_{N}^{C}\left(  t\right)  \cap V_{N}^{C}\left(  t+1\right)  $. Suppose
$x_{1}\neq u$. Since $\mathbf{G}_{N}^{C}\left(  t\right)  $ is connected,
there is a path $ux_{L-1}...x_{1}$ which is \emph{inside }$\mathbf{G}_{N}%
^{C}\left(  t\right)  $ and hence \emph{does not contain }$v$. Then the tree
$\mathbf{T}$ contains a cycle $ux_{L-1}...x_{1}vu$ which is a contradiction.
Hence the only edge \emph{possibly} n-cleared at $t+1$ (if not already n-clear
at $t$) is $uv$. In short%
\begin{equation}
V_{N}^{C}\left(  t+1\right)  =V_{N}^{C}\left(  t\right)  \cup\left\{
v\right\}  \text{ and }E_{N}^{C}\left(  t+1\right)  =E_{N}^{C}\left(
t\right)  \cup\left\{  uv\right\}  . \label{eq0402}%
\end{equation}
The case that $u$ is still guarded at $t+1$ (i.e. $u$ contained more than one
searcher at time $t$) is omitted, since it is similar but easier to treat than
the case of unguarded $u$. \noindent

\noindent\textbf{II. }\emph{Edge Game}. Again, we first examine the case that
only one searcher was located in $u$ at time $t$. So $u$ was guarded (and
e-clear) at $t$ and is unguarded at $t+1$. We will now show that $u$ remains
e-clear at $t+1$; in other words that
\begin{equation}
u\in V_{E}^{C}\left(  t\right)  \cap V_{E}^{C}\left(  t+1\right)  .
\label{eq0441}%
\end{equation}
Suppose, on the contrary, that we have%
\begin{equation}
u\in V_{E}^{C}\left(  t\right)  \cap V_{E}^{D}\left(  t+1\right)  ;
\label{eq0443}%
\end{equation}
There are only two ways for (\ref{eq0443})\ to hold, each of which we examine separately.

\begin{enumerate}
\item[\textbf{II.1}] $\exists x_{1}\neq v$ such that $ux_{1}\in E_{E}%
^{D}\left(  t\right)  \cap E_{E}^{D}\left(  t+1\right)  $. Then $ux_{1}\in
E_{E}^{D}\left(  t\right)  =E_{N}^{D}\left(  t\right)  $. Since the move at
$t+1$ was $u\rightarrow v$, $u$ was guarded at $t$ and $x_{1}\in V_{N}%
^{D}\left(  t\right)  $. Now, from $u\rightarrow v$ and $x_{1}\neq v$ we get
$x_{1}\in V_{N}^{D}\left(  t+1\right)  $; from this and $u$ unguarded at $t+1$
we conclude that $u\in V_{N}^{D}\left(  t+1\right)  $ which contradicts
monotonicity of the node search.

\item[\textbf{II.2}] Alternatively, there exists a path $ux_{1}...x_{K-1}%
x_{K}$ (with $K>1$) such that $x_{K}x_{K-1}\in E_{E}^{D}\left(  t\right)  \cap
E_{E}^{D}\left(  t+1\right)  $ and $ux_{1}...x_{K-1}x_{K}$ is e-unguarded at
$t+1$. First note that, since $ux_{1}...x_{K}$ is e-unguarded at $t+1$ and the
$\left(  t+1\right)  $-th move is $u\rightarrow v$, it follows that
$v\notin\left\{  x_{1},...,x_{K-1}\right\}  $; and since $\mathbf{T}$ is a
tree, $v\neq x_{K}$ too; in short $v\notin\left\{  x_{1},...,x_{K}\right\}  $.
Also $u\notin\left\{  x_{1},...,x_{K}\right\}  $, since $ux_{1}...x_{K}$ is a
path. Now%
\[
\left.
\begin{array}
[c]{l}%
\text{the }\left(  t+1\right)  \text{-th move is }u\rightarrow v\\
u\notin\left\{  x_{1},...,x_{K}\right\} \\
ux_{1}...x_{K}\text{ is e-unguarded at }t+1
\end{array}
\right\}  \Rightarrow ux_{1}...x_{K}\text{ is e-unguarded at }t
\]
and this together with $x_{K}x_{K-1}\in E_{E}^{D}\left(  t\right)  $ implies
$x_{K-1}\in V_{E}^{D}\left(  t\right)  =V_{N}^{D}\left(  t\right)  $. From
$v\neq x_{K-1}$ follows that $x_{K-1}$ was not entered at $t+1$ and so
$x_{K-1}\in V_{N}^{D}\left(  t+1\right)  $. From this and $ux_{1}...x_{K-1}$
being n-unguarded at $t+1$ we conclude that $u\in V_{N}^{D}\left(  t+1\right)
$ which contradicts monotonicity of the node search.
\end{enumerate}

In short, $u$ cannot become e-dirty at $t+1$, i.e., we have proved
(\ref{eq0441}). Using this we now will show that no previously e-clear edge
$xy$ can become e-dirty at $t+1$. Because this would require the existence of
some edge $pq\in E_{E}^{D}\left(  t\right)  \cap E_{E}^{D}\left(  t+1\right)
$ and a path $xy...pq$ which is e-unguarded at $t+1$ but e-guarded at $t$
(otherwise $xy$ would already be e-dirty at $t)$. But any such path would have
to contain $u$ (no searcher was removed from any other node at $t+1$) which
would in turn imply $u\in V_{E}^{D}\left(  t+1\right)  $, which contradicts
(\ref{eq0441}). Since no edge is recontaminated no node is recontaminated
either (the only node guarded at $t$ and unguarded at $t+1$ is $u$, for which
we have (\ref{eq0441}) ).

On the other hand, the only edge \emph{possibly }e-cleared at time $t+1$ (if
not already e-clear at $t$) is $uv$, since no other edge was traversed at
$t+1$. Hence the only node \emph{possibly }e-cleared at time $t+1$ (if not
already e-clear at $t$) \ is $v$.

In short \ we have shown
\begin{equation}
V_{E}^{C}\left(  t+1\right)  =V_{E}^{C}\left(  t\right)  \cup\left\{
v\right\}  \text{ and }E_{E}^{C}\left(  t+1\right)  =E_{E}^{C}\left(
t\right)  \cup\left\{  uv\right\}  . \label{eq0403}%
\end{equation}
The treatment of the case that $u$ is still guarded at $t+1$ (i.e., at $t$ it
contained more than one searcher) is omitted, since it is similar but easier
to the case of $u$ unguarded. \medskip From \noindent(\ref{eq0401}),
(\ref{eq0402}) and (\ref{eq0403}) we obtain
\begin{equation}
V_{E}^{C}\left(  t+1\right)  =V_{N}^{C}\left(  t+1\right)  \text{ and }%
E_{E}^{C}\left(  t+1\right)  =E_{N}^{C}\left(  t+1\right)  \label{eq0404}%
\end{equation}
and, proceeding inductively, we obtain the required result:
\begin{equation}
\text{for }t=0,1,2,...:\qquad V_{E}^{C}\left(  t\right)  =V_{N}^{C}\left(
t\right)  \ \text{ and }\ E_{E}^{C}\left(  t\right)  =E_{N}^{C}\left(
t\right)  \label{eq0405a}%
\end{equation}
Monotonicity and connectedness in the edge game follow from the fact that
these hold in the node game and from eq.(\ref{eq0405a}).
\end{proof}

\begin{theorem}
\label{prp0402}Given a tree $\mathbf{T}$ and an internal \emph{rooted }search
schedule $\mathbf{S}$ with root $u_{0}$, if $\mathbf{S}$ is monotone connected
\emph{in the edge game} and satisfies either of the following constraints

\begin{enumerate}
\item[\textbf{C1}] $E_{E}^{C}\left(  1\right)  =\emptyset$, $E_{E}^{C}\left(
2\right)  =\left\{  u_{0}v\right\}  ,$

\item[\textbf{C2}] $E_{E}^{C}\left(  1\right)  =\emptyset$, $E_{E}^{C}\left(
2\right)  =\emptyset$, $E_{E}^{C}\left(  3\right)  =\left\{  u_{0}v\right\}  $,
\end{enumerate}

then%
\begin{equation}
\text{for }t=0,1,2,...:\qquad V_{E}^{C}\left(  t\right)  =V_{N}^{C}\left(
t\right)  \ \text{ and }\ E_{E}^{C}\left(  t\right)  =E_{N}^{C}\left(
t\right)  \label{eq0401b}%
\end{equation}
and $\mathbf{S}$ is monotone connected in the node game.
\end{theorem}

\begin{remark}
\label{prp0403}Before proving the theorem, let us discuss the significance of
constraints \textbf{C1} and \textbf{C2}. Three remarks must be made

\begin{enumerate}
\item Regarding the kind of searches which will satisfy either of the
constraints:\ \textbf{C1} will be satisfied by searches with $\mathbf{S}%
\left(  1\right)  =0\rightarrow u_{0}$, $\mathbf{S}\left(  2\right)
=u_{0}\rightarrow v$ (provided $u_{0}$ has a single neighbor, namely $v$);
\textbf{C2} will be satisfied by searches with $\mathbf{S}\left(  1\right)
=0\rightarrow u_{0}$, $\mathbf{S}\left(  2\right)  =0\rightarrow u_{0}$,
$\mathbf{S}\left(  3\right)  =u_{0}\rightarrow v$; in both cases $u_{0}$ is,
obviously, the root of the search.

\item The constraints are imposed to exclude situations similar to the one
discussed in Remark \ref{prp0202}, where an edge is e-cleared and e-dirtied
during the \emph{same }time step $t$.

\item Finally, \textbf{C1} and \textbf{C2} are not exceedingly restrictive; as
will be seen later, every \textquotedblleft interesting\textquotedblright%
\ IMC\ search on a tree either satisfies \textbf{C1} / \textbf{C2} or can
easily be converted to an equivalent search which does.
\end{enumerate}
\end{remark}

\begin{proof}
[Proof of Theorem \ref{prp0402}]The theorem will be proved by induction. It is
easy to check that
\begin{equation}
\text{for }t\in\left[  0,t_{1}\right]  :\qquad V_{E}^{C}\left(  t\right)
=V_{N}^{C}\left(  t\right)  \text{ and }E_{E}^{C}\left(  t\right)  =E_{N}%
^{C}\left(  t\right)
\end{equation}
with $t_{1}=2$ when \textbf{C1 }holds and $t_{1}=3$ when \textbf{C2} holds.
Now suppose that at time $t\geq t_{1}$ we have
\begin{equation}
V_{E}^{C}\left(  t\right)  =V_{N}^{C}\left(  t\right)  \text{,\qquad}E_{E}%
^{C}\left(  t\right)  =E_{N}^{C}\left(  t\right)  \label{eq0401a}%
\end{equation}
and consider the move performed at $t+1$.

We first dispose of the case where the move involves the node 0. Because of
internality, $v\neq0$ (no searcher is removed from the graph). If the move is
$0\rightarrow v$, by rootedness we must have $v=u_{0}$. By \textbf{C1} and
\textbf{C2}, $u_{0}$ is the first node cleared (at $t=1$). In the edge game,
at $t+1$, no edge is recontaminated (by monotonicity)\ and no edge is
e-cleared (since no edge is traversed); hence $E_{E}^{C}\left(  t\right)
=E_{E}^{C}\left(  t+1\right)  $; since no searcher is removed from a node,
$V_{E}^{C}\left(  t\right)  =V_{E}^{C}\left(  t+1\right)  $ as well. In the
node game, $u_{0}$ is already n-clear at $t_{1}$ and no other node is entered
at $t+1$, so no node is n-cleared; no searcher is removed from the graph, so
no node is recontaminated; hence $V_{N}^{C}\left(  t\right)  =V_{N}^{C}\left(
t+1\right)  $ which implies $E_{N}^{C}\left(  t\right)  =E_{N}^{C}\left(
t+1\right)  $ as well. In short, if the $\left(  t+1\right)  $-th move
involves node 0, from (\ref{eq0401a})\ and the above arguments follows
\[
V_{E}^{C}\left(  t+1\right)  =V_{N}^{C}\left(  t+1\right)  \text{,\qquad}%
E_{E}^{C}\left(  t+1\right)  =E_{N}^{C}\left(  t+1\right)  .
\]

Next we examine the case where the $\left(  t+1\right)  $-th move is
$u\rightarrow v$, with $u,v\neq0$. We will examine the effects of this move
separately for the node and edge game. \noindent

\noindent\textbf{I.} \emph{Edge Game}. By monotonicity of the search, no
previously e-clear edge becomes e-dirty at $t+1$; and hence the only node that
can \emph{possibly} become e-dirty is $u$. We will show however that $u$ also
remains e-clear. Since $E_{E}^{C}\left(  t_{1}\right)  =\left\{
u_{0}v\right\}  $ and $t\geq t_{1}$, from edge monotonicity follows that
$E_{E}^{C}\left(  t\right)  \neq\emptyset$ for all $t\geq t_{1}$. And, since
$u\in V_{N}^{C}\left(  t\right)  =V_{E}^{C}\left(  t\right)  $ and
$\mathbf{G}_{E}^{C}\left(  t\right)  $ is connected, there exists $uz\in
E_{E}^{C}\left(  t\right)  $; if $u$ becomes e-dirty at $t+1$, then $uz\in
E_{E}^{D}\left(  t+1\right)  $, which contradicts monotonicity of the edge
search. Hence no previously e-clear node becomes e-dirty at $t+1$.

Regarding new clearings, edge $uv$ and node $v$ \emph{may} become e-clear at
$t+1$ (if they were not already e-clear at $t$). No edge other than $uv$ was
traversed, hence no edge other than $uv$ (and no node other than $v$) can be
e-cleared at $t+1$. \ 

In short, in the edge game we have%
\begin{equation}
V_{E}^{C}\left(  t+1\right)  =V_{E}^{C}\left(  t\right)  \cup\left\{
v\right\}  \text{ and }E_{E}^{C}\left(  t+1\right)  =E_{E}^{C}\left(
t\right)  \cup\left\{  uv\right\}  . \label{eq0406}%
\end{equation}

\noindent\textbf{II.}\ \emph{Node Game}. Again, we first examine the case that
only one searcher was located on $u$ at time $t$. So $u$ was guarded (hence
n-clear) at $t$ and is unguarded at $t+1$. We will now show that $u$ remains
n-clear at $t+1$; in other words that
\begin{equation}
u\in V_{N}^{C}\left(  t\right)  \cap V_{N}^{C}\left(  t+1\right)  .
\label{eq0444}%
\end{equation}
Suppose on the contrary that we have%
\begin{equation}
u\in V_{N}^{C}\left(  t\right)  \cap V_{N}^{D}\left(  t+1\right)  ;
\label{eq0445}%
\end{equation}
(\ref{eq0445})\ can happen in only two ways, each of which we examine separately.

\begin{enumerate}
\item[\textbf{II.1}] $\exists x_{1}\neq v$ such that $ux_{1}\in E$ and
$x_{1}\in V_{N}^{D}\left(  t\right)  \cap V_{N}^{D}\left(  t+1\right)  $. Then
$x_{1}\in V_{N}^{D}\left(  t\right)  =V_{E}^{D}\left(  t\right)  \Rightarrow
ux_{1}\in E_{E}^{D}\left(  t\right)  $ and hence (since the move at $t+1$ was
$u\rightarrow v$ and $x_{1}\neq v$) $ux_{1}\in E_{E}^{D}\left(  t+1\right)  $.

Since $E_{E}^{C}\left(  t\right)  \neq\emptyset$ (for all $t\geq t_{1}$) and
$u\in V_{N}^{C}\left(  t\right)  =V_{E}^{C}\left(  t\right)  $ and
$\mathbf{G}_{E}^{C}\left(  t\right)  $ is connected, there exists $uz\in
E_{E}^{C}\left(  t\right)  $. However, since $ux_{1}\in E_{E}^{D}\left(
t+1\right)  $ and $u$ is unguarded at $t+1$, we conclude $uz\in E_{E}%
^{D}\left(  t+1\right)  $, which contradicts monotonicity of the edge search.

\item[\textbf{II.2}] Alternatively, there exists a path $ux_{1}...x_{K-1}%
x_{K}$ (with $K>1$) such that $x_{K}\in V_{N}^{D}\left(  t\right)  \cap
V_{N}^{D}\left(  t+1\right)  $ and $ux_{1}...x_{K-1}x_{K}$ is n-unguarded at
$t+1$ (so $v\notin\left\{  u,x_{1},...,x_{K}\right\}  $). First note that,
since $ux_{1}...x_{K}$ is a path, $u\notin\left\{  x_{1},...,x_{K}\right\}  $.
Now%
\[
\left.
\begin{array}
[c]{l}%
\text{the }\left(  t+1\right)  \text{-th move is }u\rightarrow v\\
u\notin\left\{  x_{1},...,x_{K}\right\} \\
ux_{1}...x_{K}\text{ is n-unguarded at }t+1
\end{array}
\right\}  \Rightarrow ux_{1}...x_{K}\text{ is e-unguarded at }t
\]
Since $x_{K}\in V_{N}^{D}\left(  t\right)  =V_{E}^{D}\left(  t\right)  $ it
follows that $x_{K-1}x_{K}\in E_{E}^{D}\left(  t\right)  $; since also
$x_{K-1}x_{K}$ is not traversed at $t+1$, it follows that $x_{K-1}x_{K}\in
E_{E}^{D}\left(  t+1\right)  $. By the same reasoning as in the previous case,
there exists $uz\in E_{E}^{C}\left(  t\right)  $ and, since $ux_{1}\in
E_{E}^{D}\left(  t+1\right)  $ and $u$ is unguarded at $t+1$, we conclude
$uz\in E_{E}^{D}\left(  t+1\right)  $, which contradicts the edge monotonicity
of $\mathbf{S}$.
\end{enumerate}

In short, $u$ cannot become n-dirty at $t+1$, i.e., we have proved
(\ref{eq0444}). Using this fact and an argument similar to that of Theorem
\ref{prp0401}, we conclude that no previously n-clear node can become n-dirty
at $t+1$. Since no node is recontaminated, no edge is recontaminated either.

On the other hand, the only node \emph{possibly }n-cleared at time $t+1$ is
$v$ (no other node was entered) and hence the only edge \emph{possibly
}n-cleared at time $t+1$ is $uv$ (if another edge was cleared, by an analysis
similar to that of Theorem \ref{prp0401}, we conclude that a cycle must exist
in $\mathbf{T}$, which is impossible). Hence the only edge possibly n-cleared
at $t+1$ is $uv$.

In short, in the node game we have%
\begin{equation}
V_{N}^{C}\left(  t+1\right)  =V_{N}^{C}\left(  t\right)  \cup\left\{
v\right\}  \text{ and }E_{N}^{C}\left(  t+1\right)  =E_{N}^{C}\left(
t\right)  \cup\left\{  uv\right\}  . \label{eq0407a}%
\end{equation}
The treatment of the case that $u$ is still guarded at $t+1$ (i.e., it
contained more than one searcher at $t$) is omitted, since it is similar to
but easier than the case of $u$ unguarded. \medskip From \noindent
(\ref{eq0401a}), (\ref{eq0406}) and (\ref{eq0407a}) we obtain
\begin{equation}
V_{E}^{C}\left(  t+1\right)  =V_{N}^{C}\left(  t+1\right)  \text{ and }%
E_{E}^{C}\left(  t+1\right)  =E_{N}^{C}\left(  t+1\right)
\end{equation}
and, proceeding inductively, we obtain the required result
\begin{equation}
\text{for }t=0,1,2,...:\qquad V_{E}^{C}\left(  t\right)  =V_{N}^{C}\left(
t\right)  \ \text{ and }\ E_{E}^{C}\left(  t\right)  =E_{N}^{C}\left(
t\right)  . \label{eq0447}%
\end{equation}
Monotonicity and connectedness in the node game follow from the fact that
these hold in the edge game and eq.(\ref{eq0447}).
\end{proof}

We will now review some results (obtained by Barriere et al.
\cite{Barriere1,Barriere2}) regarding IMC\ \emph{edge }search on trees and
will show that similar results hold for IMC\ \emph{node }search on trees. In
the following presentation the terminology and notation of
\cite{Barriere1,Barriere2} is somewhat changed, to conform with the one used
in the current paper.

The first basic result of \cite{Barriere1} is the following.

\begin{theorem}
\label{prp0404}\cite{Barriere1} For every tree $\mathbf{T}$ there is an IMC
edge clearing search schedule $\mathbf{S}$ which uses $s_{E}^{ic}\left(
\mathbf{T}\right)  $ searchers. Moreover, in $\mathbf{S}$ all searchers are
initially placed at the same node $u_{0}$ (i.e., the search schedule is
rooted) and the first step (after placing the searchers) consists in clearing
an edge incident to $u_{0}$.
\end{theorem}

\begin{corollary}
\label{prp0405}For every tree $\mathbf{T}$ we have $s_{E}^{ic}\left(
\mathbf{T}\right)  =s_{E}^{imc}\left(  \mathbf{T}\right)  $. There exists a
\emph{rooted }IMC\ search schedule which achieves this bound.
\end{corollary}

Theorem \ref{prp0404} can be extended to the following sequence of
inequalities\cite{Barriere2}.

\begin{theorem}
\label{prp0406}\cite{Barriere2} For every tree $\mathbf{T}$ we have$.$%
\begin{equation}
s_{E}\left(  \mathbf{T}\right)  =s_{E}^{i}\left(  \mathbf{T}\right)
=s_{E}^{m}\left(  \mathbf{T}\right)  \leq s_{E}^{im}\left(  \mathbf{T}\right)
=s_{E}^{c}\left(  \mathbf{T}\right)  =s_{E}^{ic}\left(  \mathbf{T}\right)
=s_{E}^{mc}\left(  \mathbf{T}\right)  =s_{E}^{imc}\left(  \mathbf{T}\right)
\leq2s_{E}\left(  \mathbf{T}\right)  -2. \label{eq0451}%
\end{equation}
Furthermore, there are trees $\mathbf{T}$ for which the inequality $s_{E}%
^{m}\left(  \mathbf{T}\right)  \leq s_{E}^{im}\left(  \mathbf{T}\right)  $ is strict.
\end{theorem}

Barriere et al. also present the \textbf{Search }algorithm \cite{Barriere1},
which computes a minimal edge-clearing rooted IMC\ schedule $\mathbf{S}$ for
every tree $\mathbf{T}$. The next Theorem shows that $\mathbf{S}$ is also a
minimal \emph{node-}clearing rooted IMC\ schedule.

\begin{theorem}
\label{prp0413}For every tree $\mathbf{T}$, the search $\mathbf{S}$ produced
by the \textbf{Search} algorithm is

\begin{enumerate}
\item a minimal edge-clearing IMC\ search of $\mathbf{T}$, which is also rooted;

\item a minimal node-clearing IMC\ search of $\mathbf{T}$, which is also rooted.
\end{enumerate}
\end{theorem}

\begin{proof}
Part 1 of the Theorem is proved as Lemma 9 in \cite{Barriere1}. Let us now
prove part 2. Since $\mathbf{S}$ is an edge-clearing strategy it is also
node-clearing. We next show it is minimal.\ 

Let us first consider the case of trees $\mathbf{T}$ such that $s_{E}%
^{imc}\left(  \mathbf{T}\right)  =1$. It is easy to see that such trees are
paths of the form $u_{1}u_{2}...u_{N}$, and \textbf{Search }produces the
\textquotedblleft obvious\textquotedblright\ edge clearing schedule:
$0\rightarrow u_{1}$, $u_{1}\rightarrow u_{2}$, ... (or the reverse, starting
at $u_{N}$) which is also a rooted minimal node clearing IMC\ schedule.

Suppose now that $s_{E}^{imc}\left(  \mathbf{T}\right)  \geq2$. By part 1,
\textbf{Search }will produce a rooted minimal edge clearing IMC\ schedule
$\mathbf{S}$. Also, by relaxing the requirement that all searchers are placed
into the graph in the first move (in other words, by interspersing searcher
placements with edge slidings) we can obtain from $\mathbf{S}$ a new search
schedule $\mathbf{S}^{\prime}$ which

\begin{enumerate}
\item uses the same number of searchers as $\mathbf{S}$\textbf{ }%
(i.e.\textbf{, }$\overline{sn}\left(  \mathbf{S}\right)  =\overline{sn}\left(
\mathbf{S}^{\prime}\right)  $\textbf{);}

\item is rooted (i.e. all searcher placements are into $u_{0}$);

\item satisfies either \textbf{C1} or \textbf{C2 }of Theorem \ref{prp0402};

\item produces the same sequence of clear graphs $\mathbf{G}_{N}^{C}\left(
t\right)  $ as $\mathbf{S}$ \ (with a slight time adjustment, to account for
the changes in searcher placement times).
\end{enumerate}

Hence $\mathbf{S}^{\prime}$ satisfies the conditions of Theorem \ref{prp0402}
and is IMC\ in the edge game which means that, using $\mathbf{S}^{\prime}$, we
have
\begin{equation}
\text{for }t=0,1,2,...:\qquad V_{E}^{C}\left(  t\right)  =V_{N}^{C}\left(
t\right)  \ \text{ and }\ E_{E}^{C}\left(  t\right)  =E_{N}^{C}\left(
t\right)  .
\end{equation}
From this also follows that $\mathbf{S}^{\prime}$ is a node clearing rooted
IMC\ search schedule of $\mathbf{T}$.

Now take another node clearing IMC\ search of $\mathbf{T}$, call it
$\mathbf{S}^{\prime\prime}$, which is minimal for the node game (i.e.,
$\overline{sn}\left(  \mathbf{S}^{\prime\prime}\right)  =$ $s_{N}^{imc}\left(
\mathbf{T}\right)  $). Then, by Theorem \ref{prp0401}, $\mathbf{S}%
^{\prime\prime}$ is also edge clearing, hence%
\[
\overline{sn}\left(  \mathbf{S}^{\prime\prime}\right)  \geq s_{E}^{imc}\left(
\mathbf{T}\right)  =\overline{sn}\left(  \mathbf{S}\right)  .
\]
On the other hand, since $\mathbf{S}^{\prime}$ is node clearing,
\[
\overline{sn}\left(  \mathbf{S}^{\prime\prime}\right)  =s_{N}^{imc}\left(
\mathbf{T}\right)  \leq\overline{sn}\left(  \mathbf{S}^{\prime}\right)
=\overline{sn}\left(  \mathbf{S}\right)  .
\]
Hence
\[
s_{N}^{imc}\left(  \mathbf{T}\right)  =\overline{sn}\left(  \mathbf{S}%
^{\prime\prime}\right)  =\overline{sn}\left(  \mathbf{S}\right)  =s_{E}%
^{imc}\left(  \mathbf{T}\right)
\]
which completes the proof.
\end{proof}

\begin{corollary}
\label{prp0414}For every tree $\mathbf{T}$ we have $s_{E}^{imc}\left(
\mathbf{T}\right)  =s_{N}^{imc}\left(  \mathbf{T}\right)  $.
\end{corollary}

\begin{remark}
\label{prp0416}Theorem \ref{prp0413} shows that the Barriere \textbf{Search}
algorithm can be used to compute a rooted IMC\ node clearing schedule for
every tree $\mathbf{T}$. How good is such a schedule? In other words, can we
node clear a tree $\mathbf{T}$ with fewer than $s_{N}^{imc}\left(
\mathbf{T}\right)  $ searchers? We can certainly do this with an internal
(non-monotone, non-connected)\ search. For example, let $\mathbf{T}^{\left(
2M\right)  }$ be the complete binary tree of height $2M$; then $s_{N}%
^{i}\left(  \mathbf{T}^{\left(  2M\right)  }\right)  =M$, while $s_{N}%
^{imc}\left(  \mathbf{T}^{\left(  2M\right)  }\right)  =2M$; but we do not
know the values of $s_{N}^{im}\left(  \mathbf{T}^{\left(  2M\right)  }\right)
$ and $s_{N}^{ic}\left(  \mathbf{T}^{\left(  2M\right)  }\right)  $. More
generally, we do not have the analog for node search of the inequalities
(\ref{eq0451}). It is easy to see that%
\begin{equation}
s_{N}\left(  \mathbf{G}\right)  =s_{N}^{i}\left(  \mathbf{G}\right)  \leq
s_{N}^{m}\left(  \mathbf{G}\right)  \leq%
\begin{array}
[c]{c}%
s_{N}^{im}\left(  \mathbf{G}\right) \\
s_{N}^{c}\left(  \mathbf{G}\right)  =s_{N}^{ic}\left(  \mathbf{G}\right)
\end{array}
\leq s_{E}^{mc}\left(  \mathbf{G}\right)  =s_{E}^{imc}\left(  \mathbf{G}%
\right)  \label{eq0452}%
\end{equation}
for every graph $\mathbf{G}$ and, in particular, for every tree $\mathbf{T}$
(and the inequalities in (\ref{eq0452}) can be strict). Refining
(\ref{eq0452}) to something like (\ref{eq0451}) is a subject of our future
research. The question is of special interest to us because, as already
mentioned, the one \emph{absolute} requirement for robotic pursuit / evasion
is internality; monotonicity and connectedness are desirable but not indispensable.
\end{remark}

\begin{remark}
\label{prp0417}A basic component of Barriere's \textbf{Search} algorithm is
the \textquotedblleft$\lambda$ labeling\textquotedblright\ of edges, performed
by the auxiliary \textbf{Label} algorithm\cite{Barriere1}. As will be seen in
the next section, the Barriere $\lambda$ labels are also used by our GSST
algorithm to node-clear \emph{general graphs}.
\end{remark}

\section{Search on Graphs}

\label{sec05}

We now turn to the study of rooted IMC node searches of an arbitrary graph
$\mathbf{G}=\left(  V,E\right)  $ (which will always be assumed to have $N$
nodes, i.e., $\left\vert V\right\vert =N$). In every such search, the $t$-th
move has the form $\mathbf{S}\left(  t\right)  =u\rightarrow v$, where $u,v\in
V$ or, perhaps, $u=0$, the \textquotedblleft source\textquotedblright\ node.

We present several variants of a basic node clearing algorithm. All variants
are guaranteed to find a node clearing search schedule; we prove that two of
the variants will find a minimal schedule with probability $1-\alpha^{M}$
where $\alpha\in\left(  0,1\right)  $ and $M$ is the number of iterations of
the algorithm; even non-minimal search schedules require a reasonably small
number of searchers, as will be seen by the experiments of Section \ref{sec06}.

\subsection{Motivation}

\label{sec0501}

Our basic algorithmic idea is motivated by the following rather simple
observation: \emph{every rooted IMC\ node-clearing search of }$\mathbf{G}%
$\emph{ generates a spanning tree}. This observation can be refined in the
form of the following theorem.

\begin{theorem}
\label{prp0501}Given a graph $\mathbf{G}=\left(  V,E\right)  $ and a rooted
IMC node clearing search $\mathbf{S}$ of $\mathbf{G}$. The clearing moves of
$\mathbf{S}$ generate a sequence of trees, $\left(  \mathbf{T}_{0}%
,\mathbf{T}_{1},\mathbf{T}_{2},....,\mathbf{T}_{N}\right)  $, where (for
$n=1,2,...,N$)\ $\mathbf{T}_{n}=\left(  V_{n},E_{n}\right)  $ and the
following hold:\ 

\begin{enumerate}
\item[\textbf{D1}] $\mathbf{T}_{0}$ is the empty graph ($V_{0}=\emptyset$,
$E_{0}=\emptyset$);

\item[\textbf{D2}] $\mathbf{T}_{N}$ is a spanning tree of $\mathbf{G}%
\ (V_{N}=V$, $E_{N}\subseteq E$);

\item[\textbf{D3}] for $n=1,2,...,N$: $V_{n-1}\subseteq V_{n}$, $E_{n-1}%
\subseteq E_{n}$ (in other words $\mathbf{T}_{n-1}$ is a subtree of
$\mathbf{T}_{n}$);

\item[\textbf{D4}] for $n=1,2,...,N$: $V_{n}=V_{n-1}\cup\left\{
u_{n}\right\}  $, and for $n=2,3,...$: $E_{n}=E_{n-1}\cup\left\{  u_{i}%
u_{n}\right\}  $, with $i\in\left[  1,n-1\right]  $.
\end{enumerate}
\end{theorem}

\begin{proof}
Inductively. Since $\mathbf{S}$ is monotone, it involves $N$ clearing moves.
$\mathbf{T}_{0}$ is the empty graph. $\mathbf{T}_{1}$ is formed by the first
move of $\mathbf{S}$, which consists in placing a searcher at the root node.
So $\mathbf{T}_{1}$ is the tree with a single node and trivially has
$\mathbf{T}_{0}\ $as a subgraph. Suppose \textbf{D3}\ and \textbf{D4}\ hold up
to $m=n$ and consider the $\left(  n+1\right)  $-th clearing move of
$\mathbf{S}$. Since $\mathbf{S}$ is connected, we add to $\mathbf{T}_{n}$\ one
node $u_{n+1}$ and (as explained in the Proof of Theorem \ref{prp0401})
\emph{exactly }one edge $u_{i}u_{n+1}$ (with $i\in\left[  1,n\right]  $) to
obtain a new tree $\mathbf{T}_{n+1}$ (of $n+1$ nodes and $n$ edges)\ which
also satisfies \ \textbf{D3}\ and \textbf{D4}. Hence \textbf{D3}\ and
\textbf{D4}\ hold for $m=1,2,...,N$. At $m=N$, $V_{N}$ contains $N$ nodes,
hence $V_{N}=V$; since $\mathbf{T}_{N}$ is a tree, it is a spanning tree of
$\mathbf{G}$.
\end{proof}

\begin{remark}
\label{prp0502}An abbreviated statement of Theorem \ref{prp0501} could be:
\textquotedblleft Every rooted IMC node clearing search of $\mathbf{G}$
specifies a spanning tree of $\mathbf{G}$\textbf{ }and an order of clearing
the nodes\textquotedblright. But the order must be consistent with the edge
structure of the graph and the spanning tree (e.g., we cannot use an order of
node clearing which requires non-existent edges). This consistence is exactly
what conditions \textbf{D3}\ and \textbf{D4}\ describe.
\end{remark}

Theorem \ref{prp0501} gives the motivation for our algorithm GSST
(\textbf{G}\emph{uaranteed \textbf{S}earch with \textbf{S}panning
\textbf{T}rees}), which is informally described below.

\begin{algorithm}[h]
\caption{ GSST, Informal Description}
\begin{algorithmic}
\STATE \textbf{Input:} Graph $\mathbf{G}$
\STATE Select a spanning tree $\mathbf{T}$ of $\mathbf{G}$ and a root $u_0$ of $\mathbf{T}$
\STATE Find a rooted IMC node-clearing search $\mathbf{S^\prime}$ of $\mathbf{T}$
\STATE Apply $\mathbf{S}$ to $\mathbf{G}$
\IF{at some step of $\mathbf{S^\prime}$ a move $u\rightarrow v$ would result in recontamination in $\mathbf{G}$}
\STATE send a ``guard'' to $u$
\STATE execute $u\rightarrow v$
\ENDIF
\STATE \textbf{Output:} The search $\mathbf{S}$ obtained by combining $\mathbf{S^\prime}$ with the guard moves
\end{algorithmic}
\label{alg:random}
\end{algorithm}

Note that, by construction, all clearing moves take place along edges of the
spanning tree $\mathbf{T}$.

In the above description we have used the term \textquotedblleft
guard\textquotedblright. Stated informally, the searchers play two
roles:\ \textquotedblleft\emph{tree searchers}\textquotedblright\ perform the
clearing moves, always along the edges of the spanning tree; \textquotedblleft%
\emph{guards}\textquotedblright\ are stationary and block potential
recontamination paths. However, note that a particular searcher can change
roles during the course of the search.

The main advantage of the GSST\ algorithm is that it is \emph{fast}. A random
spanning tree $\mathbf{T}$ can be quickly generated and searched. Since
$\mathbf{T}$ and $\mathbf{G}$ have the same node set, node-clearing
$\mathbf{T}$ in a \emph{node-monotone manner} results in node-clearing
$\mathbf{G}$ as well. The main issue is: \emph{how many guards will be
required to block recontamination through non-tree edges of }$G$? The number
usuallly turns out to be quite reasonable, because (a)\ guards can be reused
and (b)\ tree searchers\ can also be used as guards when they do not perform
clearing moves.

Both the labeling and the traversal phase of GSST\ can be executed in either a
centralized or distributed manner; the latter is useful for robotic
applications, where each robot can share some of the computational load. In
the distributed implementation, all searchers share the underlying spanning
tree and labeling. When a searcher reaches a node, he checks to see if he can
move without recontamination. If he can, he determines his next move based on
the traversal strategy, and he shares this move with the team.

Hence GSST has short execution time and can be run repeatedly (in reasonable
time), using many different (randomly selected)\ spanning trees. \ Our
algorithm depends on the quick discovery of a spanning tree corresponding to a
minimal or near-minimal search (Theorem \ref{prp0501}). Our experiments in
Section \ref{sec06} show that GSST discovers near-minimal searches in only a
short time for several families of complex graphs.

An important characteristic of the GSST\ algorithm is its \textquotedblleft
anytime-ness\textquotedblright. Anytime algorithms return a partial answer
before completion and they keep providing improved answers, the improvement
increasing with computation time \cite{Zilberstein}. GSST\ has these
characteristics, as will be understood in Section \ref{sec0502}. Namely,
GSST\ is characterized by monotonicity (the solution only improves over time),
recognizable quality (the quality of the solution, i.e. number of searchers,
can be determined at run time), consistency (the algorithm will not spend too
much time finding a single solution), and interruptibility.

\subsection{The GSST Algorithms}

\label{sec0502}

A detailed description of the GSST algorithm is given by the following
pseudocode on p.\pageref{gsst}. The notation $\mathbf{S}=\mathbf{S}|\left(
u\rightarrow v\right)  $ means that the move $\left(  u\rightarrow v\right)  $
is appended to the previously determined search schedule $\mathbf{S}$ (i.e.,
becomes the next move of $\mathbf{S}$). Several subroutines appearing in the
following listing will be discussed presently.

\begin{algorithm}[h]
\label{gsst}
\caption{GSST}
\begin{algorithmic}
\STATE \textbf{Input:} $\mathbf{G}$: a graph; $M$:  no. of spanning trees to use.
\STATE $\mathbf{S}_{min}=\emptyset$
\STATE $s_{min}=\infty$
\FOR{$m=1:M$}
\STATE $\mathbf{S}=\emptyset$
\STATE $\mathbf{T}$=\textbf{GenerateTree}($\mathbf{G}$)
\STATE Randomly choose root $u_0$
\STATE $\mathbf{S}=\mathbf{S}|(0\rightarrow u_0)$
\STATE $V_N^C=\{  u_0 \}$, $E_N^C=\emptyset$
\STATE $V_N^D=V-\{  u_0 \}$, $E_N^D=E$
\STATE $\mathbf{G}_N^C=(V_N^C,E_N^C)$
\STATE $\lambda$=\textbf{R-Label(T)}
\WHILE{$V_N^D \neq \emptyset$}
\STATE $uv$=\textbf{SelectEdge}($\mathbf{G}_N^C,\mathbf{T},E_N^D,\lambda$)
\IF{a searcher can traverse $uv$ without recontamination}
\STATE Move that searcher to $u$ \emph{staying inside the clear graph}
\STATE $\mathbf{S}=\mathbf{S}|(u\rightarrow v)$
\STATE $V_C=V_N^C\cup\{  v \}$, $E_N^C=E_C\cup\{  uv \}$
\STATE $V_N^D=V_N^D-\{ v \}$, $E_N^D=E_N^D-\{  uv \}$
\STATE $\mathbf{G}_N^C=(V_N^C,E_N^C)$
\ELSE
\STATE $\mathbf{S}=\mathbf{S}|(0\rightarrow u_0)$
\ENDIF
\ENDWHILE
\IF{$\overline{sn}(\mathbf{S})<s_{min}$}
\STATE $\mathbf{S}_{min}=\mathbf{S}$
\STATE $s_{min}=\overline{sn}(\mathbf{S})$
\ENDIF
\ENDFOR
\STATE \textbf{Output:} Node clearing schedule $\mathbf{S}_{min}$.
\end{algorithmic}
\label{alg:random}
\end{algorithm}

The following remarks explain the operation of the algorithm.

\begin{enumerate}
\item Generate a random spanning tree $\mathbf{T}$ by \textbf{GenerateTree.}
We have used two different methods of random spanning tree generation.

\begin{enumerate}
\item The \emph{uniform} method is an implementation of Wilson's
\cite{Wilson1} \emph{loop erased random walk algorithm}.

\item The \emph{DFS} method selects a root node and randomly moves down the
tree in a depth-first manner. At each node, a random incident edge is chosen
and set as an edge in the spanning tree. A visited list is maintained, and
when a node is visited more than once, the edge used to reach it the second
time is set as a non-tree edge. This eliminates cycles in the graph and thus
generates a tree. When a leaf is reached, the algorithm recurses to ensure
that all nodes are included in the tree (i.e, it is a spanning tree of the
original graph). The motivation for this method is to bias towards spanning
trees that require fewer guards. The intuition is that a DFS traversal will
generate only a few nodes with non-tree edges, thus leading to few required guards.
\end{enumerate}

\item Label the edges of $\mathbf{T}$ by the \textbf{R-Label} algorithm (this,
a modification of Barriere's \textbf{Label} algorithm, is listed and discussed
in Appendix \ref{secA}).

\item While n-dirty nodes still exist, select an n-dirty edge $uv$ of
$\mathbf{T}$ (by \textbf{SelectEdge,} to be discussed presently). Let a
searcher traverse $uv$ if this does not cause node
recontamination\footnote{The line \textquotedblleft Move that searcher to $u$
staying inside the clear graph\textquotedblright\ is actually a
simplification, i.e., in the interest of brevity, we do not indicate how such
a path is obtained; however this is always possible, usually in more than one
ways.}; if no such searcher exists, then use a new searcher (originally placed
at the root)\ to traverse $uv$.

\item Repeat the process until all nodes are n-cleared. (Hence the algorithm
does not allow node recontamination and will generate as many searchers as
necessary to prevent it).

\item When all nodes have been n-cleared, a rooted IMC node clearing search
has been generated, which performs all its clearing moves along the edges of
$\mathbf{T}$.

\item Go back to step 1 and repeat the process with a new spanning tree.

\item After the maximum number of trees and (corresponding searches) has been
generated, return a search $\mathbf{S}_{\min}$ which attains $s_{\min}$, the
minimum value of $\overline{sn}\left(  \mathbf{S}\right)  $.
\end{enumerate}

The subroutine \textbf{SelectEdge} chooses an n-dirty edge $uv$
\emph{belonging to }$\mathbf{T}$ and adjacent to the current clear graph
$\mathbf{G}_{\mathbf{N}}^{C}$. There are several ways to perform this selection.

\begin{enumerate}
\item \textbf{Labeled Selection (L). }Select the next edge $uv$ of
$\mathbf{T}$ to be traversed according to the Barriere $\lambda$ labels (as in
algorithm \textbf{R-Search, }presented in Appendix \ref{secA}); however, if
traversing $uv$ would cause recontamination then select the next edge in the
Barriere sequence; if, at some stage of the search, traversing \emph{any} edge
of $\mathbf{T}$ would cause recontamination (i.e., if all searchers are stuck)
then introduce a new searcher at the root.

\item \textbf{Labeled Selection with Randomized Tie-breaking (LR).} Same as
the previous except that ties of edge labels are broken randomly.

\item \textbf{Randomized Selection (R). }Choose $uv$ randomly (without making
\emph{any} use of the Barriere labeling)\ by a uniform probability on all
n-dirty edges \emph{belonging to }$\mathbf{T}$ and adjacent to $\mathbf{G}%
_{\mathbf{N}}^{C}$:%
\[
\Pr\left(  uv\right)  =\left\{
\begin{array}
[c]{cl}%
c\text{ } & \text{if }u\text{ is n-clear}\ \text{and }uv\text{ is
n-dirty}\ \text{and an edge of }\mathbf{T}\text{\textbf{;}}\\
0 & \text{else.}%
\end{array}
\right.
\]

\item \textbf{Labeled Weighted Selection (LW). }This rule is intermediate
between R and L: edge selection is still random but, instead of a uniform
probability distribution, the probability of an edge $uv$ being selected is
inversely proportional to its Barriere $\lambda$ label.

\item \textbf{Label Dominated Selection} \textbf{(LD).} This can be done by
labeling edges that lead to parts of the graph that are trees (subtrees of the
graph). A list of searchers who can move without recontamination can be
maintained during search. If an edge adjacent to $V_{N}^{C}\left(  t\right)  $
leads to a subtree of the graph, and enough free searchers are available,
clearing this subtree can only improve the search strategies.
\end{enumerate}

By using each of the above rules in the \textquotedblleft
basic\textquotedblright\ GSST algorithm, we obtain \emph{ten} GSST variants:
uniform GSST-L, uniform GSST-LR, ... , uniform GSST-LD, DFS\ GSST-L, ...,
DFS\ GSST-LD. These variants (except for the two GSST-R's) utilize the
Barriere labeling, \emph{but do not necessarily produce a Barriere traversal
of the spanning tree}\footnote{This is true even of GSST-L, because an edge
which would be next in Barriere's traversal \ order may be temporarily skipped
if its traversal (at the current stage of the search)\ would cause
recontamination.}. In a sense Randomized Selection is the simplest or most
naive rule that can be used to select the next move of the search
schedule:\ every n-dirty edge of $\mathbf{T}$ (adjacent to the clear
graph)\ is equally likely to be selected. The remaining three rules can be
understood as ways to bias the probability by which edges are selected in some
meaningful way. The effectiveness of these rules will be judged by the
experiments of Section \ref{sec06}. From the theoretical point of view, we
will show in Section \ref{sec0503}\ that the uniform GSST-R and uniform
GSST-LD will find a minimal node clearing schedule with probability
$1-\alpha^{M}$ where $M$ is the number of iterations and $\alpha\in\left(
0,1\right)  $; we conjecture that this property does not hold for Barriere
selection. As a practical matter, the issue is how large $M$ has to be for
$1-\alpha^{M}$ to be sufficiently close to 1. However, the experiments of
Section \ref{sec06} indicate that the above rules find good search schedules
in very reasonable time.

Variants of GSST can also be produced by replacing the \textbf{GenerateTree}
subroutine with an \emph{exhaustive generation} of all spanning trees of
$\mathbf{G}$ (it can be used in conjunction with any of the above variants of
\textbf{SelectEdge}). To do this we have used Char's spanning tree enumeration
algorithm \cite{Jayakumar}. Exhaustive enumeration is feasible only for
relatively small graphs.

Finally note that the search schedules produced by (every variant of)\ GSST
are IMC. This holds for the search of both $\mathbf{T}$ and $\mathbf{G}$.
Indeed, for the search of $\mathbf{G}$ to work, the search of $\mathbf{T}$
\emph{must} be IMC. In other words, no obvious modification of GSST will
produce, for example, an internal, connected, \emph{non-monotone}
node-clearing search of $\mathbf{G}$. The question arises:\ how good is a
minimal IMC\ node clearing of $\mathbf{G}$ (as compared to, for example, an
internal but not monotone / connected node clearing)? This is the question
already hinted at in Remark \ref{prp0416}.

Let us close this section by repeating that \emph{the basic idea of GSST\ is
to perform all clearing moves along the edges of a spanning tree}. This idea
exploits the facts that (a)\ spanning trees can be both generated and searched
quickly and (b)\ blocking recontamination does not require an excessively
large number of guards (because a searcher can change roles as a guard and a
tree searcher).

\subsection{Completeness}

\label{sec0503}

\begin{definition}
\label{prp0504}Given a graph $\mathbf{G}$ and a search $\mathbf{S}$ of
$\mathbf{G}$, the \emph{frontier }\ at $t$ (under $\mathbf{S}$) is
\[
V_{N}^{F}\left(  t\right)  =\left\{  u:u\in V_{N}^{C}\left(  t\right)  \text{
and }\exists v:v\in V_{N}^{D}\left(  t\right)  ,uv\in E\right\}  ,
\]
i.e., the n-clear nodes which are connected to n-dirty nodes.
\end{definition}

\begin{lemma}
\label{prp0505}In a node search, for every $t$, the nodes $u\in V_{N}%
^{F}\left(  t\right)  $ are guarded.
\end{lemma}

\begin{proof}
A frontier node $u$ is by definition n-clear and adjacent to an n-dirty node
$v$. This is only possible if $u$ is guarded.
\end{proof}

\begin{definition}
\label{prp0506}Consider a rooted IMC node clearing search $\mathbf{S}$ of
$\mathbf{G}$. Let $t_{0}=0$ and suppose the clearing moves of $\mathbf{S}$
take place at times $t_{1},...,t_{N}$; let also $t_{0}=0$. The $n$-th
\emph{phase} of $\mathbf{S}$ (for $n=1,2,...,N$) is the time interval $\left[
t_{n-1}+1,t_{n}\right]  $, i.e. the interval between the $\left(  n-1\right)
$-th and $n$-th clearing move.
\end{definition}

The following remarks are rather obvious. For $m=2,3,...$, a \emph{target edge
}$u_{i}u_{m}$ (with $i\in\left[  1,m-1\right]  $) corresponds to the $m$-th
phase (here we take the root node to be $u_{1}$). While $t\in\left[
t_{m-1}+1,t_{m}\right]  $ the algorithm moves a searcher towards $u_{i}u_{m}$.
For $t\in\left[  t_{m-1}+1,t_{m}-1\right]  $, $u_{i}u_{m}$ is n-dirty, $u_{i}$
is n-clear, $u_{m}$ is n-dirty. At $t=t_{m}$ we have $\mathbf{S}\left(
t_{m}\right)  =u_{i}\rightarrow u_{m}$ and $u_{i}u_{m}$, $u_{m}$ are n-cleared.

\begin{lemma}
\label{prp0507}Given a graph $\mathbf{G}=\left(  V,E\right)  $ and a rooted
IMC node clearing search $\mathbf{S}$ of $\mathbf{G}$, produced by either
GSST-R or GSST-LW, let $t_{1}$, $t_{2}$, ... , $t_{N}$ $\ $be the times at
which clearing moves take place; let also $t_{0}=0$. Then for $n=2,...,N:$

\begin{enumerate}
\item for $t\in\left[  t_{n-1},t_{n}-1\right]  $: $V_{N}^{F}\left(  t\right)
=V_{N}^{F}\left(  t_{n-1}\right)  $;

\item for $t\in\left[  t_{n-1},t_{n}-1\right]  $: every $u\in V_{N}^{F}\left(
t\right)  $ contains exactly one searcher, except one node $u\left(  t\right)
$ which possibly contains two searchers;

\item for $t=t_{n}$: every $u\in V_{N}^{F}\left(  t_{n}\right)  $ contains
exactly one searcher.
\end{enumerate}
\end{lemma}

\begin{proof}
The proof is by induction on $n$. Items 1, 2, 3 of the theorem hold for $n=1$,
$t\in\left[  t_{0},t_{1}\right]  =\left\{  0,1\right\}  $. Suppose they also
hold up to the $m$-th phase. In the interval $\left[  t_{m}+1,t_{m+1}\right]
$ a target edge $u_{i}u_{m+1}$ (with $i\in\left[  1,m\right]  $) is selected
and an available searcher is sent towards $u_{i}u_{m+1}$. Because GSST-R /
GSST-LW always selects for clearing an edge adjacent to the clear graph,
$u_{i}\in$ $V_{N}^{F}\left(  t_{m}\right)  $, $u_{m+1}\in$ $V_{N}^{D}\left(
t_{m}\right)  $. There are three cases.

\noindent\textbf{I}.\ $u_{i}$ is neighbor of a single n-dirty node, namely
$u_{m+1}$. In this case $t_{m+1}=t_{m}+1$ (i.e., $u_{m+1}$ is n-cleared in one
step) and $\left[  t_{m},t_{m+1}\right]  =\left\{  t_{m},t_{m}+1\right\}  $.
Node $u_{i}\notin V_{N}^{F}\left(  t_{m+1}\right)  $; \ node $u_{m+1}$ may or
may not belong to $V_{N}^{F}\left(  t_{m+1}\right)  $ but, at any rate, it
contains exactly one searcher; no searchers enter or exit any other nodes,
hence (by the inductive hypothesis)\ all $u\in V_{N}^{F}\left(  t_{m+1}%
\right)  $ contain exactly one searcher.

\noindent\textbf{II}.\ $u_{i}$ is neighbor of more than one n-dirty nodes, one
of which is $u_{m+1}$, and there are free searchers. Since $u_{i}$ is a
frontier node, by hypothesis it contains a single searcher who, consequently,
is stuck. However, if free searchers are available inside n-clear,
non-frontier nodes, one of these searchers will be sent to $u_{i}u_{m+1}$ by a
sequence of moves. At every $t$ during this sequence, the searcher may enter a
frontier node $w$; for that particular $t$, $w$ will be the only frontier node
which contains two searchers. At $t=t_{m+1}-1$ the searcher will be located at
$u_{i}$ (which will now contain two searchers) and at $t_{m+1}$ he will enter
$u_{m+1}$; this leaves at $t_{m+1}$ every $u\in V_{N}^{F}\left(  t_{m}\right)
$ with a single searcher and also places a searcher at $u_{m+1}$. The only
node which may have been added to the frontier is $u_{m+1}$ which contains a
single searcher; every other node $u\in V_{N}^{F}\left(  t+1\right)  $ was
previously in the frontier, contained a single searcher and, if a searcher
entered $u$ at some $t\in\left[  t_{m}+1,t_{m+1}-1\right]  $, it exited $u$ at
$t+1$; hence every such node at $t_{m+1}$ contains a single searcher.

\noindent\textbf{III}.\ The final case is when all frontier-located searchers
are stuck and there are no searchers inside n-clear, non-frontier nodes. In
this case a new searcher is placed at the root node and the rest of the
analysis is identical to that of Case II. Hence items 1, 2, 3 also hold for
$\left[  t_{m}+1,t_{m+1}\right]  $ and we can complete the induction for
$n=1,2,...,N$.
\end{proof}

\begin{lemma}
\label{prp0508}Given a graph $\mathbf{G}=\left(  V,E\right)  $ and a tree
sequence $\left(  \mathbf{T}_{0},\mathbf{T}_{1},...,\mathbf{T}_{N}\right)  $
which satisfies conditions \textbf{D1-D4} of Theorem \ref{prp0501}.
Then\ uniform GSST-R / GSST-LW with $M=1$ (i.e., using a single spanning tree)
has a nonzero probability of producing a search $\mathbf{S}$ which generates
$\left(  \mathbf{T}_{0},\mathbf{T}_{1},...,\mathbf{T}_{N}\right)  $.
\end{lemma}

\begin{proof}
The probability that Algorithm 1 generates the tree sequence $\left(
\mathbf{T}_{0},\mathbf{T}_{1},...,\mathbf{T}_{N}\right)  $ is
\[
\Pr\left(  \mathbf{T}_{0},\mathbf{T}_{1},...,\mathbf{T}_{N}\right)  =\left[
{\displaystyle\prod\limits_{n=1}^{N}}
\Pr\left(  \mathbf{T}_{n}|\mathbf{T}_{N},\mathbf{T}_{0},...,\mathbf{T}%
_{n-1}\right)  \right]  \Pr\left(  \mathbf{T}_{0}|\mathbf{T}_{N}\right)
\Pr\left(  \mathbf{T}_{N}\right)  .
\]
Note that the conditioning in the above expression \emph{always} includes
$\mathbf{T}_{N}$, since this is the first choice made in running GSST-R /
GSST-LW. Now obviously, $\Pr\left(  \mathbf{T}_{0}|\mathbf{T}_{N}\right)  =1$.
By Wilson's Theorem 1 \cite{Wilson1}, $\Pr\left(  \mathbf{T}_{N}\right)  >0$
for every spanning tree $\mathbf{T}_{N}$. Also, $\Pr\left(  \mathbf{T}%
_{n}\mathbf{|T}_{N}\mathbf{,T}_{0},\mathbf{T}_{1}\mathbf{,...,T}_{n-1}\right)
$ is the probability of expanding (at the $n$-th step)\ $\mathbf{T}_{n-1}$ by
the edge $u_{i}u_{n}\in\mathbf{E}_{n}-\mathbf{E}_{n-1}$ which, by the
construction of both GSST-R and GSST-LW, is always positive. Finally,
$\Pr\left(  \mathbf{T}_{N}\mathbf{|\mathbf{T}}_{N}\mathbf{,T}_{0}%
,\mathbf{T}_{1}\mathbf{,...,T}_{N-1}\right)  =1$. Hence $\Pr\left(
\mathbf{T}_{0},\mathbf{T}_{1},...,\mathbf{T}_{N}\right)  >0$ for every
sequence $\mathbf{T}_{1},...,\mathbf{T}_{N}$.
\end{proof}

\begin{lemma}
\label{prp0509}Given a graph $\mathbf{G}=\left(  V,E\right)  $ and a rooted
IMC node clearing search $\mathbf{S}$ of $\mathbf{G}$; let $\left(
\mathbf{T}_{0},\mathbf{T}_{1},...,\mathbf{T}_{N}\right)  $ be the tree
sequence generated by $\mathbf{S}$. Let $\mathbf{S}^{\prime}$ be a search
produced by either GSST-R or GSST-LW\ and \emph{also} generating $\left(
\mathbf{T}_{0},\mathbf{T}_{1},...,\mathbf{T}_{N}\right)  $. Then
$\overline{sn}\left(  \mathbf{S}\right)  \geq\overline{sn}\left(
\mathbf{S}^{\prime}\right)  $.
\end{lemma}

\begin{proof}
The proof is exactly the same for GSST-R and GSST-LW, so we only prove the
first one, by induction. Let $t_{1},...,t_{N}$ be the clearing times of
$\mathbf{S}$ and $t_{1}^{\prime},...,t_{N}^{\prime}$ be the clearing times of
$\mathbf{S}^{\prime}$. Also let $t_{0}=t_{0}^{\prime}=0.$

At $t_{0}=t_{0}^{\prime}=0\ $we have $sn\left(  \mathbf{S}^{\prime},0\right)
=sn\left(  \mathbf{S},0\right)  =0.$

The only times at which $sn\left(  \mathbf{S}^{\prime},t\right)  $ may change
are $1$, $t_{1}^{\prime}+1,...,t_{N-1}^{\prime}+1$ . Suppose that%
\[
sn\left(  \mathbf{S},t_{n}\right)  \geq sn\left(  \mathbf{S}^{\prime}%
,t_{n}^{\prime}\right)  .
\]
Further, suppose that at $t_{n}^{\prime}+1$ a new searcher is introduced in
$\mathbf{S}^{\prime}$. This can only happen (in the $\mathbf{S}^{\prime}$
search) if all of the following hold:

\begin{enumerate}
\item at $t_{n}^{\prime}$ exactly $\left\vert V_{N}^{F}\left(  t\right)
\right\vert $ searchers exist in $\mathbf{G}$;

\item there are no searchers inside nodes $u\in V_{N}^{C}\left(  t_{n}%
^{\prime}\right)  -V_{N}^{F}\left(  t_{n}^{\prime}\right)  $ (i.e., all
searchers are located inside frontier nodes);

\item all searchers are stuck (i.e., moving a searcher out of a frontier node
$u$ exposes $u$ to recontamination).
\end{enumerate}

The sequence $\left(  \mathbf{T}_{0},\mathbf{T}_{1},...,\mathbf{T}_{N}\right)
$ along with the clearing times determines the frontier $V_{N}^{F}\left(
t\right)  $ for every $t$. Hence $\mathbf{S}^{\prime}$ at $t_{n}^{\prime}$ has
the same frontier as $\mathbf{S}$ at $t_{n}$. If conditions 1-3 above hold in
$\mathbf{S}^{\prime}$, then every searcher is located in a frontier \ node and
is stuck. It is possible that non-stuck searchers exist in $\mathbf{S}$
(located either in frontier or non-frontier nodes) but this also means that
$sn\left(  \mathbf{S},t_{n}\right)  \geq sn\left(  \mathbf{S}^{\prime}%
,t_{n}^{\prime}\right)  +1$; hence adding a searcher in $\mathbf{S}^{\prime}$
at $t_{n}^{\prime}+1$ preserves
\[
sn\left(  \mathbf{S},t_{n}\right)  \geq sn\left(  \mathbf{S}^{\prime}%
,t_{n}^{\prime}+1\right)  .
\]
Since no searchers are added in $\mathbf{S}^{\prime}$ for $t\in\left[
t_{n}^{\prime}+2,t_{n+1}^{\prime}\right]  $ and no searchers are ever removed
in $\mathbf{S}$ (i.e., $sn\left(  \mathbf{S},t_{n+1}\right)  \geq sn\left(
\mathbf{S},t_{n}\right)  $) we also get
\[
sn\left(  \mathbf{S},t_{n+1}\right)  \geq sn\left(  \mathbf{S}^{\prime
},t_{n+1}^{\prime}\right)  .
\]
From the above inequality inductively we get $sn\left(  \mathbf{S}%
,t_{N}\right)  \geq sn\left(  \mathbf{S}^{\prime},t_{N}^{\prime}\right)
\ $which proves the Lemma.
\end{proof}

\begin{theorem}
\label{prp0510}Given a graph $\mathbf{G=}\left(  V,E\right)  $.

\begin{enumerate}
\item Uniform GSST-R will generate a minimal rooted IMC clearing node search
of $\mathbf{G}$ with probability greater than or equal to $1-\alpha_{1}^{M}$
where $M$ is the number of iterations and $\alpha_{1}\in\left(  0,1\right)  $.

\item Uniform GSST-LW will generate a minimal rooted IMC clearing node search
of $\mathbf{G}$ with probability greater than or equal to $1-\alpha_{2}^{M}$
where $M$ is the number of iterations and $\alpha_{2}\in\left(  0,1\right)  $.
\end{enumerate}
\end{theorem}

\begin{proof}
The proof is exactly the same for GSST-R and GSST-LW, so we only prove the
first one. $\mathbf{G}$ has at least one minimal rooted IMC node \ clearing
search $\mathbf{S}$ of $\mathbf{G}$. Let $\left(  \mathbf{T}_{0}%
,\mathbf{T}_{1},...,\mathbf{T}_{N}\right)  $ be the tree sequence generated by
$\mathbf{S}$. By Lemma \ref{prp0508}, GSST-R has a nonzero probability, call
it $\beta_{1}$, of generating \emph{in a single iteration} a search
$\mathbf{S}^{\prime}$ with the same tree sequence as $\mathbf{S}$. Then, by
Lemma \ref{prp0509},
\[
\overline{sn}\left(  \mathbf{S}\right)  \geq\overline{sn}\left(
\mathbf{S}^{\prime}\right)  .
\]
Since $\mathbf{S}$ is minimal, $\overline{sn}\left(  \mathbf{S}\right)
=\overline{sn}\left(  \mathbf{S}^{\prime}\right)  $ and so $\mathbf{S}%
^{\prime}$ is minimal too. Now, the probability of \emph{not} generating
$\mathbf{S}^{\prime}$ in a single iteration is $\alpha_{1}=1-\beta_{1}$; and
the probability of \emph{not} generating $\mathbf{S}^{\prime}$ in $M$
iterations is $\alpha_{1}^{M}=\left(  1-\beta_{1}\right)  ^{M}$, while the
probability of generating $\mathbf{S}^{\prime}$ in $M$ iterations is
$1-\alpha_{1}^{M}$.
\end{proof}

Finally, let us mention that the GSST algorithm can be modified to produce an
edge- rather than node-clearing search using the following theorem.

\begin{center}

\end{center}

\begin{theorem}
\label{prp0511}Given a graph $\mathbf{G}$ and an IMC node clearing search
$\mathbf{S}$ using $K$ searchers, there is an edge clearing search
$\mathbf{S}^{\prime}$ using either $K$ or $K+1$ searchers.
\end{theorem}

\begin{proof}
Suppose that $\mathbf{G}$ contains $L_{0}$ edges, that the length (i.e.,
number of moves)\ of $\mathbf{S}$ is $t_{fin}$ and that $\overline{sn}\left(
\mathbf{S}\right)  =K$. The new search $\mathbf{S}^{\prime}$ will consist of
the $\mathbf{S}$ moves (executed at integer times $t=1,2,...,t_{fin}$)
combined with the moves of an extra searcher, the \textquotedblleft edge
cleaner\textquotedblright, who will only (if at all)\ move at
\emph{fractional} time steps of the form $t+\frac{l}{3L_{0}+1}$,
$l=1,2,...,3L_{0}$ (fractional times are introduced to preserve the
\textquotedblleft alignment\textquotedblright\ of $\mathbf{S}$ and
$\mathbf{S}^{\prime}$, i.e. to ensure $\mathbf{S}\left(  t\right)
=\mathbf{S}^{\prime}\left(  t\right)  $ at integer times; of course, once
$\mathbf{S}^{\prime}$ has been obtained, the time scale can be renormalized,
so that \emph{all }moves occur at integer times). We will use the notation
$t^{-}=t-\frac{1}{3L_{0}+1}$.

We will describe the moves of of the edge cleaner on a step-by-step basis, for
$t=1,2,...,t_{fin}$ in such a manner that at the same time we will complete an
inductive proof of the fact that%
\[
\text{for }t=1,2,...,t_{fin}+1:\qquad V_{N}^{C}\left(  t-1\right)  =V_{E}%
^{C}\left(  t-1\right)  ,\quad V_{N}^{F}\left(  t-1\right)  =V_{E}^{F}\left(
t-1\right)  ,\quad E_{N}^{C}\left(  t^{-}\right)  =E_{E}^{C}\left(
t^{-}\right)  ;
\]
recall that $V_{N}^{F}$ is the frontier, i.e., the set of n-clear nodes
connected to n-dirty nodes (and similarly for $V_{E}^{F}$) and that the
frontier nodes are always guarded in both the node and edge game.

For $t=1$ we have $V_{N}^{C}\left(  0\right)  =V_{E}^{C}\left(  0\right)
=\emptyset$, $V_{N}^{F}\left(  0\right)  =V_{E}^{F}\left(  0\right)
=\emptyset$, and the edge cleaner is not used, so $E_{N}^{C}\left(
1^{-}\right)  =E_{N}^{C}\left(  0\right)  =E_{E}^{C}\left(  0\right)
=E_{E}^{C}\left(  1^{-}\right)  $. Now suppose that
\begin{equation}
\text{for }s=1,2,...,t\qquad V_{N}^{C}\left(  t-1\right)  =V_{E}^{C}\left(
t-1\right)  ,\quad V_{N}^{F}\left(  t-1\right)  =V_{E}^{F}\left(  t-1\right)
,\quad E_{N}^{C}\left(  t^{-}\right)  =E_{E}^{C}\left(  t^{-}\right)
\label{eq5900}%
\end{equation}
and consider $s=t+1$. Let the $\left(  t+1\right)  $-th move of $\mathbf{S}$
be $u\rightarrow v$. We consider three cases.

\begin{enumerate}
\item \textbf{Case I}. If $u$ is an interior node (i.e., \ a non-frontier
node:\ $u\in V_{N}^{C}\left(  t-1\right)  -V_{N}^{F}\left(  t-1\right)
=V_{E}^{C}\left(  t-1\right)  -V_{E}^{F}\left(  t-1\right)  $ ), then $uv$ is
an edge of $\mathbf{G}_{N}^{C}\left(  t\right)  =\mathbf{G}_{E}^{C}\left(
t\right)  $ and no new nodes/edges are cleared, either in the node or edge
game. No path becomes unguarded, and so no recontaminaiton is possible either.
Hence
\begin{align}
V_{N}^{C}\left(  t\right)   &  =V_{N}^{C}\left(  t-1\right)  =V_{E}^{C}\left(
t-1\right)  =V_{E}^{C}\left(  t\right)  ,\label{eq5901}\\
V_{N}^{F}\left(  t\right)   &  =V_{N}^{F}\left(  t-1\right)  =V_{E}^{F}\left(
t-1\right)  =V_{E}^{F}\left(  t\right)  .\nonumber
\end{align}
The edge cleaner is not used, hence also
\begin{equation}
E_{N}^{E}\left(  \left(  t+1\right)  ^{-}\right)  =E_{N}^{E}\left(
t^{-}\right)  =E_{E}^{E}\left(  t^{-}\right)  =E_{N}^{E}\left(  \left(
t+1\right)  ^{-}\right)  . \label{eq5902}%
\end{equation}

\item \textbf{Case II}. If $u$ is a frontier node ($u\in V_{N}^{F}\left(
t-1\right)  =V_{E}^{F}\left(  t-1\right)  $ ) which contains a \emph{single
}searcher at $t-1$, then $u$ becomes unguarded at $t$. Clearly
\begin{equation}
V_{N}^{C}\left(  t\right)  =V_{N}^{C}\left(  t-1\right)  \cup\left\{
v\right\}  . \label{eq5903}%
\end{equation}
It is also easy to see that there is no edge $uy$ with $y\in V_{N}^{D}\left(
t-1\right)  =V_{E}^{D}\left(  t-1\right)  $ and $y\neq v$: if such an edge
existed, then we would have $y\in V_{N}^{D}\left(  t\right)  $ ($y$ was not
entered at $t$) and so $u\in V_{N}^{D}\left(  t\right)  $; but $u\in V_{N}%
^{D}\left(  t-1\right)  $ (it was guarded at $t-1$) and so node monotonicity
of $\mathbf{S}$ would be violated.

Take any edge $pq\in E_{E}^{C}\left(  t^{-}\right)  =E_{N}^{C}\left(
t^{-}\right)  $. Edge $pq$ cannot be recontaminated in the node game (no node
was recontaminated). Edge $pq$ cannot be recontaminated in the edge game
either; for this to happen there must exist an e-unguarded path from $pq$ to
some $xz\in E_{E}^{D}\left(  t^{-}\right)  $; but all such paths must go
through $u$ (no other node became unguarded at $t$)\ and hence must include
$uv$ (the only edge e-dirty at $t^{-}$ and incident on $u$) but $v$ is guarded
at $t$. Hence no edge is e-dirtied at $t$ and so no node is e-dirtied either.
In other words%
\begin{equation}
V_{E}^{C}\left(  t\right)  =V_{E}^{C}\left(  t-1\right)  \cup\left\{
v\right\}  . \label{eq5904}%
\end{equation}
which, together with (\ref{eq5903})\ shows that $V_{N}^{C}\left(  t\right)
=V_{E}^{C}\left(  t\right)  $. Also, $V_{N}^{F}\left(  t\right)  =V_{E}%
^{F}\left(  t\right)  $; in both the node and edge game, $u$ was removed from
the frontier and $v$ was \emph{possibly} added to it. In short, we have
established (\ref{eq5901})\ \ for Case II as well.

From the previous remarks we know that, in the edge game and at time $t$, no
edge was recontaminated and edge $uv$ was e-cleared. I.e.,%
\[
E_{E}^{C}\left(  t\right)  =E_{E}^{C}\left(  t^{-}\right)  \cup\left\{
vu\right\}  .
\]
In the node game, on the other hand,
\begin{equation}
E_{N}^{C}\left(  t\right)  =E_{N}^{C}\left(  t^{-}\right)  \cup\left\{
vx_{0},vx_{1},...,vx_{L}\right\}  =E_{N}^{C}\left(  \left(  t+1\right)
^{-}\right)  \label{eq5905}%
\end{equation}
where $L\geq0$, $x_{0}=u$ and there \emph{may }exist other nodes $x_{1}%
,x_{2},...,x_{L}$ which must (a)\ be neighbors of $v$ and (b)\ belong to
$V_{N}^{F}\left(  t-1\right)  \subseteq V_{N}^{C}\left(  t-1\right)  $. Hence
$x_{1},...,x_{L}$ were guarded at $t-1$ and remain so at $t$. In other words,
at time $t$ the edges $vx_{1},...,vx_{L}$ have both endpoints guarded and are
n-clear but e-dirty. Now we invoke the edge cleaner who (at times $t+\frac
{1}{3L_{0}+1},t+\frac{2}{3L_{0}+1},...$ ) moves to $v$ and then performs the
moves $v\rightarrow x_{1}$, $x_{1}\rightarrow v$, $v\rightarrow x_{2}$,
$x_{2}\rightarrow v$, ... , $x_{L}\rightarrow v$. This entire sequence can be
performed in no more than $3L_{0}$ moves, so at $t=\left(  t+1\right)  ^{-}$
the edges $vx_{0},vx_{1},...,vx_{L}$ have been e-cleared and so
\begin{equation}
E_{E}^{C}\left(  \left(  t+1\right)  ^{-}\right)  =E_{N}^{C}\left(
t^{-}\right)  \cup\left\{  vx_{0},vx_{1},...,vx_{L}\right\}  . \label{eq5906}%
\end{equation}
Combining (\ref{eq5905})\ and (\ref{eq5906})\ we get $E_{N}^{C}\left(  \left(
t+1\right)  ^{-}\right)  =E_{E}^{C}\left(  \left(  t+1\right)  ^{-}\right)  $.
In short, we have established (\ref{eq5902})\ \ for Case II as well.%
\[
\qquad
\]
\medskip

\item \textbf{Case III}. The final case to examine is when $u$ is a frontier node
which contains more than one searcher at $t-1$. We omit a detailed treatment
because the proof combines elements from the previous two cases; namely,
recontamination does not happen (for the same reasons as in Case I) bu the
edge cleaner may possibly be required (as in Case II).
\end{enumerate}

Hence, in all three cases considered, starting from (\ref{eq5900})\ we have
established (\ref{eq5901})\ and (\ref{eq5902}). Hence we can complete the
induction up to time $t=t_{fin}+1$ $\ $which means%
\begin{equation}
V_{N}^{C}\left(  t_{fin}\right)  =V_{E}^{C}\left(  t_{fin}\right)  ,\quad
E_{N}^{C}\left(  \left(  t_{fin}+1\right)  ^{-}\right)  =E_{E}^{C}\left(
\left(  t_{fin}+1\right)  ^{-}\right)  . \label{eq5907}%
\end{equation}
Eqs.(\ref{eq5907})\ imply that all nodes are e-cleared at $t_{fin}$ but a few
extra steps may be required to e-clear all edges (by time $\left(
t_{fin}+1\right)  ^{-}$ at most).
\end{proof}

\begin{center}

\end{center}

\section{Graph Search Experiments}

\label{sec06}

In this section we evaluate the performance of the GSST\ algorithm by
numerical experiments. Some of these experiments involve specific graphs
(Section \ref{sec0601}) and others involve families of graphs (Sections
\ref{sec0602} and \ref{sec0603} -- in which case we present average results).
We use ten \emph{variants }of GSST, obtained by using two different methods of
spanning tree generation (uniform and DFS) and five methods of edge traversal
(GSST-L, GSST-R, GSST-LR, GSST-LW, and GSST-LD).

\subsection{Experiments using Individual Graphs}

\label{sec0601}

\subsubsection{Simple Graph}

The first graph we have used appears in Fig.\ref{fig02}. This is a relatively
simple graph (with IMC$\ $node search number $s_{N}^{imc}=$3)\ which we use to
illustrate the basic principles of GSST\ operation.

\begin{figure}[h]
\centering
\scalebox{0.6}{\includegraphics{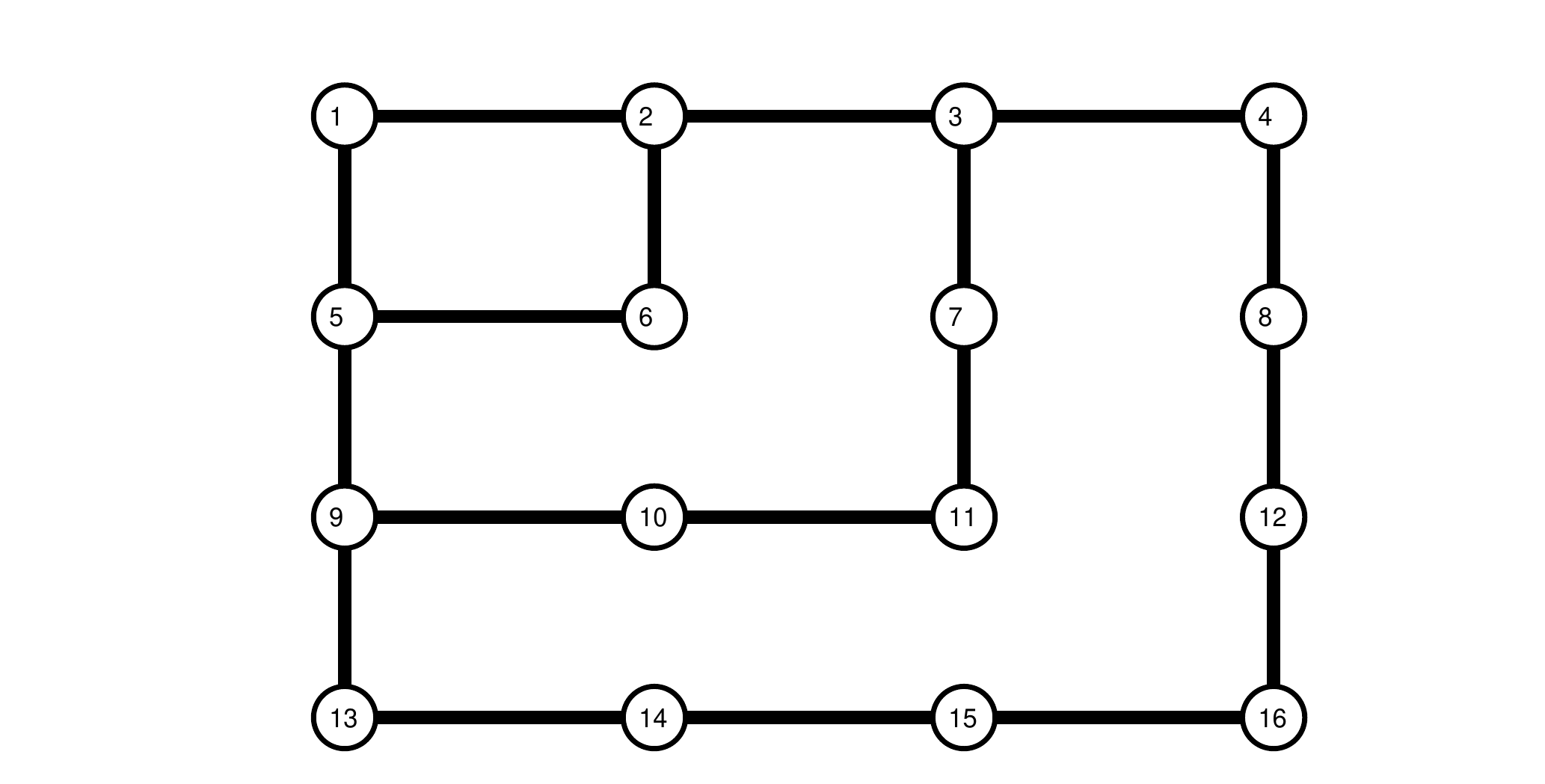}}\caption{A simple graph.}%
\label{fig02}%
\end{figure}

We node-search the graph using the ten GSST variants, each with $M=5000$
generated spanning trees; to each spanning tree corresponds a search
$\mathbf{S}_{m}$, $m=1,2,...,5000$. In Table 2 we list for each GSST\ variant:
(a)$\ \min_{m=1,2,...,M}\overline{sn}\left(  \mathbf{S}_{m}\right)  $ (i.e.,
the minimum number of searchers achieved by the specific combination),
(b)\ the proportion of minimal searches out of the the $M$ total searches and
(c)\ the time\footnote{All computations were performed by the gsearch.exe
program (see Appendix \ref{secD}), on a PC\ with Intel Dual Core E7500 CPU,
running at 2.93 GHz with 3 GB RAM; computation time is reported in seconds,
unless otherwise indicated.} (in sec)\ required to run the $M$ searches.

\begin{center}%
\begin{tabular}
[c]{|l|l|l|l|l|l|l|}\hline
& \multicolumn{3}{|c}{\textbf{Uniform ST generation}} &
\multicolumn{3}{|c|}{\textbf{DFS ST generation}}\\\hline
\textbf{Edge Traversal} & Min$\ $ & Prop. of min & Time & Min & Prop. of min &
Time\\\hline
GSST-L & \multicolumn{1}{|r|}{3} & \multicolumn{1}{|r|}{0.1562} &
\multicolumn{1}{|r|}{0.203125} & \multicolumn{1}{|r|}{3} &
\multicolumn{1}{|r|}{0.1804} & \multicolumn{1}{|r|}{0.203125}\\\hline
GSST-R & \multicolumn{1}{|r|}{3} & \multicolumn{1}{|r|}{0.2031} &
\multicolumn{1}{|r|}{0.171875} & \multicolumn{1}{|r|}{3} &
\multicolumn{1}{|r|}{0.2848} & \multicolumn{1}{|r|}{0.203125}\\\hline
GSST-LR & \multicolumn{1}{|r|}{3} & \multicolumn{1}{|r|}{0.2656} &
\multicolumn{1}{|r|}{0.265625} & \multicolumn{1}{|r|}{3} &
\multicolumn{1}{|r|}{0.2096} & \multicolumn{1}{|r|}{0.265625}\\\hline
GSST-LW & \multicolumn{1}{|r|}{3} & \multicolumn{1}{|r|}{0.2292} &
\multicolumn{1}{|r|}{0.218750} & \multicolumn{1}{|r|}{3} &
\multicolumn{1}{|r|}{0.2854} & \multicolumn{1}{|r|}{0.250000}\\\hline
GSST-LD & \multicolumn{1}{|r|}{3} & \multicolumn{1}{|r|}{0.2031} &
\multicolumn{1}{|r|}{0.218750} & \multicolumn{1}{|r|}{3} &
\multicolumn{1}{|r|}{0.3020} & \multicolumn{1}{|r|}{0.234375}\\\hline
\end{tabular}

\end{center}

\noindent\textbf{Table 2.} Node-clearing the \textquotedblleft simple
graph\textquotedblright\ by the various GSST variants: minimum number of
searchers attained and proportion of minimal solutions; number of spanning
trees generated is $M=5\cdot10^{3}.$

\bigskip

The true node search number (i.e., 3) has been found by every variant of GSST.
Generally, the DFS\ variants perform better than the uniform ones, as can be
seen by the higher proportion of minimal solutions achieved. A better
understanding of the distribution of the number of searchers required by each
search can be obtained by looking at the \emph{histogram} of the distribution;
one such histogram (for the variant with BH edge traversal and uniform
spanning tree generation) is plotted in Fig.\ref{fig03}. We can see that this
simple graph has a high proportion of spanning trees which yield minimal schedules.

\begin{figure}[h]
\centering\scalebox{0.35}{\includegraphics{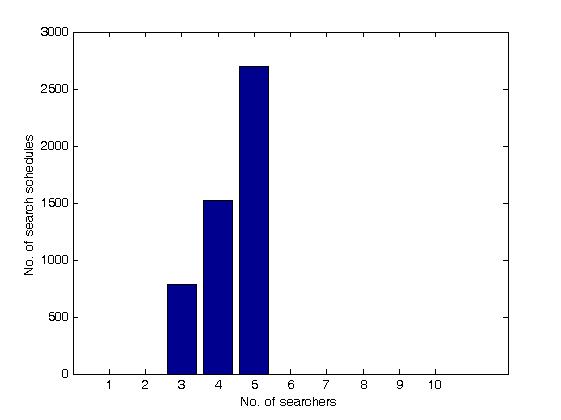}}\caption{The
histogram of the distribution of minimum number of searchers required to node
clear the \textquotedblleft simple graph\textquotedblright\ by the uniform
GSST-L\ variant.}%
\label{fig03}%
\end{figure}

An additional tool to evaluate the anytime performance of GSST is the plot of
$R\left(  m\right)  =\min_{i\leq m}\overline{sn}\left(  \mathbf{S}_{i}\right)
$ (the minimum node clearing number of searchers achieved by the first $m$
searches). $R\left(  m\right)  $ is decreasing with $m$. The overall minimum
achieved by GSST is $R\left(  M\right)  $ (having tried $M$ searches). If this
minimum is achieved for a small value of $m$, then the minimal solution has
been obtained quickly. A graph of $R\left(  m\right)  $ (for the variant with
labeled edge traversal and uniform spanning tree generation) appears in
Fig.\ref{fig04}; as can be seen a minimal solution (clearing the graph with
three searchers) is achieved by the seventh computed search, approximately at
time $t=0.203125\cdot7/5000=\allowbreak2.\,\allowbreak843\,8\times10^{-4}\ $sec.

\begin{figure}[h]
\centering
\scalebox{0.6}{\includegraphics{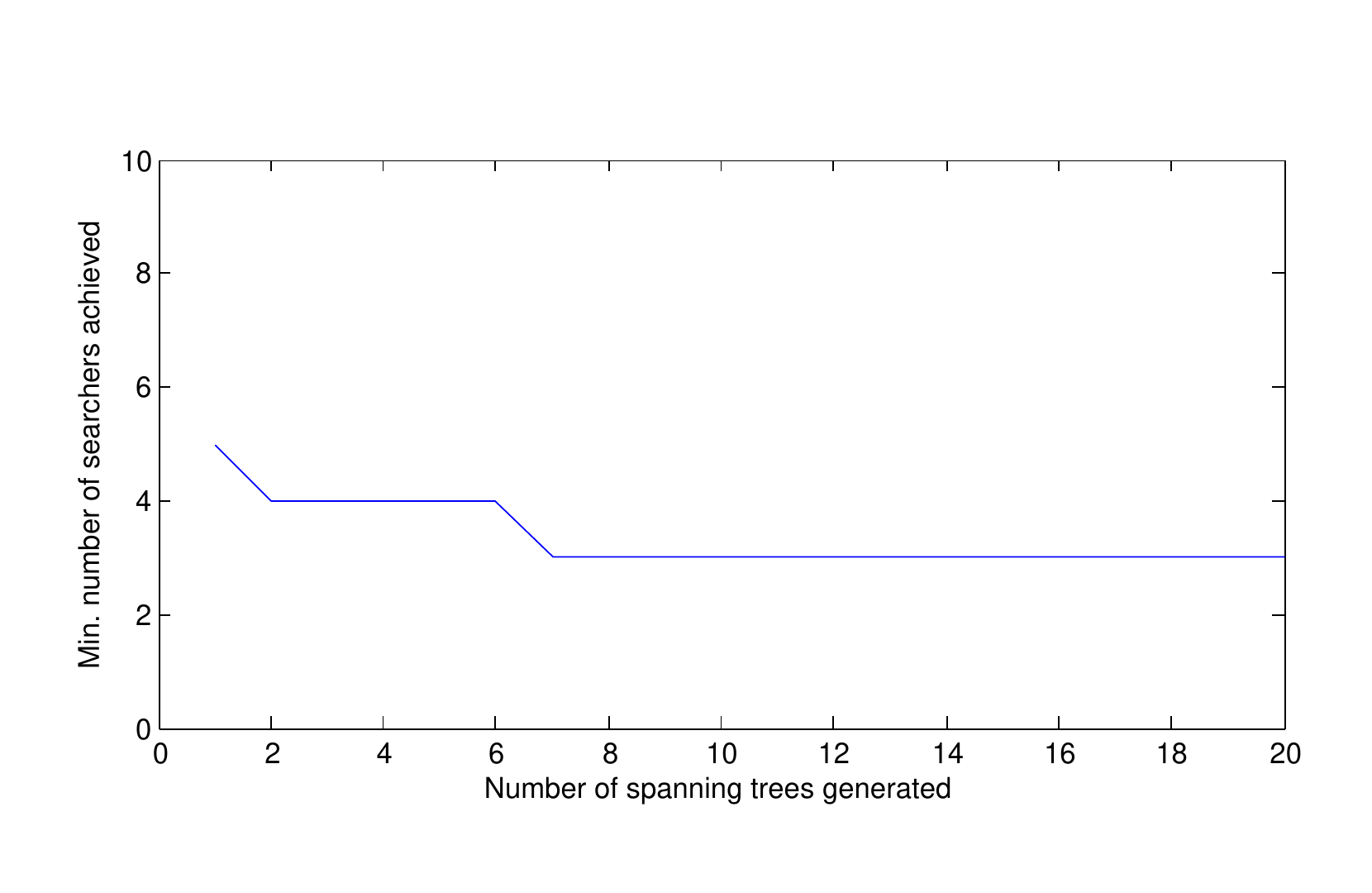}}\caption{Plot of $R\left(
m\right)  $ vs. $m$ for the \textquotedblleft simple graph\textquotedblright%
\ and the uniform GSST-L\ variant. The $m$ axis is truncated at $m=20$, since
the rest of the plot shows no change.}%
\label{fig04}%
\end{figure}

The \textquotedblleft simple graph\textquotedblright\ is simple enough to have
a relatively small number of spanning trees (namely 272, as computed by
Kirchoff's theorem \cite{HarrisKirchoff} ) and so we can also run GSST on
\emph{all }spanning trees (exhaustive enumeration). The computation takes
0.03125 (for the GSST-L variant) sec and shows that 83 out of the 272 spanning
trees (\emph{rooted at node} 1) yield minimal searches\footnote{Note that all
GSST variants, except GSST-L, use a randomized element in the order of edge
traversal -- hence multiple runs with the same spanning tree can yield
different searches.}.

\subsubsection{Tree/Grid Graph}

The second graph we have used appears in Fig.\ref{fig05}. It consists of a
\textquotedblleft root\textquotedblright\ node and two branches under it; the
left branch is a tree and the right one a grid; hence the name
\textquotedblleft tree/grid". The graph has $s_{N}^{imc}=$4 and we have
deliberately designed it to \textquotedblleft trick\textquotedblright\ the
GSST\ algorithm. For the sake of definiteness consider GSST-L. If the root of
the search is node 6, then GSST-L \ will find a four searcher node clearing
IMC\ schedule. However, if the root is node 1, then GSST-L\ will only find a
five searcher schedule, even after enumerating all spanning trees. A four
searcher IMC\ node clearing schedule is possible from either starting node;
but it requires the use of a non-Barriere edge traversal (for example, one
produced by the GSST-R variant); but GSST-L will always first send the
searchers down the right branch (towards the grid); actually, going first to
the left branch, towards the tree, is better (i.e. yields a four searcher schedule).

\begin{figure}[h]
\centering
\scalebox{0.6}{\includegraphics{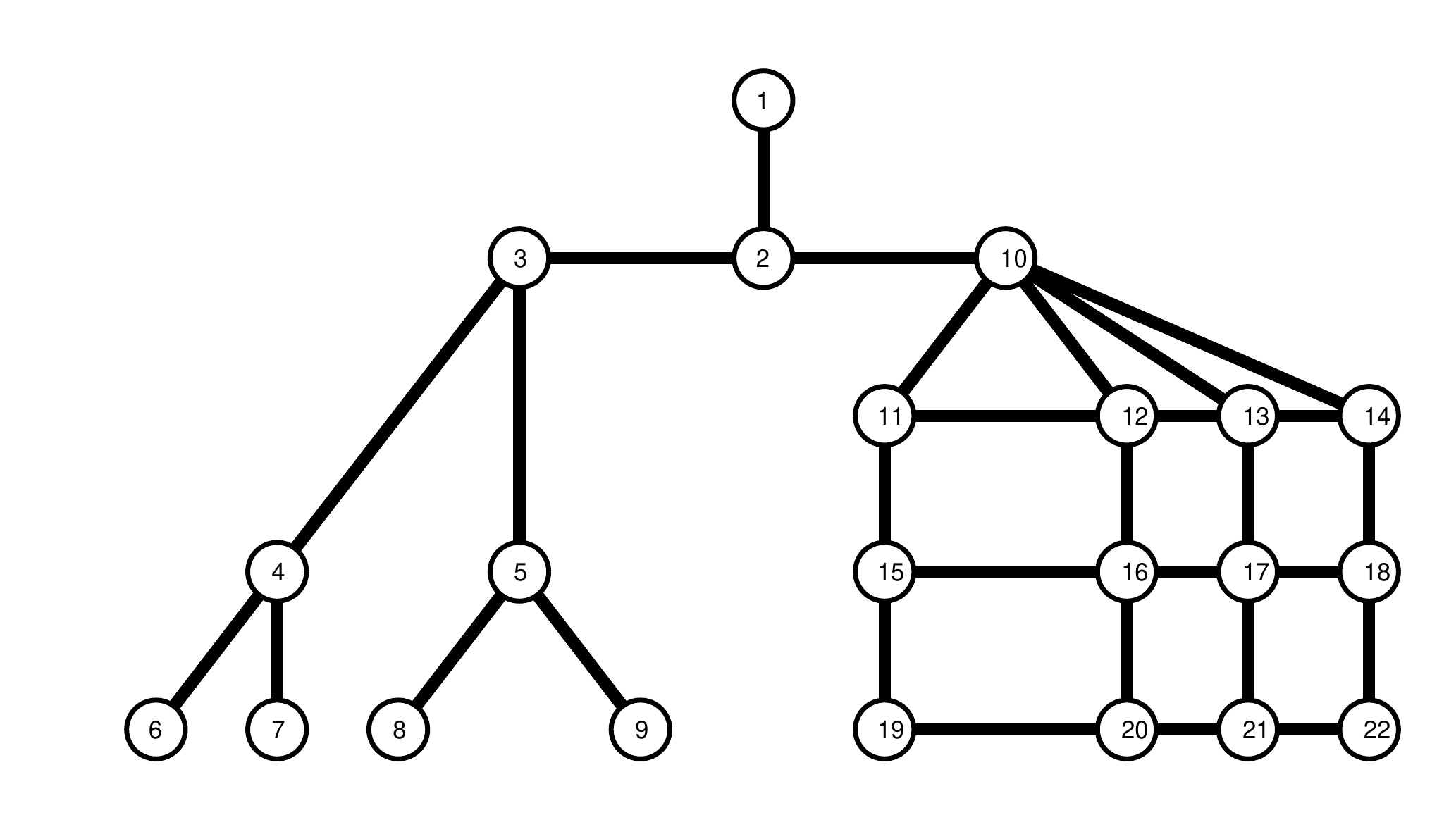}}\caption{The \textquotedblleft
tree/grid\textquotedblright\ graph.}%
\label{fig05}%
\end{figure}

Table 3 summarizes the results of our experiment, using the ten variants of
GSST and $M=5\cdot10^{4}$ spanning trees per variant. Note that in this table
(as in all others) the starting node is selected randomly. The uniform
variants are able to generate minimal search schedules, but not the DFS\ ones.

\begin{center}%
\begin{tabular}
[c]{|l|l|l|l|l|l|l|}\hline
& \multicolumn{3}{|c}{\textbf{Uniform ST generation}} &
\multicolumn{3}{|c|}{\textbf{DFS ST generation}}\\\hline
\textbf{Edge Traversal} & Min & Prop. of min & Time & Min & Prop. of min &
Time\\\hline
GSST-L & \multicolumn{1}{|r|}{4} & \multicolumn{1}{|r|}{0.00054} &
\multicolumn{1}{|r|}{5.921875} & \multicolumn{1}{|r|}{5} &
\multicolumn{1}{|r|}{0.06380} & \multicolumn{1}{|r|}{6.687500}\\\hline
GSST-R & \multicolumn{1}{|r|}{4} & \multicolumn{1}{|r|}{0.00328} &
\multicolumn{1}{|r|}{3.734375} & \multicolumn{1}{|r|}{5} &
\multicolumn{1}{|r|}{0.06956} & \multicolumn{1}{|r|}{4.390625}\\\hline
GSST-LR & \multicolumn{1}{|r|}{4} & \multicolumn{1}{|r|}{0.00076} &
\multicolumn{1}{|r|}{8.328125} & \multicolumn{1}{|r|}{5} &
\multicolumn{1}{|r|}{0.06976} & \multicolumn{1}{|r|}{8.765625}\\\hline
GSST-LW & \multicolumn{1}{|r|}{4} & \multicolumn{1}{|r|}{0.00206} &
\multicolumn{1}{|r|}{4.812500} & \multicolumn{1}{|r|}{5} &
\multicolumn{1}{|r|}{0.07252} & \multicolumn{1}{|r|}{5.609375}\\\hline
GSST-LD & \multicolumn{1}{|r|}{4} & \multicolumn{1}{|r|}{0.00232} &
\multicolumn{1}{|r|}{4.968750} & \multicolumn{1}{|r|}{5} &
\multicolumn{1}{|r|}{0.09376} & \multicolumn{1}{|r|}{5.656250}\\\hline
\end{tabular}

\end{center}

\noindent\textbf{Table 3. }Node-clearing the \textquotedblleft
tree/grid\textquotedblright\ by the various GSST variants: minimum number of
searchers attained and proportion of minimal solutions; number of spanning
trees generated is $M=5\cdot10^{4}.$\textbf{ }

\bigskip

In this case the uniform variants find four-searcher clearing schedules and
hence outperform the DFS\ generated ones, which can only clear the graph with
five or more searchers\footnote{We have experimented with larger values of $M$
(e.g., $M=10^{6}$) but the DFS\ variants still cannot achieve a four searcher
clearing schedule.}. Note also the lower proportion of minimal solutions, even
for the uniform variants. For example, uniform GSST-L finds $0.00054\cdot
50000=\allowbreak27$ minimal solutions. This in an indication that this
graph\ is indeed harder than the \textquotedblleft simple\textquotedblright%
\ one, at least for the GSST\ algorithm. Also, the tree/grid graph\ has 31529
spanning trees (and correspondingly many cycles) which is an additional
indication that it is (much) harder to search than the simple graph, which has
272 spanning trees. The small number of minimal solutions can also be
appreciated by looking at the histogram (for uniform GSST-L\ variant it is
plotted in Fig.\ref{fig06}).

\begin{figure}[h]
\centering\scalebox{0.35}{\includegraphics{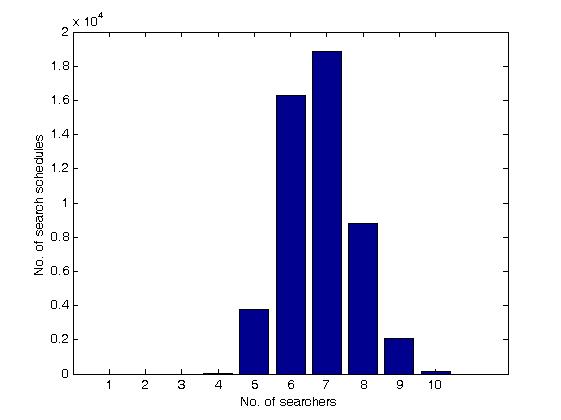}}\caption{The
histogram of the distribution of minimum number of searchers required to node
clear the \textquotedblleft hard graph\textquotedblright\ by the uniform
GSST-L\ variant. There is a small, barely noticeable, bar at $k=4$,
corresponding to the four-searcher clearing schedules.}%
\label{fig06}%
\end{figure}

However, despite the small proportion of minimal solutions, the first one is
always found after a relatively small number of iterations. For example, the
uniform GSST-L variant finds a four searcher schedule with the 420-th spanning
tree generated (out of a total of $5\cdot10^{4}$ spanning trees),
approximately at time $t=5.921875\cdot420/50000=\allowbreak4.\,\allowbreak
974\,4\times10^{-2}\ $sec.

\subsubsection{NSH Graph}

The next graph we use has been obtained by discretization of an actual
floorplan, namely the first floor of the Newel-Simon building in the Carnegie
Mellon University \ campus. In Fig.\ref{fig08} we present the actual floorplan
and its discretization; in Fig.\ref{fig09} we present the resulting graph (the
node numbers in Fig.\ref{fig09} correspond to the cell numbers in
Fig.\ref{fig08}).

\begin{figure}[h]
\centering
\scalebox{0.6}{\includegraphics{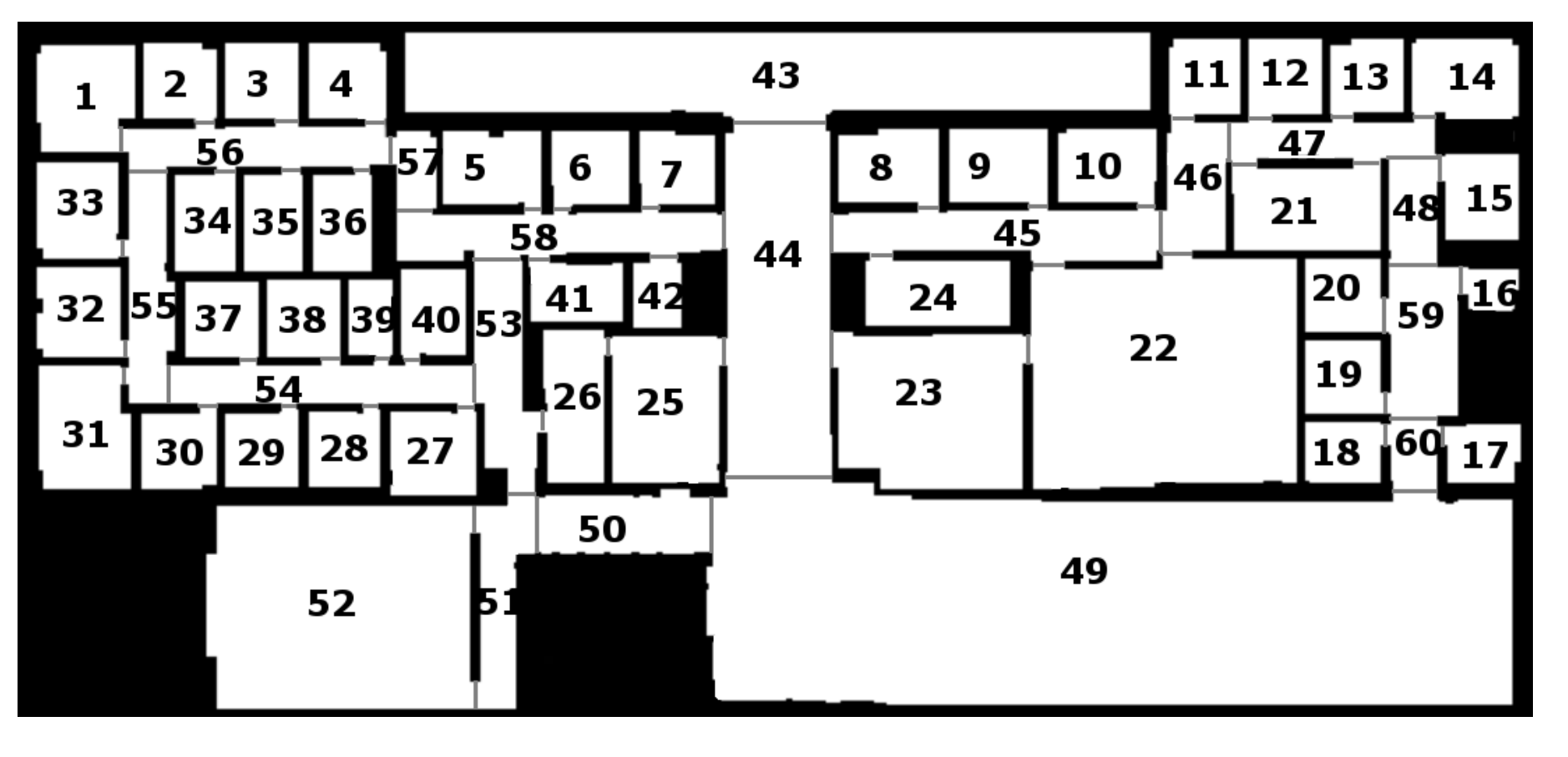}}\caption{The NSH floorplan.}%
\label{fig08}%
\end{figure}

\begin{figure}[h]
\centering
\scalebox{0.4}{\includegraphics{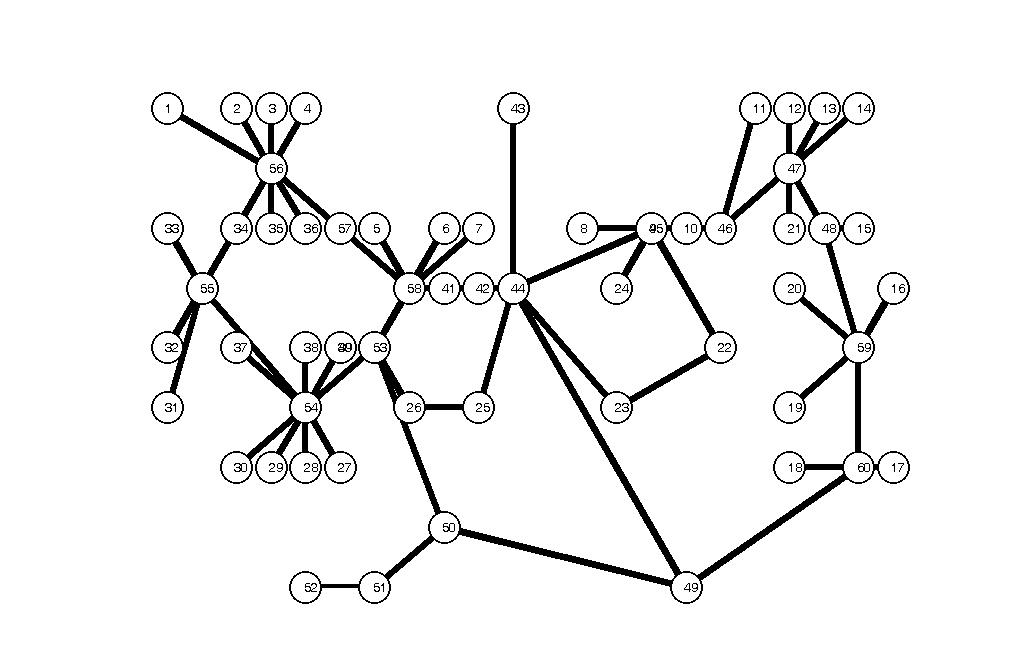}}\caption{The NSH graph.}%
\label{fig09}%
\end{figure}

\clearpage

This graph, with 60 nodes and 64\ edges, appears to be more complex than the
previous ones; however it has a relatively small number of spanning trees
(namely 3604, as computed by Kirchoff's theorem). By visual inspection it
appears very unlikely that the graph can be node-cleared with two searchers;
on the other hand, as will be seen presently, GSST can compute three-searcher
clearing schedules. Hence we conjecture that the search number $s_{N}^{imc}$
is three.

Applying the \textquotedblleft standard\textquotedblright\ variants of
GSST\ and using $M=5\cdot10^{4}$ spanning trees per variant we obtain the
results of Table 4.

\begin{center}%
\begin{tabular}
[c]{|l|l|l|l|l|l|l|}\hline
& \multicolumn{3}{|c}{\textbf{Uniform ST generation}} &
\multicolumn{3}{|c|}{\textbf{DFS ST generation}}\\\hline
\textbf{Edge Traversal} & Min & Prop. of min & Time & Min & Prop. of min &
Time\\\hline
GSST-L & \multicolumn{1}{|r|}{3} & \multicolumn{1}{|r|}{0.00224} &
\multicolumn{1}{|r|}{14.265625} & \multicolumn{1}{|r|}{3} &
\multicolumn{1}{|r|}{0.01038} & \multicolumn{1}{|r|}{14.578125}\\\hline
GSST-R & \multicolumn{1}{|r|}{3} & \multicolumn{1}{|r|}{0.00002} &
\multicolumn{1}{|r|}{9.406250} & \multicolumn{1}{|r|}{3} &
\multicolumn{1}{|r|}{0.00006} & \multicolumn{1}{|r|}{10.140625}\\\hline
GSST-LR & \multicolumn{1}{|r|}{3} & \multicolumn{1}{|r|}{0.00478} &
\multicolumn{1}{|r|}{21.937500} & \multicolumn{1}{|r|}{3} &
\multicolumn{1}{|r|}{0.01026} & \multicolumn{1}{|r|}{21.390625}\\\hline
GSST-LW & \multicolumn{1}{|r|}{3} & \multicolumn{1}{|r|}{0.00328} &
\multicolumn{1}{|r|}{11.546875} & \multicolumn{1}{|r|}{3} &
\multicolumn{1}{|r|}{0.00608} & \multicolumn{1}{|r|}{12.593750}\\\hline
GSST-LD & \multicolumn{1}{|r|}{3} & \multicolumn{1}{|r|}{0.00544} &
\multicolumn{1}{|r|}{12.625000} & \multicolumn{1}{|r|}{3} &
\multicolumn{1}{|r|}{0.00880} & \multicolumn{1}{|r|}{13.578125}\\\hline
\end{tabular}

\end{center}

\noindent\textbf{Table 4. }Node-clearing the NSH\ graph\ \ by the various GSST
variants: minimum number of searchers attained and proportion of minimal
solutions; number of spanning trees generated is $M=5\cdot10^{4}.$

\bigskip

We see that all the GSST\ variants using uniform spanning tree generation
achieve the true search number, namely three. The random edge traversal
variants perform poorest of all: the uniform GSST-R finds only one minimal
search and the DFS GSST-R variant three. Generally, while the time required to
complete $M$ searches is higher for the NSH\ graph than for the
\textquotedblleft tree/grid\textquotedblright\ graph, the proportion of
correct solutions is in some cases higher (at least for the uniform variants).
A histogram of the searchers required appears in Fig.\ref{fig10}, for the
uniform GSST-L\ variant; for this variant the first minimal search schedule is
computed at step 845 (out of 50000), approximately at time $14.265625\cdot
845/50000=\allowbreak0.241\,09\ $sec.

\begin{figure}[h]
\centering\scalebox{0.35}{\includegraphics{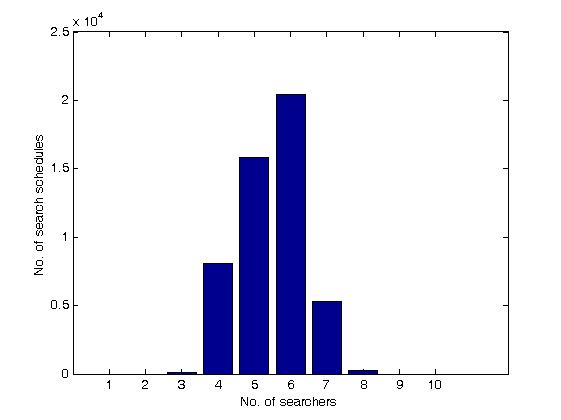}}\caption{The
histogram of the distribution of minimum number of searchers required to node
clear the NSH\ graph\ by the uniform GSST-L\ variant.}%
\label{fig10}%
\end{figure}

\subsubsection{National Art Gallery Graph}

The next graph we use has also been obtained by discretization of an actual
floorplan, namely the first floor of the National Gallery of Art, in
Washington, DC. In Fig.\ref{fig12} we present the actual floorplan and the
discretization we have used; in Fig.\ref{fig13} we plot the corresponding graph.

\begin{figure}[h]
\centering
\scalebox{0.6}{\includegraphics{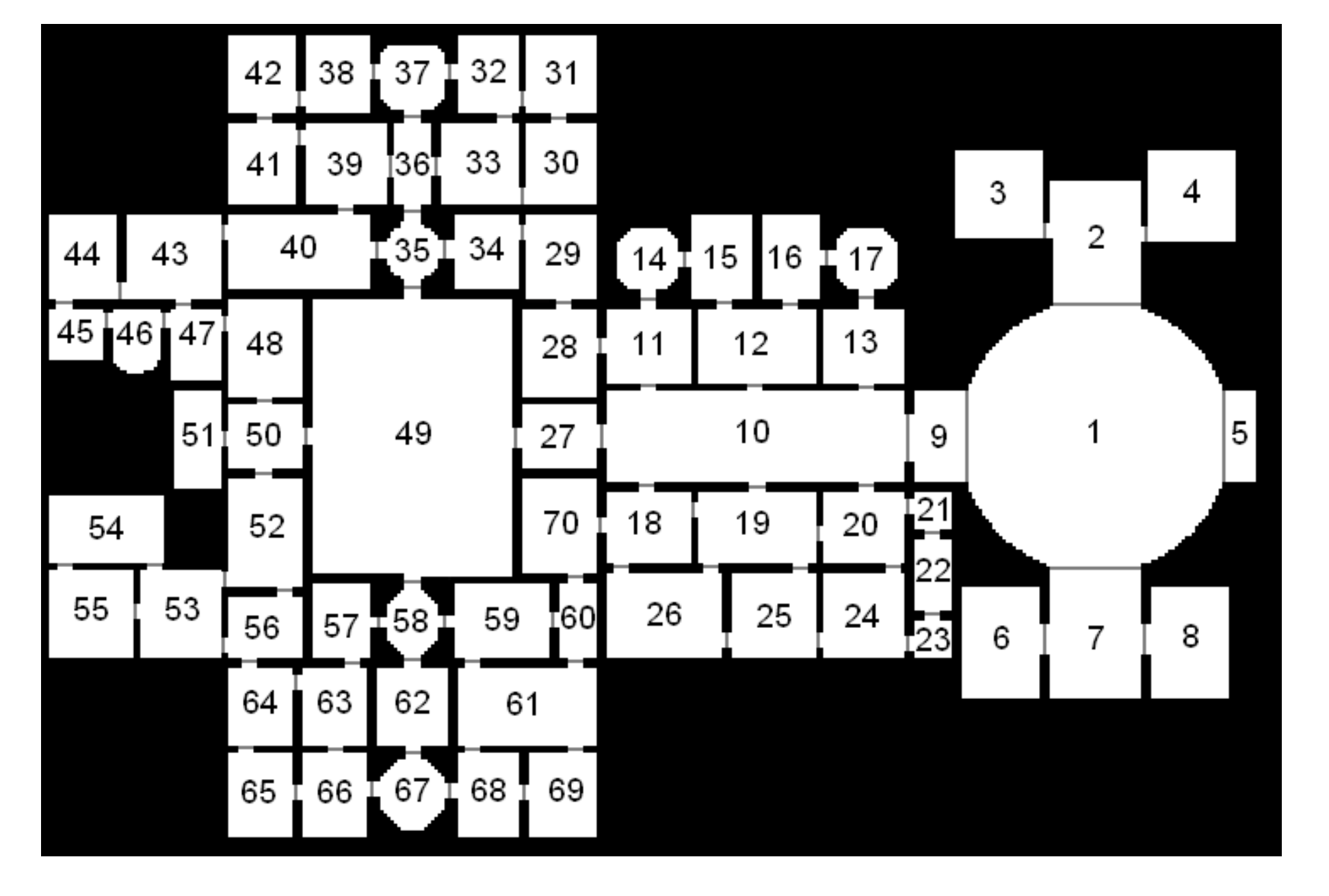}}\caption{The National Gallery
of Art floorplan.}%
\label{fig12}%
\end{figure}

\begin{figure}[h]
\centering\scalebox{0.4}{\includegraphics{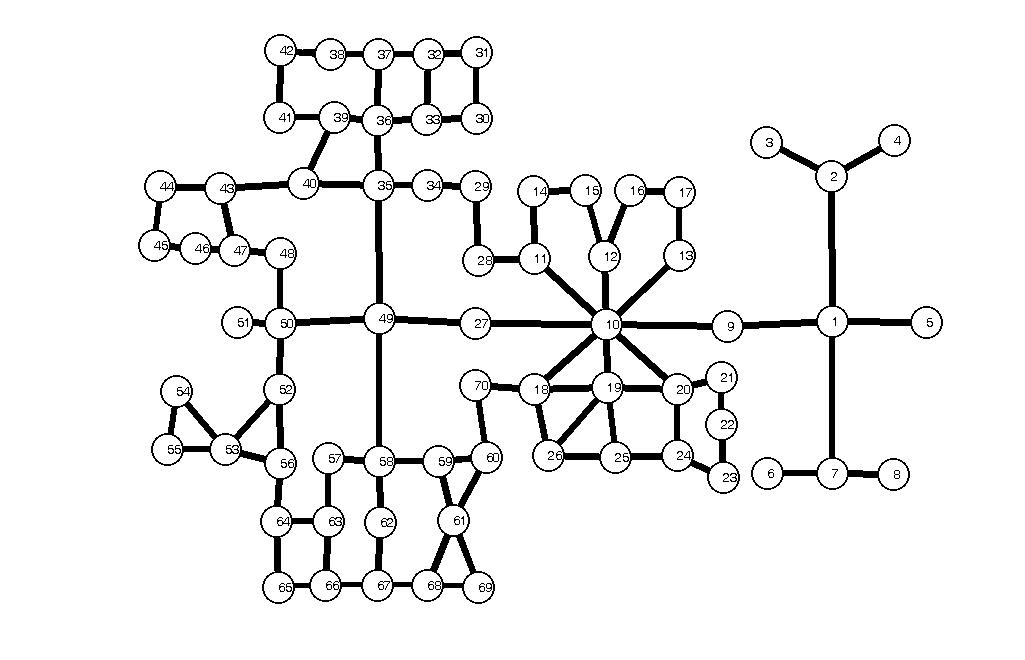}}\caption{The National
Gallery of Art graph.}%
\label{fig13}%
\end{figure}

\clearpage

This graph, has 70 nodes and 93\ edges and is more complex than any of the
previously used graphs. In particular, it is more complex than the NSH\ graph.
While the National Art Gallery graph has about 50\% more edges than the
NSH\ graph, it has a \emph{much }larger number of spanning, approximately
5.3$\cdot10^{14}$, as compared to 3604; this precludes use of the GSST with
exhaustive enumeration of the spanning trees. Also, we have no obvious way to
compute the true node search number. Nevertheless, we proceed to apply the ten
GSST\ variants using $M=1.5\cdot10^{5}$ spanning trees per combination. The
results obtained appear in Table 5.

\begin{center}%
\begin{tabular}
[c]{|l|l|l|l|l|l|l|}\hline
& \multicolumn{3}{|c}{\textbf{Uniform ST generation}} &
\multicolumn{3}{|c|}{\textbf{DFS ST generation}}\\\hline
\textbf{Edge Traversal} & Min & Prop. of min & Time & Min & Prop. of min &
Time\\\hline
GSST-L & \multicolumn{1}{|r|}{5} & \multicolumn{1}{|r|}{0.0000066} &
\multicolumn{1}{|r|}{74.378237} & \multicolumn{1}{|r|}{5} &
\multicolumn{1}{|r|}{0.0004466} & \multicolumn{1}{|r|}{77.214947}\\\hline
GSST-R & \multicolumn{1}{|r|}{6} & \multicolumn{1}{|r|}{0.0013133} &
\multicolumn{1}{|r|}{45.500000} & \multicolumn{1}{|r|}{5} &
\multicolumn{1}{|r|}{0.0004800} & \multicolumn{1}{|r|}{50.984375}\\\hline
GSST-LR & \multicolumn{1}{|r|}{5} & \multicolumn{1}{|r|}{0.0000066} &
\multicolumn{1}{|r|}{78.172631} & \multicolumn{1}{|r|}{5} &
\multicolumn{1}{|r|}{0.0002000} & \multicolumn{1}{|r|}{87.250000}\\\hline
GSST-LW & \multicolumn{1}{|r|}{5} & \multicolumn{1}{|r|}{0.0000200} &
\multicolumn{1}{|r|}{62.859375} & \multicolumn{1}{|r|}{5} &
\multicolumn{1}{|r|}{0.0007600} & \multicolumn{1}{|r|}{77.500000}\\\hline
GSST-LD & \multicolumn{1}{|r|}{5} & \multicolumn{1}{|r|}{0.0000066} &
\multicolumn{1}{|r|}{64.656250} & \multicolumn{1}{|r|}{5} &
\multicolumn{1}{|r|}{0.0009800} & \multicolumn{1}{|r|}{77.546875}\\\hline
\end{tabular}

\end{center}

\noindent\textbf{Table 5. }Node-clearing the National Art Gallery\ graph\ \ by
the various GSST variants: minimum number of searchers attained and proportion
of minimal solutions; number of spanning trees generated is $M=1.5\cdot
10^{5}.$

\bigskip

We see that all the GSST\ variants\footnote{Except for uniform GSST-R; however
5-searcher clearing schedules can also be found by this variant if it is run
for a sufficiently large $M$.} find node clearing searches with five
searchers. While we cannot be sure that the true node search number of the
National Art Gallery is five, we have been unable to find a lower search
number using any method (including extensive inspection by the authors). In
addition, the middle part of the graph resembles a five-by-six grid, which is
known to have a node search number of five (see Section \ref{sec0603}). A
histogram of the searchers required appears in \ref{fig14}, for the uniform
GSST-L variant; for this variant the first minimal search schedule is found at
step 101823 (out of a total $1.5\cdot10^{5}$). Similar results hold for the
other variants.

\begin{figure}[h]
\centering\scalebox{0.35}{\includegraphics{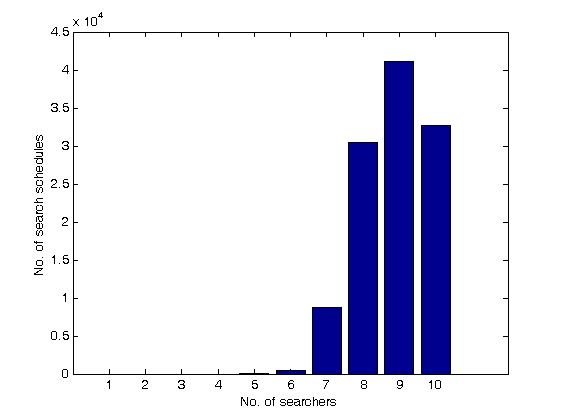}}\caption{The
histogram of the distribution of minimum number of searchers required to node
clear the National Gallery of Art \ graph\ by the uniform GSST-L\ variant.}%
\label{fig14}%
\end{figure}

The National Art Gallery graph results (as well as the NSH\ ones)\ show the
applicability of GSST on graphs that are derived from representations of real
indoor environments (which has been the main motivation for our research). The
methods that incorporate Barriere labeling improve performance on the NSH map
because it is similar to a tree (it becomes a tree if we remove only a
\emph{few} edges). In this case, Barriere labeling, which is based on trees,
helps to improve the schedules generated by GSST. The National Art Gallery, on
the other hand, is more similar to a grid, which lessens the advantage of
using Barriere labeling to guide traversal (though a significant improvement
is still obtained).

\subsection{Interval Graphs}

\label{sec0602}

The next experiment presented involves, unlike the ones of Section
\ref{sec0601}, a large number of graphs belonging to the same \emph{family}.
Our goal is to evaluate the \emph{average }performance of the GSST\ algorithm.
While, strictly speaking, the results are specific to the family of interval
graphs, they also suggest general properties of the GSST\ algorithm.

To evaluate the performance of GSST\ we must have some estimate of the actual
search number of each graph we use; then we can compare this search number
with the minimum $\overline{sn}\left(  \mathbf{S}\right)  $ achieved by GSST.
To satisfy this condition, we will work with \emph{interval graphs}. Briefly,
an interval graph is obtained from a \emph{system of intervals}, i.e. a
collection of intervals of real numbers; each interval corresponds to a graph
node and two nodes are connected by an edge iff the corresponding intervals
intersect. It is known \cite{IntervalFomin,IntervalKaplan} that the
\emph{(non-monotone, non-connected) edge search number} of an interval graph
$\mathbf{G}$ (i.e., $s_{E}\left(  \mathbf{G}\right)  $)\ is equal to its
\emph{interval width}, which is defined to be the size of the largest clique
of the graph. It is also known that the interval width of an interval graph
$\mathbf{G}$ of $N$ nodes can be computed in time O$\left(  N\right)  $
\cite{IntervalFomin,IntervalKaplan}. Of course, while we know the \emph{edge}
search number $s_{E}\left(  \mathbf{G}\right)  $ exactly, this only provides
an estimate of the IMC\ \emph{node}-search number $s_{N}^{imc}\left(
\mathbf{G}\right)  $, in which we are really interested. There is no strict
inequality connecting $s_{E}\left(  \mathbf{G}\right)  $ and $s_{N}%
^{imc}\left(  \mathbf{G}\right)  $; we just know that $s_{N}\left(
\mathbf{G}\right)  \leq s_{E}\left(  \mathbf{G}\right)  $ and $s_{N}\left(
\mathbf{G}\right)  \leq s_{N}^{imc}\left(  \mathbf{G}\right)  $. In general,
we expect that the discrepancy between $s_{E}\left(  \mathbf{G}\right)  $ and
$s_{N}^{imc}\left(  \mathbf{G}\right)  $ is not too large and hence
$s_{E}\left(  \mathbf{G}\right)  $ can be used to evaluate GSST\ performance.

We randomly generate interval systems (and the corresponding interval graphs)
by the following mechanism. First we select two parameters of the family: $N$,
the number of intervals, and $\Delta$, the average interval length. Then we
generate $N$ intervals, the $n$-th interval having its left endpoint at $n$
and the right endpoint at $n+\delta$, where $\delta$ follows an exponential
probability law $f\left(  \delta\right)  \sim e^{-\delta/\Delta}$. We form the
interval graph $\mathbf{G}$ corresponding to this interval system and compute
its interval width. We repeat the process 100 times to obtain $100$ interval
graphs; these form a family characterized by the parameters $N,\Delta$
(specified by us) and also by the average number of edges and the average
interval width. We repeat the process for five different choices of $\left(
N,\Delta\right)  $; this information is summarized in Table 6.

\begin{center}%
\begin{tabular}
[c]{|l|l|c|c|}\hline
$N$ & $\Delta$ & \textbf{Average num. of Edges} & \textbf{Average Interval
Width}\\\hline
\multicolumn{1}{|r|}{30} & \multicolumn{1}{|r|}{3} &
\multicolumn{1}{|r|}{74.22} & \multicolumn{1}{|r|}{5.49}\\\hline
\multicolumn{1}{|r|}{30} & \multicolumn{1}{|r|}{5} &
\multicolumn{1}{|r|}{113.06} & \multicolumn{1}{|r|}{7.66}\\\hline
\multicolumn{1}{|r|}{25} & \multicolumn{1}{|r|}{7} &
\multicolumn{1}{|r|}{182.01} & \multicolumn{1}{|r|}{9.51}\\\hline
\multicolumn{1}{|r|}{35} & \multicolumn{1}{|r|}{10} &
\multicolumn{1}{|r|}{231.84} & \multicolumn{1}{|r|}{11.68}\\\hline
\multicolumn{1}{|r|}{40} & \multicolumn{1}{|r|}{15} &
\multicolumn{1}{|r|}{366.79} & \multicolumn{1}{|r|}{15.53}\\\hline
\end{tabular}

\textbf{Table 6.} Characteristics of the various families of interval graphs
used for the experiments.
\end{center}

\bigskip

Finally, we apply the ten GSST variants to each of the five familes (using
$M=2\cdot10^{4}$ spanning trees per variant). The results are summarized in
Tables 7.a (for the uniform variants) and 7.b (for the DFS\ variants). In
these tables every row corresponds to a family, the families being indexed by
their average interval widths (appearing in the first column). Each of the
remaining columns corresponds to one GSST\ variant and lists the average (over
the 100 graphs)\ minimum number of searchers achieved by the respective
variant. Hence, an estimate of the efficiency of the GSST\ variants can be
obtained by comparing the first column of each table to the remaining ones.
For example, in the first row of Table 7.a we see that all uniform
GSST\ variants node clear the graphs of the first family with fewer searchers
(on the average)\ than those \textquotedblleft predicted\textquotedblright\ by
the interval width of the graphs (e.g., $5.15<5.49$). As we proceed down the
rows of Table 7.a to graphs of higher interval width (and, presumably, of
greater complexity) the average number of searchers required to clear a graph
increases above the average interval width. For example, in the last row the
average interval width is 15.53 and the average minimum number of searchers
required by the uniform GSST-L variant is 19.05; in other words the uniform
GSST-L requires $\frac{19.05-15.53}{15.53}=\allowbreak22.66\%$ more searchers
than expected by the interval width estimate. Things get better with the
DFS\ variants; for example, for the $\left(  N=40,\Delta=15\right)  $ family,
the searcher overhead incurred by the DFS GSST-L\ variant is $\frac
{17.68-15.53}{15.53}=\allowbreak13.84\%$, which is actually quite good for
graphs of such high complexity.

\begin{center}%
\begin{tabular}
[c]{|l|l|l|l|l|l|}\hline
\textbf{Av. Int. Width} & \textbf{GSST-L} & \textbf{GSST-R} & \textbf{GSST-LR}
& \textbf{GSST-LW} & \textbf{GSST-LD}\\\hline
\multicolumn{1}{|r|}{5.49} & \multicolumn{1}{|r|}{5.15} &
\multicolumn{1}{|r|}{5.32} & \multicolumn{1}{|r|}{5.36} &
\multicolumn{1}{|r|}{5.14} & \multicolumn{1}{|r|}{5.16}\\\hline
\multicolumn{1}{|r|}{7.66} & \multicolumn{1}{|r|}{7.82} &
\multicolumn{1}{|r|}{8.08} & \multicolumn{1}{|r|}{8.03} &
\multicolumn{1}{|r|}{7.80} & \multicolumn{1}{|r|}{7.87}\\\hline
\multicolumn{1}{|r|}{9.51} & \multicolumn{1}{|r|}{11.01} &
\multicolumn{1}{|r|}{11.42} & \multicolumn{1}{|r|}{11.45} &
\multicolumn{1}{|r|}{11.02} & \multicolumn{1}{|r|}{11.30}\\\hline
\multicolumn{1}{|r|}{11.68} & \multicolumn{1}{|r|}{13.61} &
\multicolumn{1}{|r|}{13.91} & \multicolumn{1}{|r|}{13.91} &
\multicolumn{1}{|r|}{13.60} & \multicolumn{1}{|r|}{13.63}\\\hline
\multicolumn{1}{|r|}{15.53} & \multicolumn{1}{|r|}{19.05} &
\multicolumn{1}{|r|}{19.37} & \multicolumn{1}{|r|}{19.50} &
\multicolumn{1}{|r|}{19.06} & \multicolumn{1}{|r|}{19.25}\\\hline
\end{tabular}

\end{center}

\noindent\textbf{Table 7.a}: Average minimum number of searchers required to
node-clear each family of interval graphs\ by the uniform GSST variants;
number of spanning trees generated is $M=2\cdot10^{4}.$

\begin{center}
\medskip%

\begin{tabular}
[c]{|l|l|l|l|l|l|}\hline
\textbf{Av. Int. Width} & \textbf{GSST-L} & \textbf{GSST-R} & \textbf{GSST-LR}
& \textbf{GSST-LW} & \textbf{GSST-LD}\\\hline
\multicolumn{1}{|r|}{5.49} & \multicolumn{1}{|r|}{5.04} &
\multicolumn{1}{|r|}{5.13} & \multicolumn{1}{|r|}{5.04} &
\multicolumn{1}{|r|}{5.08} & \multicolumn{1}{|r|}{5.06}\\\hline
\multicolumn{1}{|r|}{7.66} & \multicolumn{1}{|r|}{7.34} &
\multicolumn{1}{|r|}{7.47} & \multicolumn{1}{|r|}{7.39} &
\multicolumn{1}{|r|}{7.44} & \multicolumn{1}{|r|}{7.36}\\\hline
\multicolumn{1}{|r|}{9.51} & \multicolumn{1}{|r|}{10.23} &
\multicolumn{1}{|r|}{10.34} & \multicolumn{1}{|r|}{10.25} &
\multicolumn{1}{|r|}{10.25} & \multicolumn{1}{|r|}{10.24}\\\hline
\multicolumn{1}{|r|}{11.68} & \multicolumn{1}{|r|}{12.70} &
\multicolumn{1}{|r|}{12.80} & \multicolumn{1}{|r|}{12.64} &
\multicolumn{1}{|r|}{12.73} & \multicolumn{1}{|r|}{12.62}\\\hline
\multicolumn{1}{|r|}{15.53} & \multicolumn{1}{|r|}{17.68} &
\multicolumn{1}{|r|}{17.87} & \multicolumn{1}{|r|}{17.75} &
\multicolumn{1}{|r|}{17.72} & \multicolumn{1}{|r|}{17.66}\\\hline
\end{tabular}

\end{center}

\noindent\textbf{Table 7.b}: Average minimum number of searchers required to
node-clear each family of interval graphs\ by the DFS GSST variants; number of
spanning trees generated is $M=2\cdot10^{4}.$

\bigskip

The results on interval graphs show that GSST can yield near-minimal schedules
on a large class of graphs. They also demonstrate that GSST \ scales well with
increasing complexity both in performance and computation time. This is a
direct result of the linear scalability of the algorithm in the number of
nodes in the environment.

\subsection{Grid Graphs}

\label{sec0603}

\subsubsection{Full Grids}

We now present some experiments involving grid graphs, i.e. graphs with nodes
located at points with integer coordinates (an example appears in Fig.
\ref{fig18}). It is easily seen that a grid graph of dimension $J_{1}\times
J_{2}$ (i.e. containing nodes with coordinates $\left(  j_{1},j_{2}\right)
\in\left\{  1,2,...,J_{1}\right\}  \times\left\{  1,2,...,J_{2}\right\}  $)
can be node cleared using $J_{0}=\min\left(  J_{1},J_{2}\right)  $ searchers;
the corresponding search schedule is obvious\footnote{Namely (assuming that
the height $J_{1}$ of the grid is less than or equal to the width $J_{2}$)
place $J_{1}$ searchers on one of the vertical sides of the grid and slide
them horizontally to the other vertical side.}. However, for a \emph{general}
graph search algorithm (such as GSST) which must work without assuming any
special structure of the graph, grid graphs are potentially extremely hard,
because they have many cycles and many spanning trees, only a few of which
correspond to minimal search schedules.

In this experiment we use 6 grid graphs, with dimensions starting at
$5\times5$ and going up to $10\times10$. To each such graph we apply the ten
GSST variants; in Tables 8.a (uniform ST\ generation variants) and 8.b
(DFS\ ST\ generation variants) we list the minimum search number attained by
each variant on each graph ($M$, the number of spanning trees generated per
graph is also listed, in the last column). Let us also note that the total
execution time for this experiment is approximately 25 hours, reflecting (a)
the large number of graphs used, (b) the high complexity of many graphs,
(c)\ the large number of spanning trees used (the \emph{much smaller
}execution time for each individual graph and GSST\ variant is not listed, for
economy of space).

\begin{center}%
\begin{tabular}
[c]{|l|r|r|r|r|r|r|}\hline
\multicolumn{7}{|c|}{\textbf{Uniform random spanning tree generation}}\\\hline
\textbf{Graph Dim.} & \textbf{GSST-L} & \textbf{GSST-LR} & \textbf{GSST-R} &
\textbf{GSST-LW} & \textbf{GSST-LD} & \textbf{No.Trees}\\\hline
\multicolumn{1}{|r|}{5$\times$5} & 5 & 5 & 5 & 5 & 5 & 10$^{5}$\\\hline
\multicolumn{1}{|r|}{6$\times$6} & 7 & 7 & 7 & 7 & 7 & 2$\cdot$10$^{5}%
$\\\hline
\multicolumn{1}{|r|}{7$\times$7} & 8 & 8 & 8 & 8 & 8 & 3$\cdot$10$^{5}%
$\\\hline
\multicolumn{1}{|r|}{8$\times$8} & 10 & 9 & 9 & 9 & 9 & 3$\cdot$10$^{5}%
$\\\hline
\multicolumn{1}{|r|}{9$\times$9} & 11 & 11 & 11 & 10 & 11 & 4$\cdot$10$^{5}%
$\\\hline
\multicolumn{1}{|r|}{10$\times$10} & 13 & 13 & 13 & 13 & 13 & 5$\cdot$10$^{5}%
$\\\hline
\end{tabular}

\end{center}

\noindent\textbf{Table 8.a}. Minimum search number attained by the uniform
GSST variants on the full-grid graphs.

\begin{center}
\bigskip%

\begin{tabular}
[c]{|l|r|r|r|r|r|r|}\hline
\multicolumn{7}{|c|}{\textbf{DFS random spanning tree generation}}\\\hline
\textbf{Graph Dim.} & \textbf{GSST-L} & \textbf{GSST-LR} & \textbf{GSST-R} &
\textbf{GSST-LW} & \textbf{GSST-LD} & \textbf{No.Trees}\\\hline
\multicolumn{1}{|r|}{5$\times$5} & 6 & 6 & 6 & 6 & 6 & 10$^{5}$\\\hline
\multicolumn{1}{|r|}{6$\times$6} & 7 & 7 & 8 & 7 & 7 & 2$\cdot$10$^{5}%
$\\\hline
\multicolumn{1}{|r|}{7$\times$7} & 9 & 9 & 9 & 9 & 9 & 3$\cdot$10$^{5}%
$\\\hline
\multicolumn{1}{|r|}{8$\times$8} & 11 & 11 & 11 & 11 & 11 & 3$\cdot$10$^{5}%
$\\\hline
\multicolumn{1}{|r|}{9$\times$9} & 12 & 13 & 12 & 13 & 12 & 4$\cdot$10$^{5}%
$\\\hline
\multicolumn{1}{|r|}{10$\times$10} & 14 & 14 & 14 & 14 & 15 & 5$\cdot$10$^{5}%
$\\\hline
\end{tabular}

\end{center}

\noindent\textbf{Table 8.b. }Minimum search number attained by the DFS\ GSST
variants on the full-grid graphs.\bigskip

\begin{center}

\end{center}

We see that, on the grids, the uniform variants perform better than the
DFS\ ones; this is exactly the opposite situation from what happens in
interval graphs. With few exceptions, the uniform variants (Table 8.a)\ either
attain the actual minimum search number (for the $5\times5$ graph) or incur a
small overhead of 1 or 2 extra searchers. The situation changes with the
$10\times10$ graph, where the minimum attained search number is 13, with a
30\% overhead over the true $s_{N}^{imc}=10$. However, the $10\times10$ grid
graph is very complex, with 100 nodes, 180 edges and 5.6943$\cdot10^{42}$
spanning trees. Judging from the distribution of the attained search numbers
(which, for economy of space, is not displayed here) lower search numbers
\emph{can} be attained but they require a much larger number of spanning
trees. To get an idea of the complexity of the full grid graphs, we list in
Table 8.c. the number of spanning trees for each of the six graphs used. This
provides a measure of the hardness of node-clearing the graph. This actually
quite reasonable: a graph with many cycles (and many spanning trees)\ is
harder to clear because it contains many escape routes for the evader.

\begin{center}%
\begin{tabular}
[c]{|l|r|}\hline
\textbf{Graph Dim.} & \textbf{No. of spaning trees}\\\hline
\multicolumn{1}{|r|}{5$\times$5} & 5.57560$\cdot10^{8}$\\\hline
\multicolumn{1}{|r|}{6$\times$6} & 3.2566$\cdot10^{13}$\\\hline
\multicolumn{1}{|r|}{7$\times$7} & 1.9872$\cdot10^{19}$\\\hline
\multicolumn{1}{|r|}{8$\times$8} & 1.2623$\cdot10^{26}$\\\hline
\multicolumn{1}{|r|}{9$\times$9} & 8.3266$\cdot10^{33}$\\\hline
\multicolumn{1}{|r|}{10$\times$10} & 5.6943$\cdot10^{42}$\\\hline
\end{tabular}

\textbf{Table 8.c. }Number of spanning trees for each of the full-grid graphs

\end{center}

On the positive side, we see that node-clearing schedules can be computed for
the full grid graphs using slightly more than the minimum number of searchers
and in reasonable computation time. For example, the uniform GSST-L\ variant
finds the first 8-searcher clearing schedule for the $7\times7$ graph in under
2 mins and an 11-searcher clearing schedule for the $9\times9$ graph in under
10 mins. Considering that (due to the large number of cycles)\ grid graphs are
potentially some of the hardest graphs for GSST\footnote{It is interesting to
note that a somewhat related fact is well known in the probabilistic inference
literature: inference is easier on graphs with low tree-width \cite{Chavira}%
.}, we find these results to be quite satisfactory.

\subsubsection{Depleted Grids}

Our final experiment involves what we call \textquotedblleft depleted grid
graphs\textquotedblright. These are obtained as follows. First we choose two
parameters, $J_{1}$ and $J_{2}$, the length and width of the grid. Just like
with the full grid graphs, we place nodes at the positions $\left(
j_{1},j_{2}\right)  \in\left\{  1,2,...,J_{1}\right\}  \times\left\{
1,2,...,J_{2}\right\}  $. We also connect all nearest neighbor nodes along the
following lines

\begin{enumerate}
\item one horizontal line: $\left\{  \left(  j_{1},1\right)  \right\}
_{j_{1}\in\left\{  1,2,...,J_{2}\right\}  }$,

\item $J_{1}$ vertical lines:$\ \left\{  \left(  j_{1},j_{2}\right)  \right\}
_{j_{2}\in\left\{  1,2,...,J_{2}\right\}  }$, for $j_{1}=1,2,...,J_{1},$
\end{enumerate}

\noindent obtaining a tree of the form indicated by the solid edges in
Fig.\ref{fig18}. Finally, we consider all pairs of nearest neighbors $\left(
i_{1},i_{2}\right)  $, $\left(  j_{1},j_{2}\right)  $ which are not already
connected and add to the graph an edge connecting each such pair with
probability $p$.

\begin{figure}[h]
\centering\scalebox{0.7}{\includegraphics{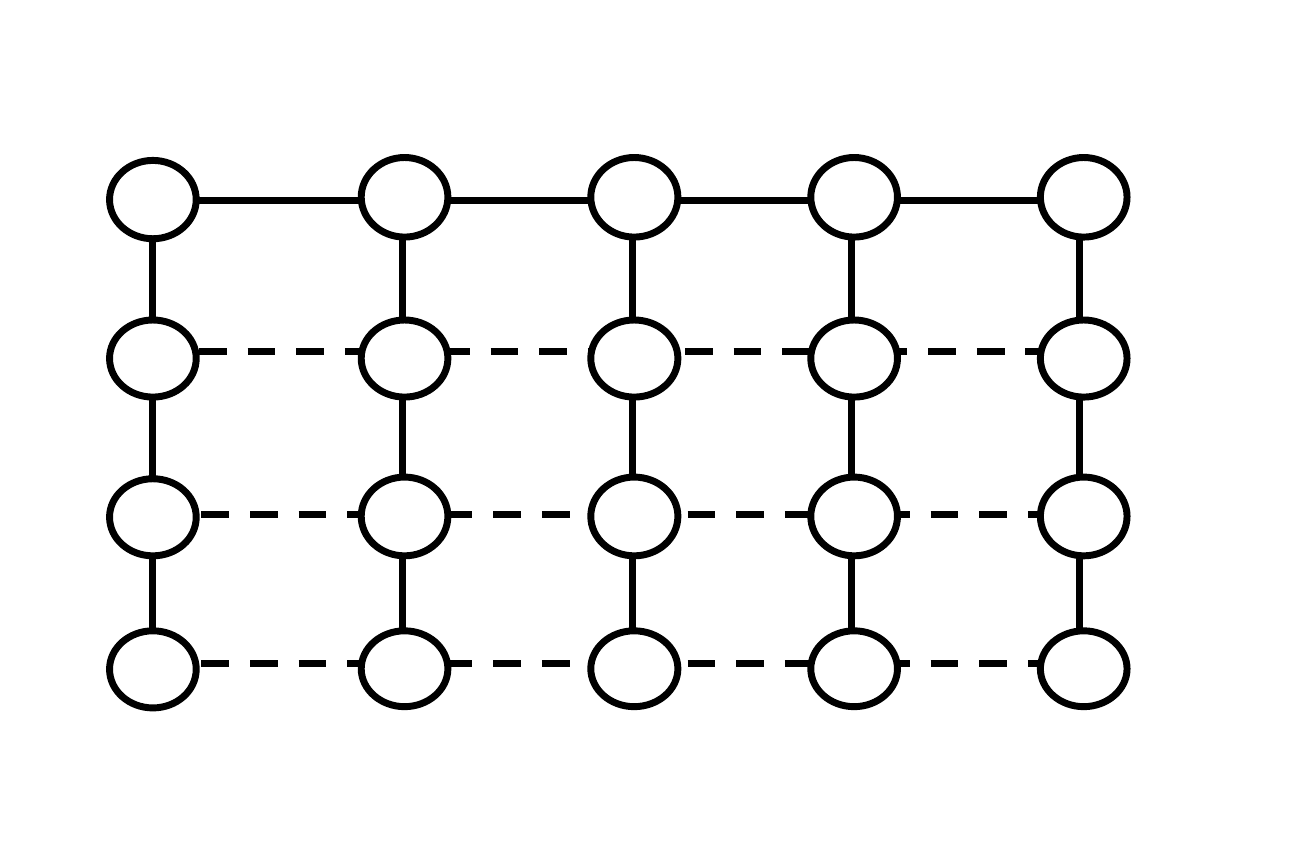}}\caption{Generating a
depleted grid graph. The solid lines correspond to edges which are always
present in the graph. Each dotted line becomes and edge with probability $p.$}%
\label{fig18}%
\end{figure}

The resulting graph $\mathbf{G}$ is a subgraph of the full $J_{1}\times J_{2}$
grid and has $s_{N}^{imc}\left(  \mathbf{G}\right)  \leq\min\left(
J_{1},J_{2}\right)  $; in other words, $\min\left(  J_{1},J_{2}\right)  $
yields an upper bound (useful for estimating the performance of GSST) of
$s_{N}^{imc}\left(  \mathbf{G}\right)  $. It seems reasonable that, on the
average, $s_{N}^{imc}\left(  \mathbf{G}\right)  $ is an increasing function of
$p$ and for $p=1$ the bound is tight, i.e. $s_{N}^{imc}\left(  \mathbf{G}%
\right)  =\min\left(  J_{1},J_{2}\right)  $; it will be useful to keep this in
mind when evaluating the results of the experiment. Note also that, as $p$
approaches 0, the graph becomes more similar to a tree (has fewer cycles).

We generate twelve families of depleted grids; namely we use dimensions
$5\times5$, $7\times7$, $8\times8$, $10\times10$ and $p$ values 0.4375, 0.7500
and 0.9375, yielding twelve combinations. We generate \emph{fifty }graphs from
each family and apply to each of these the ten GSST\ variants, \emph{using
}$2\cdot10^{5}$\emph{ spanning trees for each graph}. The results appear in
Tables 9.a (uniform variants) and 9.b (DFS\ variants). Namely, for each
combination we present the minimum search number attained, averaged over the
fifty graphs of the corresponding family. It can be seen that the resulting
numbers are quite low, especially for the uniform variants, often going under
the $\min\left(  J_{1},J_{2}\right)  $ bound.

\begin{center}%
\begin{tabular}
[c]{|c|c|c|c|c|c|c|}\hline
\textbf{Graph Dim.} & $p$ & \textbf{GSST-L} & \textbf{GSST-R} &
\textbf{GSST-LR} & \textbf{GSST-LW} & \textbf{GSST-LD}\\\hline
\multicolumn{1}{|r|}{5$\times$5} & \multicolumn{1}{|r|}{0.4375} &
\multicolumn{1}{|r|}{2.86} & \multicolumn{1}{|r|}{2.86} &
\multicolumn{1}{|r|}{2.86} & \multicolumn{1}{|r|}{2.86} &
\multicolumn{1}{|r|}{2.86}\\\hline
\multicolumn{1}{|r|}{5$\times$5} & \multicolumn{1}{|r|}{0.7500} &
\multicolumn{1}{|r|}{3.68} & \multicolumn{1}{|r|}{3.68} &
\multicolumn{1}{|r|}{3.50} & \multicolumn{1}{|r|}{3.48} &
\multicolumn{1}{|r|}{3.48}\\\hline
\multicolumn{1}{|r|}{5$\times$5} & \multicolumn{1}{|r|}{0.9375} &
\multicolumn{1}{|r|}{4.53} & \multicolumn{1}{|r|}{4.53} &
\multicolumn{1}{|r|}{4.12} & \multicolumn{1}{|r|}{4.18} &
\multicolumn{1}{|r|}{4.20}\\\hline
\multicolumn{1}{|r|}{7$\times$7} & \multicolumn{1}{|r|}{0.4375} &
\multicolumn{1}{|r|}{3.63} & \multicolumn{1}{|r|}{3.56} &
\multicolumn{1}{|r|}{3.64} & \multicolumn{1}{|r|}{3.57} &
\multicolumn{1}{|r|}{3.62}\\\hline
\multicolumn{1}{|r|}{7$\times$7} & \multicolumn{1}{|r|}{0.7500} &
\multicolumn{1}{|r|}{5.24} & \multicolumn{1}{|r|}{5.14} &
\multicolumn{1}{|r|}{5.32} & \multicolumn{1}{|r|}{5.08} &
\multicolumn{1}{|r|}{5.06}\\\hline
\multicolumn{1}{|r|}{7$\times$7} & \multicolumn{1}{|r|}{0.9375} &
\multicolumn{1}{|r|}{6.62} & \multicolumn{1}{|r|}{6.40} &
\multicolumn{1}{|r|}{6.56} & \multicolumn{1}{|r|}{6.30} &
\multicolumn{1}{|r|}{6.28}\\\hline
\multicolumn{1}{|r|}{8$\times$8} & \multicolumn{1}{|r|}{0.4375} &
\multicolumn{1}{|r|}{4.22} & \multicolumn{1}{|r|}{4.14} &
\multicolumn{1}{|r|}{4.08} & \multicolumn{1}{|r|}{4.08} &
\multicolumn{1}{|r|}{4.00}\\\hline
\multicolumn{1}{|r|}{8$\times$8} & \multicolumn{1}{|r|}{0.7500} &
\multicolumn{1}{|r|}{6.10} & \multicolumn{1}{|r|}{6.14} &
\multicolumn{1}{|r|}{5.90} & \multicolumn{1}{|r|}{5.86} &
\multicolumn{1}{|r|}{5.84}\\\hline
\multicolumn{1}{|r|}{8$\times$8} & \multicolumn{1}{|r|}{0.9375} &
\multicolumn{1}{|r|}{7.70} & \multicolumn{1}{|r|}{7.76} &
\multicolumn{1}{|r|}{7.52} & \multicolumn{1}{|r|}{7.40} &
\multicolumn{1}{|r|}{7.46}\\\hline
\multicolumn{1}{|r|}{10$\times$10} & \multicolumn{1}{|r|}{0.4375} &
\multicolumn{1}{|r|}{5.44} & \multicolumn{1}{|r|}{5.50} &
\multicolumn{1}{|r|}{5.32} & \multicolumn{1}{|r|}{5.30} &
\multicolumn{1}{|r|}{5.28}\\\hline
\multicolumn{1}{|r|}{10$\times$10} & \multicolumn{1}{|r|}{0.7500} &
\multicolumn{1}{|r|}{7.94} & \multicolumn{1}{|r|}{8.02} &
\multicolumn{1}{|r|}{7.84} & \multicolumn{1}{|r|}{7.68} &
\multicolumn{1}{|r|}{7.80}\\\hline
\multicolumn{1}{|r|}{10$\times$10} & \multicolumn{1}{|r|}{0.9375} &
\multicolumn{1}{|r|}{10.20} & \multicolumn{1}{|r|}{10.40} &
\multicolumn{1}{|r|}{10.24} & \multicolumn{1}{|r|}{10.06} &
\multicolumn{1}{|r|}{10.94}\\\hline
\end{tabular}

\end{center}

\noindent\textbf{Table 9.a.}\emph{ }Average minimum number of searchers
required to node-clear each family of depleted-grid graphs\ by the uniform
GSST variants; number of spanning trees generated is $M=2\cdot10^{5}$.

\bigskip

\begin{center}%
\begin{tabular}
[c]{|c|c|c|c|c|c|c|}\hline
\textbf{Graph Dim.} & $p$ & \textbf{GSST-L} & \textbf{GSST-R} &
\textbf{GSST-LR} & \textbf{GSST-LW} & \textbf{GSST-LD}\\\hline
\multicolumn{1}{|r|}{5$\times$5} & \multicolumn{1}{|r|}{0.4375} &
\multicolumn{1}{|r|}{3.36} & \multicolumn{1}{|r|}{3.36} &
\multicolumn{1}{|r|}{3.36} & \multicolumn{1}{|r|}{3.36} &
\multicolumn{1}{|r|}{3.36}\\\hline
\multicolumn{1}{|r|}{5$\times$5} & \multicolumn{1}{|r|}{0.7500} &
\multicolumn{1}{|r|}{4.38} & \multicolumn{1}{|r|}{4.36} &
\multicolumn{1}{|r|}{4.36} & \multicolumn{1}{|r|}{4.36} &
\multicolumn{1}{|r|}{4.36}\\\hline
\multicolumn{1}{|r|}{5$\times$5} & \multicolumn{1}{|r|}{0.9375} &
\multicolumn{1}{|r|}{5.04} & \multicolumn{1}{|r|}{5.04} &
\multicolumn{1}{|r|}{5.02} & \multicolumn{1}{|r|}{5.04} &
\multicolumn{1}{|r|}{4.98}\\\hline
\multicolumn{1}{|r|}{7$\times$7} & \multicolumn{1}{|r|}{0.4375} &
\multicolumn{1}{|r|}{4.24} & \multicolumn{1}{|r|}{4.24} &
\multicolumn{1}{|r|}{4.24} & \multicolumn{1}{|r|}{4.24} &
\multicolumn{1}{|r|}{4.24}\\\hline
\multicolumn{1}{|r|}{7$\times$7} & \multicolumn{1}{|r|}{0.7500} &
\multicolumn{1}{|r|}{5.82} & \multicolumn{1}{|r|}{5.78} &
\multicolumn{1}{|r|}{5.84} & \multicolumn{1}{|r|}{5.80} &
\multicolumn{1}{|r|}{5.77}\\\hline
\multicolumn{1}{|r|}{7$\times$7} & \multicolumn{1}{|r|}{0.9375} &
\multicolumn{1}{|r|}{7.24} & \multicolumn{1}{|r|}{7.30} &
\multicolumn{1}{|r|}{7.28} & \multicolumn{1}{|r|}{7.28} &
\multicolumn{1}{|r|}{7.32}\\\hline
\multicolumn{1}{|r|}{8$\times$8} & \multicolumn{1}{|r|}{0.4375} &
\multicolumn{1}{|r|}{4.68} & \multicolumn{1}{|r|}{4.66} &
\multicolumn{1}{|r|}{4.66} & \multicolumn{1}{|r|}{4.63} &
\multicolumn{1}{|r|}{4.64}\\\hline
\multicolumn{1}{|r|}{8$\times$8} & \multicolumn{1}{|r|}{0.7500} &
\multicolumn{1}{|r|}{6.56} & \multicolumn{1}{|r|}{6.62} &
\multicolumn{1}{|r|}{6.62} & \multicolumn{1}{|r|}{6.64} &
\multicolumn{1}{|r|}{6.62}\\\hline
\multicolumn{1}{|r|}{8$\times$8} & \multicolumn{1}{|r|}{0.9375} &
\multicolumn{1}{|r|}{8.58} & \multicolumn{1}{|r|}{8.52} &
\multicolumn{1}{|r|}{8.62} & \multicolumn{1}{|r|}{8.54} &
\multicolumn{1}{|r|}{8.56}\\\hline
\multicolumn{1}{|r|}{10$\times$10} & \multicolumn{1}{|r|}{0.4375} &
\multicolumn{1}{|r|}{5.80} & \multicolumn{1}{|r|}{5.80} &
\multicolumn{1}{|r|}{5.78} & \multicolumn{1}{|r|}{5.78} &
\multicolumn{1}{|r|}{5.80}\\\hline
\multicolumn{1}{|r|}{10$\times$10} & \multicolumn{1}{|r|}{0.7500} &
\multicolumn{1}{|r|}{8.62} & \multicolumn{1}{|r|}{8.58} &
\multicolumn{1}{|r|}{8.68} & \multicolumn{1}{|r|}{8.72} &
\multicolumn{1}{|r|}{8.62}\\\hline
\multicolumn{1}{|r|}{10$\times$10} & \multicolumn{1}{|r|}{0.9375} &
\multicolumn{1}{|r|}{11.34} & \multicolumn{1}{|r|}{11.38} &
\multicolumn{1}{|r|}{11.54} & \multicolumn{1}{|r|}{11.38} &
\multicolumn{1}{|r|}{11.26}\\\hline
\end{tabular}

\end{center}

\noindent\textbf{Table 9.b}.Average minimum number of searchers required to
node-clear each family of depleted-grid graphs\ by the DFS\ GSST variants;
number of spanning trees generated is $M=2\cdot10^{5}$.

\bigskip

The depleted grid results show the high performance of GSST on graphs with
large numbers of spanning trees and many cycles. Better results are obtained
for smaller values of $p$, i.e. when the graph becomes more similar to a tree;
this is to be expected, since a basic component of the GSST algorithm is
search along a spanning tree. It is interesting to note that even for a quite
high value of $p$, namely 0.9375, results (at least for uniform
GSST\ variants)\ are quite good; taking into account also the results of the
full grids, it appears that the increase of attained search number over
$\min\left(  J_{1},J_{2}\right)  $ happens only when $p$ gets \emph{very}
close to 1. Once again we observe that the uniform variants perform better
than the DFS\ ones.

\subsection{Discussion}

\label{grid0604}

Our experiments demonstrate the performance of GSST on several different
classes of graphs. On many complex graphs, GSST was able to find a minimal
search schedule in reasonable time. On \emph{all} graphs, at least one
near-minimal schedule is computed early in the execution of the algorithm. In
fact, the \textquotedblleft anytime\textquotedblright\ operation of GSST\ (see
also the discussion in Section \ref{sec0501})\ yields reasonably good
solutions in very short time and then keeps improving on these solutions as
long as additional computation time is available. While the minimal schedules
form a small percentage of the computed solutions, many more near-minimal
solutions are available. These properties of GSST\ are highly desirable,
especially given the fact that (as far as we are aware) no other algorithm has
been \emph{implemented} which can tackle graphs of the size and complexity
presented here.

We can highlight some specific conclusions supported by the experiments on the
various graphs considered in Section \ref{sec0601} -- \ref{sec0603}.

\begin{enumerate}
\item The NSH and National Art Gallery graphs show the application of GSST on
graphs that are derived from representations of indoor environments. The
methods that incorporate Barriere labeling improve performance on the NSH map
because it is similar to a tree. The National Art Gallery, on the other hand,
is more similar to a grid, which lessens the advantage of using Barriere
labeling to guide traversal.

\item The results on interval graphs show that GSST can yield near-minimal
schedules on a large class of graphs. They also demonstrate that, with
increasing complexity, GSST \ scales well both in performance and computation time.

\item Finally, the grid results show the high performance of GSST on graphs
with large numbers of spanning trees and many cycles. These are potentially
some of the hardest graphs for GSST. Also, it appears that the the increase of
attained search number over $\min\left(  J_{1},J_{2}\right)  $ happens only
when $p$ gets \emph{very} close to 1.
\end{enumerate}

Comparing spanning tree generation methods, we see that DFS\ generation
performs better on all graphs considered except for grid graphs (and the
tree/grid graph) where the uniform ST\ generation is better. At this point we
do not have a good explanation but we intend to further research this fact
because we believe it will give us a better understanding of which are the
\textquotedblleft good\textquotedblright\ spanning trees, i.e. the ones
associated with minimal searches. See Section \ref{sec07} for some additional
remarks on this issue.

Regarding edge traversal methods, there is no single GSST\ variant which
consistently outperforms all others. However, uniform GSST-L never does too
poorly and is fast; GSST-LW has the same advantages and, in addition, is
provably complete. Hence, in case node-clearing schedules must be quickly
produced (i.e. computation time is at premium), one can use a
\textquotedblleft reduced suite\textquotedblright\ of GSST\ variants,
consisting of GSST-L and GSST-LW (with both uniform and DFS\ spanning tree
generation) and expect to obtain nearly as good results (at a fraction of the
computing time)\ as when using the \textquotedblleft full\textquotedblright%
\ suite of the ten GSST\ variants.

Let us also stress that all variants of GSST\ can attain lower search numbers
than the ones presented here if sufficient running time is available.

\section{Conclusion}

\label{sec07}

Motivated by the problem of robotic pursuit / evasion, in the current paper we
have investigated \emph{node search}, i.e., the capture by a team of searchers
of an invisible evader located in the \emph{nodes} of a graph. This problem
has so far received little attention in the graph theoretic literature, with
most related publications concentrating on the problem of edge-located evader.
The basic contributions of the current paper are of two kinds.

\begin{enumerate}
\item From the theoretical point of view, we have shown that, in general
graphs, the problem of node search is easier than that of edge search, in the
sense that every edge clearing search is also a node clearing one; but the
converse does not hold in general. We have then concentrated on the internal
monotone connected (IMC)\ node search of trees and shown that it is
essentially equivalent to edge search under the same restrictions; hence
Barriere's tree search algorithm, originally designed for edge search, can
also be used for node search.

\item From the algorithmic point of view, we have presented GSST, a new
algorithm which performs IMC\ search on general graphs. This algorithm is
based on the fact that every node clearing search generates a spanning tree;
hence node-clearing a preselected spanning tree (by the use of \emph{tree
searchers}) and simultaneously blocking recontamination (by the use of
\emph{guards}) \ monotonically node-clears the graph. Because spanning tree
generation and search can both be performed very efficiently, a large number
of spanning trees can be tried by GSST\ until the one yielding the minimum
number of searchers is discovered. Experiments have shown that GSST can
quickly node-clear large and complex graphs using a small number of searchers.
\end{enumerate}

Many issues touched upon in the current paper require additional research. We
conclude by listing some future research directions, categorized as follows.

\begin{enumerate}
\item \textbf{Theory}.

\begin{enumerate}
\item \emph{Optimality}. Given a graph $\mathbf{G}$ with IMC node search
number $s_{N}^{imc}\left(  \mathbf{G}\right)  $. Suppose we have $K$ searchers
with $K\geq s_{N}^{imc}\left(  \mathbf{G}\right)  $. What is the best we can
do? Of course \textquotedblleft best\textquotedblright\ depends on some
optimality criterion. For example we may want to clear the graph in the
minimum number of steps. Or, we may want a search strategy which minimizes the
cumulative \textquotedblleft area\textquotedblright\ of the dirty set (i.e.,
$\sum_{t=1}^{T}V_{N}^{D}\left(  t\right)  $). We would like to obtain
algorithms which solve such optimization problems.

\item Our experiments show that the uniform GSST\ variants perform better than
DFS\ on grids and grid-like graphs. We want to discover a mathematical
explanation of this fact, e.g., to find necessary and / or sufficient
conditions under which uniform GSST\ outperforms DFS\ GSST. A result of this
type will probably have useful computational implications -- see also item 2.(a)\ below.\ 

\item We want to establish a sequence of inequalities between node search
numbers. In other words, to refine (\ref{eq0452}), either for a general graph
$\mathbf{G}$, or for the special case when $\mathbf{G}$ is a tree. For the
latter case, we conjecture that we can re-establish the Barriere et al.
inequalities (\ref{eq0451}) so as to hold for \emph{node-}\ rather than
edge-search numbers.

\item Finally, we are interested in questions of \emph{consistency}. Suppose
that we are given a graph along with a \emph{configuration} of each nodes
(i.e., a specification of the clear and dirty nodes). Can this configuration
be obtained by a sequence of moves in the node game? In the edge game? Also,
given a configuration of clear and dirty nodes and another one of clear and
dirty edges, are these configurations compatible in the node game? In the edge game?
\end{enumerate}

\item \textbf{Computation. }The GSST\ algorithm searches the space of spanning
trees in a random manner. Is there a way to bias the search in some useful
way?\ For this we need a characterization of \textquotedblleft
good\textquotedblright\ spanning trees, i.e., trees which yield a low number
of searchers (this includes tree searchers \emph{and} guards). We have no such
characterization at present. Perhaps we can obtain it by understanding why
uniform variants perform better than DFS\ on grid graphs (and conversely on
interval graphs).

\item \textbf{Extensions}. We want to study theoretically and develop
algorithms for the following variants of the search problem.

\begin{enumerate}
\item For less restricted types of node search, e.g., internal monotone (but
not connected), internal connected (but not monotone); or for completely
unrestricted node search.

\item For different types of \ pursuer / evader behavior; e.g., for pursuers
with extended (non-local)\ visibility, evaders with finite speed etc.

\item How does the graph search problem look from the evader's point of view?
Given a graph $\mathbf{G}$ and $K$ searchers, what is the best the evader can
do if $K<s\left(  \mathbf{G}\right)  $? What if $K\geq s\left(  \mathbf{G}%
\right)  $? In the latter case capture is guaranteed, but can the evader, for
example, maximize the number of steps until capture?\ Many other interesting
problems can be posed from the evader's point of view.
\end{enumerate}
\end{enumerate}

\bigskip

\noindent\textbf{Acknowledgement}. We gratefully acknowledge Sanjiv Singh for
many useful discussions and valuable comments on the current paper.

\newpage

\appendix

\section{Appendix: Edge Search on Trees}

\label{secA}

In this Appendix we present the \textbf{R-Label} and \textbf{R-Search}
algorithms, which perform rooted IMC\ search of \emph{trees }and are also used
by the GSST\ algorithm. \textbf{R-Label} and \textbf{R-Search} are simplified
versions of corresponding algorithms presented by Barriere et al.
\cite{Barriere1}; some theoretical results from \cite{Barriere1} are also
reviewed. In our presentation the terminology and notation of \cite{Barriere1}
are somewhat changed, to conform with the current paper.

The next Lemma and Definition play a crucial role in computing the IMC\ search
number of a tree.

\begin{lemma}
\label{prp0407}\cite{Barriere1} Given a rooted tree $\mathbf{T}_{x}$ and a
node $y$ of $\mathbf{T}_{x}$, if $y$ has more than one children, enumerate
them as $y_{1},y_{2},...,y_{K}$, so that they satisfy $s_{E}^{imc}\left(
\mathbf{T}_{x}\left[  y_{k}\right]  \right)  \geq s_{E}^{imc}\left(
\mathbf{T}_{x}\left[  y_{k+1}\right]  \right)  $ for $k=1,2,...,K-1$. Then%
\begin{equation}
s_{E}^{imc}\left(  \mathbf{T}_{x}\left[  y\right]  \right)  =\max\left(
s_{E}^{imc}\left(  \mathbf{T}_{x}\left[  y_{1}\right]  \right)  ,s_{E}%
^{imc}\left(  \mathbf{T}_{x}\left[  y_{2}\right]  \right)  +1\right)  .
\label{eqA001}%
\end{equation}

\end{lemma}

\begin{definition}
\label{prp0410}\cite{Barriere1} Given a tree $\mathbf{T=}\left(  V,E\right)  $
and a node $x\in V$, consider all the edges $xy$ incident on $x$ and label
them as follows.

\begin{enumerate}
\item If $y$ is a leaf, then $\lambda_{x}\left(  xy\right)  =1$.

\item If $y$ is not a leaf, let $z_{1},z_{2},...,z_{K}$ be the neighbors of
$y$ other than $x$, enumerated so that $\lambda_{y}\left(  yz_{k}\right)
\geq\lambda_{y}\left(  yz_{k+1}\right)  $\ (for $k=1,2,...,K-1$) and define%
\begin{equation}
\lambda_{x}\left(  xy\right)  =\max\left(  \lambda_{y}\left(  yz_{1}\right)
,\lambda_{y}\left(  yz_{2}\right)  +1\right)  . \label{eqA002}%
\end{equation}

\end{enumerate}
\end{definition}

Note that the above definition assigns \emph{two} labels to each edge $xy$,
namely $\lambda_{x}\left(  xy\right)  $ and $\lambda_{y}\left(  xy\right)  $.
Intuitively, $\lambda_{x}\left(  xy\right)  $ is the number of searchers that
must cross the \emph{directed }$xy$ edge in an IMC\ edge-clearing of the tree.
In \cite{Barriere1} Barriere et al. present the \textbf{Label }algorithm which
computes the labels $\left\{  \lambda_{x}\left(  xy\right)  \right\}  _{xy\in
E}$. Furthermore they prove the following.

\begin{lemma}
\label{prp0411}\cite{Barriere1} For every tree $\mathbf{T}$ and every edge
$xy$ of $\mathbf{T}$, we have $s_{E}^{imc}\left(  \mathbf{T}_{x}[y]\right)
=\lambda_{x}\left(  xy\right)  $.
\end{lemma}

\begin{lemma}
\label{prp0411A}\cite{Barriere1} Given a tree $\mathbf{T=}\left(  V,E\right)
$, for every $x\in V$ define $\mu\left(  x\right)  $ as follows%
\[
\mu\left(  x\right)  =\left\{
\begin{array}
[c]{ll}%
\lambda_{x}\left(  xy_{1}\right)  & \text{if }x\text{ has a single neighbor
}y_{1}\\
\max\left(  \lambda_{x}\left(  xy_{1}\right)  ,\lambda_{x}\left(
xy_{2}\right)  +1\right)  & \text{otherwise;}%
\end{array}
\right.
\]
in the second line of the definition, $y_{1},y_{2},...$ are the neighbors of
$x$, enumerated so that they satisfy $s_{E}^{imc}\left(  \mathbf{T}_{x}\left[
y_{k}\right]  \right)  \geq s_{E}^{imc}\left(  \mathbf{T}_{x}\left[
y_{k+1}\right]  \right)  $ for $k=1,2,...\ $. Then
\[
s_{E}^{imc}\left(  \mathbf{T}_{x}\right)  =\mu\left(  x\right)  \ \text{and
}s_{E}^{imc}\left(  \mathbf{T}\right)  =\min_{x\in V}\mu\left(  x\right)  .
\]

\end{lemma}

The algorithms and results of \cite{Barriere1} are geared towards computing
$s_{E}^{imc}\left(  \mathbf{T}\right)  $. The GSST\ algorithm, on the other
hand, requires the computation of $s_{E}^{imc}\left(  \mathbf{T}_{x}\right)  $
and rooted searches. Hence we introduce the \textbf{R-Label} and
\textbf{R-Search} algorithms, which are simplified, \textquotedblleft
rooted\textquotedblright\ versions of the corresponding Barriere algorithms.
The pseudocode of these algorithms is listed in the next page.

The \textbf{R-Label} algorithm takes as input a tree $\mathbf{T=}\left(
V,E\right)  $ \emph{and} a root node $u_{0}\in V$. The algorithm is a
straightforward implementation of Definition \ref{prp0410}. \textbf{R-Label}
does \emph{not} compute all the labels $\left\{  \lambda_{x}\left(  xy\right)
\right\}  _{xy\in E}$ but only the ones for which $y$ is a child of $x$ in the
rooted tree $\mathbf{T}_{u_{0}}$; for these edges, \textbf{R-Label} produces
the same $\lambda$ labels as \textbf{Label}. \textbf{R-Label }uses two
subroutines:\ given that $u_{0}$ is the root, the subroutine \textbf{Depth}%
$\left(  \mathbf{T},u_{0}\right)  $ returns a partition $\left\{  V^{\left(
l\right)  }\right\}  _{l=1}^{L}$ of \ the node set $V$, where $V^{\left(
l\right)  }$ contains all nodes of depth $l$; the subroutine
\textbf{SortChildren}$\left(  y,\mathbf{T},u_{0},\lambda\right)  $ returns the
children of $y$ sorted in decreasing order of their $\lambda$ labels.

The \textbf{R-Search} algorithm uses the $\lambda$ labels to produce a rooted
IMC (edge and node)\ clearing search of $\mathbf{T}$. \textbf{R-Search} is
almost identical to Barriere's \textbf{Search}, the only difference being that
the starting node $u_{0}$ is given, rather than chosen by the algorithm (this
implies that \textbf{R-Search} produces a minimal \emph{rooted }IMC search of
$\mathbf{T}_{u_{0}}$). \textbf{R-Search} uses Barriere's subroutine
\textbf{Move}. The notation $\mathbf{S}=\mathbf{S}|\left(  u,v,J\right)  $
means:\ append $J$ moves of the form $u\rightarrow v$ to the search schedule
$\mathbf{S}$.

\begin{algorithm}[h]
\label{label}
\caption{\textbf{R-Label}$\left(  \mathbf{T} , u_0\right)$}
\begin{algorithmic}
\STATE \textbf{Input:} Tree $\mathbf{T}=(V,E)$, start node $u_0$
\STATE $\{ V^{(l)} \}_{l=1}^L=$ \textbf{Depth}$\left(  \mathbf{T} , u_0\right)$
\FOR{$l=L-1$ \textbf{with step} $-1$ \textbf{until} 1}
\FOR{$x \in V^{(l)}$}
\FOR{$y \in$ \textbf{SortChildren}$\left(x,\mathbf{T},u_{0},\lambda\right)$}
\STATE $[z_1,z_2,...,z_K]=$ \textbf{SortChildren}$\left(y,\mathbf{T},u_{0},\lambda\right)$
\IF{\textbf{Length}$([z_1,z_2,...,z_K])$=0}
\STATE $\lambda_x(xy) = 1$
\ELSIF{\textbf{Length}$([z_1,z_2,...,z_K])=1$}
\STATE $\lambda_x(xy)=\lambda_y(yz_1)$
\ELSE
\STATE $\lambda_x(xy)=\max(\lambda_y(yz_1),\lambda_y(yz_2)+1)$
\ENDIF
\ENDFOR
\ENDFOR
\ENDFOR
\STATE \textbf{Output:} Edge labeling $\{ \lambda_x(xy) \}_{xy \in E}$
\end{algorithmic}
\end{algorithm}

\begin{algorithm}[h]
\caption{\textbf{R-Search}$\left(  \mathbf{T}, u_0, \lambda \right)$}
\begin{algorithmic}
\STATE \textbf{Input:} Tree $\mathbf{T}=(V,E)$, start node $u_0$, edge labeling $\lambda$
\STATE $\mathbf{S}=\emptyset$
\STATE $[y_1,y_2,...,y_K]=$ \textbf{SortChildren}$\left(u_0,\mathbf{T},u_{0},\lambda\right)$
\FOR{$k=K$ \textbf{with step} $-1$ \textbf{until} 1}
\STATE $\mathbf{S}$=\textbf{Move}$(x,y_k,\lambda_x(xy_k),\mathbf{S})$
\ENDFOR
\STATE \textbf{Output:} Search $\mathbf{S}$
\end{algorithmic}
\label{alg:search}
\end{algorithm}

\begin{algorithm}[h]
\caption{\textbf{Move}$(u,v,J,\mathbf{S})$}
\begin{algorithmic}
\STATE \textbf{Input:} Tree $\mathbf{T}=(V,E)$, start node $u_0$, edge labeling $\lambda$
\STATE $\mathbf{S}=\mathbf{S} | (u,v,J)$
\STATE $[w_1,w_2,...,w_K]=$ \textbf{SortChildren}$\left(v,\mathbf{T},u_{0},\lambda\right)$
\FOR{$k=K$ \textbf{with step} $-1$ \textbf{until} 1}
\STATE $\mathbf{S}$=\textbf{Move}$(v,w_k,\lambda_y(yw_k),\mathbf{S})$
\ENDFOR
\STATE $\mathbf{S}=\mathbf{S} | (v,u,J)$
\STATE \textbf{Output:} Search $\mathbf{S}$
\end{algorithmic}
\label{alg:search}
\end{algorithm}

\clearpage

\newpage

\section{Appendix:\ The relationship between node search and \emph{mixed} edge
search}

\label{secB}

Theorem \ref{prp0303} tells us that edge search is \textquotedblleft
weaker\textquotedblright\ than node search in the sense that, for every graph
and every search schedule, we have $V_{E}^{C}\left(  t\right)  \subseteq
V_{N}^{C}\left(  t\right)  $ and $E_{E}^{C}\left(  t\right)  \subseteq
E_{N}^{C}\left(  t\right)  $. We will now consider a variant of edge search,
the so-called \emph{mixed edge search }which, as we will see, is
\emph{equivalent} to node search.

Mixed edge search is obtained by augmenting the clearing rules. In the
\textquotedblleft\emph{Mixed Search Edge Game}\textquotedblright\ an edge can
be cleared not only by sliding, but also by guarding both its endpoints. More
precisely, to obtain the rules of the mixed edge search game, take the rules
of the edge game (presented in Section \ref{sec02}), change the terms
\textquotedblleft e-clear\textquotedblright\ and \textquotedblleft
e-dirty\textquotedblright\ to \textquotedblleft m-clear\textquotedblright\ and
\textquotedblleft m-dirty\textquotedblright\ and add the following rule:

\begin{enumerate}
\item[\textbf{{E2}}$^{\prime}$] An m-dirty edge $uv$ becomes m-clear when both
$u$ and $v$ are occupied by searchers.
\end{enumerate}

In other words, we can clear an edge either by traversing it or by guarding
both its endpoints.

We will use $V_{M}^{C}\left(  t\right)  $ to denote the set of m-clear nodes
at time $t$, and will analogously use the notations $E_{M}^{C}\left(
t\right)  $, $V_{M}^{D}\left(  t\right)  $, $E_{M}^{D}\left(  t\right)  $,
$\mathbf{G}_{M}^{C}\left(  t\right)  $ etc.

We emphasize that, in the mixed search edge game, the evader resides in the
edges of the graph and recontamination occurs by the same rules as in the
\textquotedblleft plain\textquotedblright\ edge game.

\begin{theorem}
\label{prpA01}Given a graph $\mathbf{G}$ and an internal search schedule
$\mathbf{S}$. Then, we have
\begin{equation}
\text{for }t=0,1,2,...:\qquad V_{N}^{C}\left(  t\right)  =V_{M}^{C}\left(
t\right)  \quad\text{and\quad}E_{N}^{C}\left(  t\right)  =E_{M}^{C}\left(
t\right)  . \label{eq0313}%
\end{equation}

\end{theorem}

\begin{proof}
The proof is inductive and consists of several steps. \noindent

\noindent\textbf{I.} At $t=0$ we have%
\begin{equation}
V_{N}^{C}\left(  0\right)  =V_{M}^{C}\left(  0\right)  =\emptyset
\quad\text{and}\quad E_{N}^{C}\left(  0\right)  =E_{M}^{C}\left(  0\right)
=\emptyset.
\end{equation}
Suppose that we have
\begin{equation}
\text{for }s=0,1,...,t:\qquad V_{N}^{C}\left(  s\right)  =V_{M}^{C}\left(
s\right)  \quad\text{and}\quad E_{N}^{C}\left(  s\right)  =E_{M}^{C}\left(
s\right)  . \label{eq0302}%
\end{equation}
and let us next consider time $t+1$. \noindent

\noindent\textbf{II.} We will now prove that%
\begin{equation}
E_{N}^{C}\left(  t+1\right)  \subseteq E_{M}^{C}\left(  t+1\right)  .
\label{eq0303}%
\end{equation}
To this end we will show two things.

\begin{enumerate}
\item[\textbf{II.a}] \emph{If a previously n-dirty edge becomes n-clear at
}$t+1$\emph{, then it also becomes m-clear}. Suppose that $uv\in E_{N}%
^{D}\left(  t\right)  \cap E_{N}^{C}\left(  t+1\right)  $. In other words,
$uv$ is n-cleared exactly at $t+1$; this means that one of $u,v$ was already
n-clear by $t$ and the other was n-cleared \emph{exactly} at $t+1$. Without
loss of generality, assume%
\begin{equation}
u\in V_{N}^{C}\left(  t\right)  \cap V_{N}^{C}\left(  t+1\right)  ,\qquad v\in
V_{N}^{D}\left(  t\right)  \cap V_{N}^{C}\left(  t+1\right)  \label{eq0322}%
\end{equation}
Since a node can be n-cleared only by moving into it, the $\left(  t+1\right)
$-th move must be $w\rightarrow v$. We distinguish the following cases.

\textbf{(i.1)} Suppose $w\neq u$. Since $v\in V_{N}^{D}\left(  t\right)  $ and
$u\in V_{N}^{C}\left(  t\right)  $, $u$ was guarded at $t$; since $w\neq u$,
$u$ remains guarded at $t+1$. Hence, at $t+1$, $uv$ has both ends guarded and
so $uv\in E_{M}^{C}\left(  t+1\right)  $. (Note that this analysis holds even
when $w=0$, i.e., when the move was to place a new searcher).

\textbf{(i.2)} Suppose $w=u$ and $u$ is still guarded at $t+1$ (i.e., at $t$
there was more than one searcher in $u$). Then at $t+1$ both $u$ and $v$ are
guarded and $uv\in E_{M}^{C}\left(  t+1\right)  $.

\textbf{(i.3)} The remaining possibility is that $w=u$ and $u$ is unguarded at
$t+1$ (i.e., at $t$ there was exactly one searcher in $u$). Note that
$u\rightarrow v$ at $t+1$ means that $u$ is guarded at $t$ and hence $u\in
V_{M}^{C}\left(  t\right)  $. If $u\in V_{M}^{C}\left(  t+1\right)  $ as well,
then $uv\in V_{M}^{C}\left(  t+1\right)  $. So we now will show that $u\in
V_{M}^{D}\left(  t+1\right)  $ is not possible. Indeed $u\in V_{M}^{D}\left(
t+1\right)  $ can only happen in one of the following two ways.

\textbf{(i.3.1)} There exists a node $x_{1}\neq v$ such that $ux_{1}\in
E_{M}^{D}\left(  t\right)  \cap E_{M}^{D}\left(  t+1\right)  $. But $ux_{1}\in
E_{M}^{D}\left(  t\right)  =E_{N}^{D}\left(  t\right)  $ implies $x_{1}\in
V_{N}^{D}\left(  t\right)  $ (because $u$ is guarded at $t$) and $x_{1}\in
V_{N}^{D}\left(  t\right)  $ implies $x_{1}\in V_{N}^{D}\left(  t+1\right)  $
(because $x_{1}$ was not entered at $t+1$). Finally, from $x_{1}\in V_{N}%
^{D}\left(  t+1\right)  $ and $u$ unguarded at $t+1$, we conclude $u\in
V_{N}^{D}\left(  t+1\right)  $ which contradicts assumption (\ref{eq0322}).

\textbf{(i.3.2)} Alternatively, there exists a path $vux_{1}...x_{K}$ (with
$K\geq2$), e-unguarded at $t+1$ and with $x_{K-1}x_{K}\in E_{M}^{D}\left(
t\right)  \cap E_{M}^{D}\left(  t+1\right)  $ (note also that, since
$vux_{1}...x_{K}$ is a path, $\left\{  u,v\right\}  \cap\left\{  x_{1}%
,x_{2},...,x_{K}\right\}  =\emptyset$). In this case, we have%
\[
\left.
\begin{array}
[c]{l}%
vux_{1}...x_{K}\text{ e-unguarded at }t+1\text{ }\\
u\rightarrow v\text{ at }t+1\\
u\notin\left\{  x_{1},x_{2},...,x_{K}\right\}
\end{array}
\right\}  \Rightarrow ux_{1}...x_{K}\text{ is e-unguarded at }t\text{;}%
\]
and%
\[
\left.
\begin{array}
[c]{l}%
ux_{1}...x_{K}\text{ e-unguarded at }t\\
x_{K-1}x_{K}\in E_{M}^{D}\left(  t\right)
\end{array}
\right\}  \Rightarrow x_{K-1}\in V_{M}^{D}\left(  t\right)  =V_{N}^{D}\left(
t\right)  ;
\]
and finally
\[
\left.
\begin{array}
[c]{r}%
\left.
\begin{array}
[c]{l}%
x_{K-1}\in V_{N}^{D}\left(  t\right)  \text{ }\\
u\rightarrow v\text{ at }t+1\\
v\notin\left\{  x_{1},x_{2},...,x_{K}\right\}
\end{array}
\right\}  \Rightarrow x_{K-1}\in V_{N}^{D}\left(  t+1\right) \\
ux_{1}...x_{K-1}\text{ is n-unguarded at }t+1
\end{array}
\right\}  \Rightarrow u\in V_{N}^{D}\left(  t+1\right)
\]
But $u\in V_{N}^{D}\left(  t+1\right)  $ contradicts assumption (\ref{eq0322}%
). Hence $u\notin V_{M}^{D}\left(  t+1\right)  .$ \medskip

\qquad In short, by examining cases (i.1)-(i.3)$\ $we have shown that%
\begin{equation}
uv\in E_{N}^{D}\left(  t\right)  \cap E_{N}^{C}\left(  t+1\right)  \Rightarrow
uv\in E_{M}^{C}\left(  t+1\right)  . \label{eq0304}%
\end{equation}
i.e., if at $t+1$ some $uv$ is n-cleared it is also m-cleared.

\item[\textbf{II.b}] \emph{If a previously m-clear edge becomes m-dirty at
}$t+1$\emph{, then it also becomes n-dirty.} Suppose that $uv\in E_{M}%
^{C}\left(  t\right)  \cap E_{M}^{D}\left(  t+1\right)  $. Without loss of
generality, we can assume that there exists a $x_{K-1}x_{K}\in E_{M}%
^{D}\left(  t\right)  \cap E_{M}^{D}\left(  t+1\right)  $ and a path
$uvx_{1}...x_{K}$ which was e-guarded at $t$ but became e-unguarded at $t+1$.
Since the path became e-unguarded at $t+1$, exactly one node in it was guarded
at $t$ (and became unguarded at $t+1$). Call this node $x_{\overline{k}}$,
with $\overline{k}=0$ if the node in question is $v$ (i.e., $x_{0}=v$)\ and
$\overline{k}\in\left\{  1,2,...,K-1\right\}  $ otherwise. The move at $t+1$
is $x_{\overline{k}}\rightarrow z$, where $z\notin\left\{  v,x_{1}%
,...,x_{K-1}\right\}  $ (since $uvx_{1}...x_{K}$ is e-unguarded at $t+1$).
Since $x_{K-1}x_{K}\in$ $E_{M}^{D}\left(  t\right)  =E_{N}^{D}\left(
t\right)  $, either $x_{K-1}\in V_{N}^{D}\left(  t\right)  $ or $x_{K}\in
V_{N}^{D}\left(  t\right)  $. We consider the following cases.

\textbf{(ii.1)} If $x_{K-1}\in V_{N}^{D}\left(  t\right)  $ then $x_{K-1}\in
V_{N}^{D}\left(  t+1\right)  $ as well (it was not entered) and the path
$vx_{1}...x_{K-1}$ is n-unguarded at $t+1$. So $v\in V_{N}^{D}\left(
t+1\right)  $ and $uv\in E_{N}^{D}\left(  t+1\right)  $.

\textbf{(ii.2)} If $x_{K-1}\notin V_{N}^{D}\left(  t\right)  $, then $x_{K}\in
V_{N}^{D}\left(  t\right)  $ and so $x_{K-1}$ is the guarded node in the path
$uvx_{1}...x_{K}$ (i.e., $\overline{k}=K-1$). We distinguish two subcases.

\textbf{(ii.2.1)} The move at $t+1$ is $x_{K-1}\rightarrow z$ and $z\neq
x_{K}$. Then $x_{K}\in V_{N}^{D}\left(  t+1\right)  $ and so $x_{K-1}\in
V_{N}^{D}\left(  t+1\right)  $ and, by the same argument as in (ii.1), we get
$uv\in E_{N}^{D}\left(  t+1\right)  $.

\textbf{(ii.2.2)} The move at $t+1$ is $x_{K-1}\rightarrow x_{K}$. Since we
have assumed $x_{K-1}x_{K}\in E_{M}^{D}\left(  t\right)  \cap E_{M}^{D}\left(
t+1\right)  $, $x_{K-1}x_{K}$ must be adjacent to some \emph{other} edge
$x_{K-1}y\in$ $E_{M}^{D}\left(  t\right)  \cap E_{M}^{D}\left(  t+1\right)  $;
so we can use $y$ in place of $x_{K}$, $uvx_{1}...x_{K-1}y$ in place of
$uvx_{1}...x_{K}$ and conclude by the reasoning of case (ii.2.1) that $uv\in
E_{N}^{D}\left(  t+1\right)  $. \medskip

\qquad In short, we have shown that%
\begin{equation}
uv\in E_{M}^{C}\left(  t\right)  \cap E_{M}^{D}\left(  t+1\right)  \Rightarrow
uv\in E_{N}^{D}\left(  t+1\right)  , \label{eq0305}%
\end{equation}
i.e., if at $t+1$ some $uv$ is m-dirtied then it is also n-dirtied.
\end{enumerate}

Hence every edge added at $t+1$ to $E_{N}^{C}\left(  t\right)  $ is also added
to $E_{M}^{C}\left(  t\right)  $ and every edge removed at $t+1$ from
$E_{M}^{C}\left(  t\right)  $ is also removed from $E_{N}^{C}\left(  t\right)
$. These facts, combined with $E_{N}^{C}\left(  t\right)  =E_{M}^{C}\left(
t\right)  $, yield%
\begin{equation}
E_{N}^{C}\left(  t+1\right)  \subseteq E_{M}^{C}\left(  t+1\right)  .
\label{eq0306}%
\end{equation}
\noindent\textbf{III.} To strengthen (\ref{eq0306}) to set equality, we need
the reverse set inclusion. In fact we can prove
\begin{equation}
\text{for }s=0,1,...,t,t+1,...:\qquad E_{M}^{C}\left(  s\right)  \subseteq
E_{N}^{C}\left(  s\right)  \label{eq0314}%
\end{equation}
by following the proof of Theorem \ref{prp0303} and replacing
\textquotedblleft e-clear\textquotedblright, \textquotedblleft
e-dirty\textquotedblright, ... with \textquotedblleft
m-clear\textquotedblright, \textquotedblleft m-dirty\textquotedblright, ...
(analogs of Lemmas \ref{prp0301}, \ref{prp0302} can also be proved
easily)\footnote{Basically the proof remains valid because the
\emph{recontamination rules} remain the same.}. From (\ref{eq0306})\ and
(\ref{eq0314})\ we obtain
\begin{equation}
E_{M}^{D}\left(  t+1\right)  =E_{N}^{D}\left(  t+1\right)  . \label{eq0315}%
\end{equation}
\noindent\textbf{IV.} Now suppose \ $u\in V_{M}^{D}\left(  t+1\right)  $. Then
$u$ is unguarded and adjacent to some $uv\in E_{M}^{D}\left(  t+1\right)
=E_{N}^{D}\left(  t+1\right)  $. Hence either $u$ or $v$ is n-dirty and so
(since $u$ is unguarded) $u\in V_{N}^{D}\left(  t+1\right)  $. On the other
hand, if $u\in V_{N}^{D}\left(  t+1\right)  $, then $u$ is unguarded and there
exists a $uv\in E_{N}^{D}\left(  t+1\right)  =E_{M}^{D}\left(  t+1\right)  $
and (since $u$ is unguarded) $u\in V_{M}^{D}\left(  t+1\right)  $. It follows
that
\[
V_{M}^{D}\left(  t+1\right)  =V_{N}^{D}\left(  t+1\right)  .
\]
\noindent\textbf{V.} In short, we have established that%
\begin{equation}
\left(
\begin{array}
[c]{c}%
E_{M}^{C}\left(  t\right)  =E_{N}^{C}\left(  t\right) \\
V_{M}^{C}\left(  t\right)  =V_{N}^{C}\left(  t\right)
\end{array}
\right)  \Rightarrow\left(
\begin{array}
[c]{c}%
E_{M}^{C}\left(  t+1\right)  =E_{N}^{C}\left(  t+1\right) \\
V_{M}^{C}\left(  t+1\right)  =V_{N}^{C}\left(  t+1\right)
\end{array}
\right)  . \label{eq0323}%
\end{equation}
Eq.(\ref{eq0323}) and the fact $V_{M}^{C}\left(  0\right)  =V_{N}^{C}\left(
0\right)  $ and $E_{M}^{C}\left(  0\right)  =E_{N}^{C}\left(  0\right)  $
complete the proof of the theorem.
\end{proof}

Several facts follow from Theorem \ref{prpA01}. First, it is well known that
mixed edge search is NP-complete \cite{Ellis} hence node search is also NP-complete.

Second, Theorem \ref{prpA01} can be used to obtain a Theorem \ref{prp0303} as
a corollary with a short proof (which, of course, presupposes the lengthy
proof of Theorem \ref{prpA01}). We give a sketch of such a proof. Suppose a
graph $\mathbf{G}$ and a search schedule $\mathbf{S}$ are given. Then, by
Theorem \ref{prpA01} we have (for every $t$) $V_{N}^{C}\left(  t\right)
=V_{M}^{C}\left(  t\right)  $ and $E_{N}^{C}\left(  t\right)  =E_{M}%
^{C}\left(  t\right)  $. Now suppose $\mathbf{S}$ is applied to $\mathbf{G}$
under the rules of the edge game and, up to time $t$ we have $V_{E}^{C}\left(
t\right)  \subseteq V_{M}^{C}\left(  t\right)  $ and $E_{E}^{C}\left(
t\right)  \subseteq E_{M}^{C}\left(  t\right)  $. Then, at time $t+1$, every
e-cleared edge is also m-cleared (since the mixed game has all the clearing
rules of the edge game and an additional one)\ and every m-dirtied edge is
also e-dirtied (since the edge and mixed games have the same recontamination
rules) except if an edge was already e-dirty. Hence $E_{E}^{C}\left(
t+1\right)  \subseteq E_{M}^{C}\left(  t+1\right)  $, which also shows that
$V_{E}^{C}\left(  t+1\right)  \subseteq V_{M}^{C}\left(  t+1\right)  $.

Finally, the above Theorem \ref{prpA01} and the monotonicity result proved in
\cite{Bienstock} (\textquotedblleft if there is a mixed edge clearing search
of $\mathbf{G}$ using $\leq K$ guards, there is a monotone mixed edge clearing
\ search of $\mathbf{G}$ using $\leq K$ guards\textquotedblright) makes the
following conjecture seem almost trivially true.

\begin{conjecture}
\label{prpA02}If there is a node clearing search $\mathbf{S}$ of $\mathbf{G}$
using at most $K$ guards, there is a monotone node clearing search
$\mathbf{S}^{\prime}$ of $\mathbf{G}$ using at most $K$ guards.
\end{conjecture}

However some additional work is required to prove the conjecture, because of
the following detail:\ edge monotonicity does not necessarily imply node
monotonicity. In other words, the implication
\begin{equation}
E_{N}^{C}\left(  t\right)  \subseteq E_{N}^{C}\left(  t+1\right)  \Rightarrow
V_{N}^{C}\left(  t\right)  \subseteq V_{N}^{C}\left(  t+1\right)
\label{eq0324}%
\end{equation}
is not necessarily true! Consider again the search schedule of Remark
\ref{prp0202}. In this case $E_{N}^{C}\left(  1\right)  =E_{N}^{C}\left(
2\right)  =\emptyset$, but $V_{N}^{C}\left(  1\right)  \subseteq V_{N}%
^{C}\left(  2\right)  $ does not hold. It is true however (and easy to
prove)\ that
\begin{equation}
\text{(\emph{strong }edge monotonicity)}\Rightarrow\text{(node monotonicity)}.
\label{eq0325}%
\end{equation}

Unfortunately, the result proved by Bienstock in \cite{Bienstock} concerns
\textquotedblleft simple\textquotedblright, not strong monotonicity. Hence one
of our future research goals is to extend Bienstock's result to node search
and thus prove that $s_{N}\left(  \mathbf{G}\right)  =s_{N}^{m}\left(
\mathbf{G}\right)  $.\newpage

\section{Appendix:\ Vertex separation and pathwidth}

\label{secC}

Vertex separation and pathwidth are strongly (but not obviously) related to
graph search and search number, as discussed in (among other
papers)\ \cite{Megiddo,FominThilikos}, where the author can find a more
detailed discussion; here we just give the basic definitions and the main theorem.

\begin{definition}
\label{def0204}Given a graph $\mathbf{G}=\left(  V,E\right)  $, a \emph{path
decomposition} of $\mathbf{G}$ is a pair $(X,\mathbf{P})$, where
$X=\{X_{1},...,X_{M}\}$ is a family of subsets of $V$ and $\mathbf{P}$ is a
path whose nodes are the subsets $X_{i}$ and they satisfy the following properties

\begin{enumerate}
\item $\cup_{m=1}^{M}X_{i}=V$.

\item For every edge $vw$ in $E$, there is a subset $X_{m}$ that contains both
$v$ and $w$.

\item If $1\leq i\leq j\leq k\leq M$ then $X_{i}\cap X_{k}\subseteq X_{j}$.
\end{enumerate}
\end{definition}

\begin{definition}
\label{def0205}The \emph{width} of a path decomposition $(X,\mathbf{P})$ (with
$X=\{X_{1},...,X_{n}\}$) is denoted by $pdw\left(  X,\mathbf{P}\right)  $ and
is defined by
\[
pdw\left(  X,\mathbf{P}\right)  =\max_{1\leq i\leq n}\left\vert X_{i}%
\right\vert .
\]

\end{definition}

\begin{definition}
\label{def0206}The \emph{pathwidth} of a graph $\mathbf{G}=\left(  V,E\right)
$ is denoted by $pw\left(  \mathbf{G}\right)  $ and defined by
\[
pw\left(  \mathbf{G}\right)  =\min_{\text{all path decompositions }\left(
X,\mathbf{P}\right)  \text{ of }\mathbf{G}\text{ }}pdw\left(  X,\mathbf{P}%
\right)  .
\]

\end{definition}

\begin{definition}
\label{def0207}Given a graph $\mathbf{G}=\left(  V,E\right)  $ (with
$\left\vert V\right\vert =N$)\ and a permutation $\mathcal{P}$ of $\left\{
1,2,...,N\right\}  $. The \emph{vertex separation }of $\mathbf{G}$ with
respect to $\mathcal{P}$ is denoted by $vs\left(  \mathbf{G},\mathcal{P}%
\right)  $ and defined by%
\[
vs\left(  \mathbf{G},\mathcal{P}\right)  =\max_{1\leq n\leq N}\left\vert
\left\{  u\in V:\mathcal{P}\left(  u\right)  \leq n\text{ and }\exists v\text{
such that: }\mathcal{P}\left(  v\right)  >n\text{ and }uv\in E\right\}
\right\vert .
\]
The \emph{vertex separation }of $G$ is denoted by $vs\left(  \mathbf{G}%
\right)  $ and defined by%
\[
vs\left(  \mathbf{G}\right)  =\min_{\text{all permutations }\mathcal{P}%
}vs\left(  \mathbf{G},\mathcal{P}\right)  .
\]

\end{definition}

The following theorem shows the connection between $vs\left(  \mathbf{G}%
\right)  $ and $pw\left(  \mathbf{G}\right)  $ and search number.

\begin{theorem}
For every graph $\mathbf{G}$, $vs\left(  \mathbf{G}\right)  =pw\left(
\mathbf{G}\right)  $ and $\left\vert vs\left(  \mathbf{G}\right)
-s_{E}\left(  \mathbf{G}\right)  \right\vert \leq1$.
\end{theorem}

\newpage

\section{Appendix: Implementations of the GSearch Algorithm}

\label{secD}

We have implemented the GSST\ algorithm in two forms which are publicly available.

\begin{enumerate}
\item As a command line executable, which runs on Windows and Linux computers.
The program and supporting material are available at the URL\ 

\qquad\qquad\texttt{http://www.frc.ri.cmu.edu/\symbol{126}%
gholling/home/software.html}.

\item As a graphical user interface (GUI), available at

\qquad\qquad\qquad\texttt{http://users.auth.gr/\symbol{126}%
kehagiat/KehagiasSoftware.htm}.
\end{enumerate}

We discuss each of these implementations separately.

\subsection{The Command Line program}

This is an executable, \texttt{gsearch.exe,} which has the following usage

\bigskip
\begin{verbatim}
       USAGE: gsearch -m [graph] -n [no.trees] -s [startnode] -g [gen.tree] 
          -y [wr-search] -w [wr-tree] -v [visualize] -t [edge traversal]
          -i [improve tree] -l [low number] -r [redundancy check]
       EXAMPLE: gsearch -m graphs/Edge01.txt -g exhaustive -s 1 -n 500000 -t bh
\end{verbatim}

\bigskip

The various options available are as follows.

\bigskip

\texttt{-m [graph]: string, name of file with edge list of the graph.}

\texttt{-n [no.trees]: int, how many sp.trees to generate (DEFAULT is 1).}

\texttt{-s [start node]: int, which node to start }

\ \ \ \ \ \ \texttt{(DEFAULT is 1, random choice is 0)}

\texttt{-g [gen.tree]: string, method of generating spanning trees}

\texttt{ \ \ (acceptable values: readtree, exhaustive, uniform, dfsrand; }

\ \ \ \ \ \ \ \texttt{DEFAULT is uniform)}

\texttt{-y [wr-search]: string, how many best searches to write}

\texttt{ \ \ (acceptable values: one, all, none; DEFAULT is one)}

\texttt{-w [wr-tree]: string, how many best sp.trees to write}

\texttt{ \ \ (acceptable values: one, all, none; DEFAULT is one)}

\texttt{-t [traversal]: string, how to break edge label ties}

\texttt{ \ \ (acceptable values: bh, random, bhrand, bhweight, bhdom; }

\ \ \ \ \ \ \ \texttt{DEFAULT is bh).}

\texttt{-i [improve tree]: boolean, use tree improve technique or not.}

\texttt{ \ \ (acceptable values are 0 / 1; DEFAULT is 0: do not use it)}

\texttt{-r [redundancy check]: boolean, check for redundant trees}

\texttt{ \ \ (acceptable values are 0 / 1; DEFAULT 1: }

\ \ \ \ \ \ \ \texttt{check for redundancy)}

\texttt{-l [low number]: break if a tree is found with this number of }

\ \ \ \ \ \ \ \texttt{searchers (DEFAULT is 0: do not break)}

\texttt{-v [visualize]: boolean, use visualizer (only supported }

\ \ \ \ \ \ \ \texttt{on linux) (DEFAULT is 0: do not use)}

\texttt{-h [help] }

\bigskip

The \texttt{-m} option indicates the file which contains the graph
description. This must be an ascii (plain text)\ file containing a list of the
edges of the graph, one edge per line, indicated as a pair of nodes. The nodes
must be continuously numbered from 1 to $N$ and these numbers are used as
labels. $N$ is assumed to be the largest number appearing in the edge list
(the graph is assumed to be undirected and connected). Examples of edge lists
can be found in the \texttt{graphs} directory.

The remaining options of \texttt{gsearch} correspond to the description of the
algorithm in Section \ref{sec05}. The \texttt{-n} option corresponds to the
$M$ parameter of GSST (number of spanning trees), \texttt{-s} is the root of
the search, \texttt{-g} describes the uniform and DFS\ methods for generating
spanning trees (there are also options for exhaustive enumeration of all
spanning trees and for reading a specific spanning tree from file) and the
\texttt{-t} options (\texttt{bh, random, bhrand, bhweight, bhdom}) corresponds
to the traversal methods (L, R, LR, LW and LD, respectively). The \texttt{-i,
-r }and \texttt{-l} options are self explanatory. The \ \texttt{-y} option
writes one or more minimal strategies in the file(s)
\texttt{output/strat*.txt}; each row of the file corresponds to one step of
the strategy and shows the nodes in which the searchers are currently located.
The \ \texttt{-w} option writes the rooted tree(s) corresponding to optimal
strategies, in the file(s) \texttt{output/tree*.txt}; each row of the file
shows one parent and her child. The \texttt{-v} visualization option only
works on Linux computers (for a visualization of the search on Windows
computers use the GUI).

The command line \ \texttt{gsearch} was implemented by G. Hollinger in ANSI\ C
using the gcc 4.1 compiler.

\subsection{The Graphical I\textbf{nterface}}

This is an executable, \texttt{gsearchGUI.exe,} which corresponds closely to
the command line program. (In fact the GUI\ is a front end for the command
line program.)\ Launching the program brings up the window of Figure
\ref{fig15}. On the left side we see several input boxes. All of these
correspond to abovementioned options of the command line program with one
exception: in addition to the edge list file (\texttt{graphs/Edge11.txt} in
Figure \ref{fig15}) there is also a file containing the $x$- and
$y$-coordinates of the nodes (\texttt{graphs/Node11.txt} in Figure \ref{fig15}
-- this is required for visualization). The GUI\ starts with the default
selection of \texttt{Edge01.txt} and \texttt{Node01.txt} ; if you want a
different graph, type the corresponding file names in the input boxes.

\begin{figure}[h]
\centering\scalebox{0.5}{\includegraphics{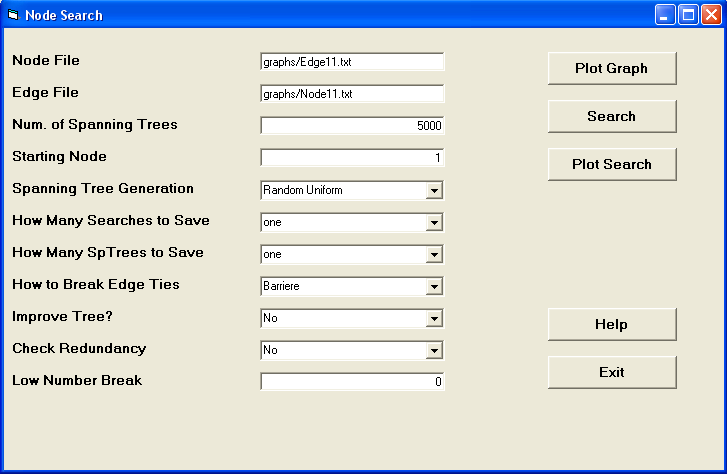}}\caption{A
screenshot of the gsearch GUI.}%
\label{fig15}%
\end{figure}

Having completed the input boxes on the left side of the window, you can now
use the buttons on the right side.

\begin{enumerate}
\item The \textsf{Search} button performs the graph search (it actually
invokes the \texttt{gsearch} executable and stores results in the
\texttt{output} directory). Time to complete the search depends on the size of
the graph and the number of spanning trees used. Do \emph{not} use the
exhaustive spanning tree enumeration option on large graphs because the
program may take too long to terminate.

\item The \textsf{Plot Graph} button plots the graph. An image of the graph is
generated and an image viewer is invoked which can be used to view the image
file (this is always named \texttt{graph.jpg}). A screenshot of the image
viewer appears in Figure \ref{fig16}.

\item The \textsf{Plot Search} button plots the graph search. Actually, first
the search strategy files are used to generate a sequence of image files, one
image corresponding to each step of the search (these files are stored in the
directory \texttt{pix}) and then the image viewer is invoked to view the
files. The user can step through the images using the arrow keys, or use the
\textsf{start} / \textsf{stop }buttons to run a slide show of the search. A
screenshot of the image viewer appears in Figure \ref{fig16}.
\end{enumerate}

\begin{figure}[h]
\centering\scalebox{0.5}{\includegraphics{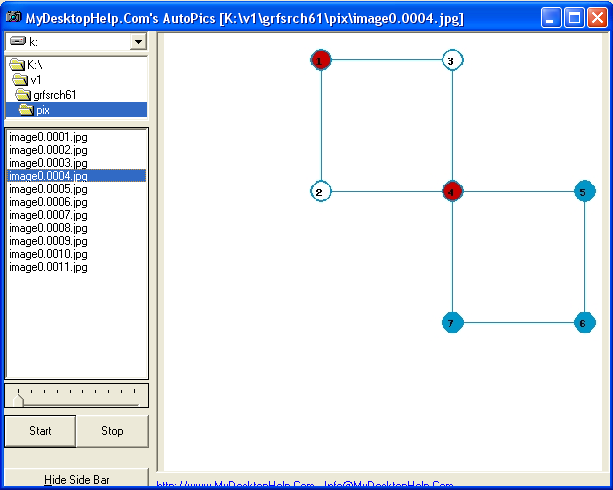}}\caption{A
screenshot of the image viewer.}%
\label{fig16}%
\end{figure}

\noindent \emph{Nota bene}: do not run \textsf{Plot Search} before running
\textsf{Search}! The \textsf{Plot Search} function uses the latest results of
\textsf{Search}; if these correspond to an earlier graph (from a previous run
of the GUI) the plotted search will produce nonsense reults.

The GUI was implemented by Ath. Kehagias in MS Visual Basic 5.0 (and a little
bit of C, using the dev-cpp 4.9.9.2 environment). The image viewer is a
freeware program called AutoPics and made available by the company
\emph{Mydesktophelp} (at \texttt{http://www.mydesktophelp.com}). The
GD\ graphics library (version gd-2.0.34-win32) has also been used; this
library is available at \texttt{http://www.libgd.org/}.

Finally it must be stressed that the current version , \texttt{gsearchGUI
v.0.9,} is still \emph{beta} and requires further development. However we make
it publicly available in the hope that it will prove useful to the graph
search community.

\end{document}